\documentclass{article}
\let\originalvec\vec
\usepackage{woo}
\let\vec\originalvec

\usepackage{times}

\theoremstyle{plain}
\newtheorem{proposition}[theorem]{Proposition}
\crefname{proposition}{proposition}{propositions}
\Crefname{proposition}{Proposition}{Propositions}

\theoremstyle{remark}

\usepackage{titlesec}
\crefname{assumption}{assumption}{assumptions}
\Crefname{assumption}{Assumption}{Assumptions}
\crefname{cor}{corollary}{corollaries}
\Crefname{cor}{Corollary}{Corollaries}
\titleformat{\part}{\normalfont\Large\bfseries}{\thepart}{1em}{}

\newcommand{\bsg}{{\boldsymbol{g}}}

\newcommand{\sgn}{\mathsf{sgn}}
\newcommand{\fdec}{\vec{f}_\mathrm{reconstruct}}

\newcommand{\PM}{\mathop{\mathrm{PM}}}

\newcommand{\dd}{\,\mathrm{d}}
\newcommand{\BKM}{\mathrm{BKM}}
\renewcommand{\Re}{\operatorname{Re}}

\newcommand{\Cov}{\operatorname{Cov}}

\NewDocumentCommand{\NewSavedEnvironment}{m m}{
    \NewDocumentEnvironment{#1}{m +b}{\begin{restatable}{donothing}{##1}
        \begin{#2}
        ##2
        \end{#2}
        \end{restatable}}{\unskip\ignorespacesafterend}
    }

\NewSavedEnvironment{savedequation}{equation}
\NewSavedEnvironment{savedalign}{align}
\newcommand{\bigsymdiff}{\mathop{\scalerel*{\bigtriangleup}{\bigoplus}}}

\newcommand{\calI}{\mc I_4}
\newif\ifreviewnotes
\reviewnotestrue

\newcounter{reviewcomment}
\crefname{reviewcomment}{review comment}{review comments}
\Crefname{reviewcomment}{Review comment}{Review comments}

\title{Learning SYK Hamiltonians}
\author{}

\date{September 2026}

\AtBeginBibliography{\sloppy}

\author{Anurag Anshu\thanks{Harvard. \texttt{anuraganshu@fas.harvard.edu}} \and Srinivasan Arunachalam\thanks{IBM Research, \texttt{srinivasan.arunachalam@ibm.com}.} \and
Sitan Chen\thanks{Harvard. \texttt{sitan@seas.harvard.edu}}
\and 
Yeongwoo Hwang\thanks{Harvard. \texttt{yeongwoohwang@g.harvard.edu}}
}

\date{}

\begin{document}

\maketitle
\begin{abstract}
	We study the problem of learning the dense Sachdev--Ye--Kitaev (SYK) Hamiltonian from copies
	of its Gibbs state. Existing algorithms for Hamiltonian learning typically rely on geometric
	locality or bounded interaction degree and therefore do not apply to SYK, where each quartic
	interaction overlaps with \(\Theta(n^3)\) others. We show that this obstruction can be
	overcome by exploiting the random mean-field structure of the model. \emph{At any constant
	temperature}, we prove that with high probability over the SYK couplings, the entire
	Hamiltonian can be learned to inverse-polynomial accuracy using polynomially many samples.
	Furthermore, when the inverse temperature is restricted to be a sufficiently
	small constant, we construct a quasipolynomial-time learning algorithm which is
	qualitatively different from the sample-efficient algorithm.

\end{abstract}

\newpage 
\setcounter{tocdepth}{2}
{\footnotesize \tableofcontents}

\newpage

\section{Introduction}

Understanding the complexity of quantum states has been a central focus of quantum computer science, with recent years bringing the new measure of \textit{learnability} to the picture. Rooted in the practical design of quantum computers as well as philosophical questions about the scope of knowledge that quantum mechanics allows, a wide range of quantum objects have been shown to be learnable in recent years (see~\cite{AnshuArunachalam24} for a survey). A notable class is that of \emph{quantum Gibbs states}, which form an expressive family of states and have been a subject of intense study in physics for over a century.   The problem of \textit{Hamiltonian learning} from Gibbs states is as follows. Can one reconstruct the Hamiltonian governing a quantum many-body system, described~as
$$
    H(\vec g)=\sum_{\alpha=1}^{m} g_\alpha E_\alpha
$$
for some known fixed basis $\{E_\alpha\}_\alpha$, given copies of the Gibbs state:
$$
    \gamma_{\vec g}
      =\frac{e^{-\beta H(\vec g)}}
             {\operatorname{tr}_{\rm ordinary}(e^{-\beta H(\vec g)})} \quad ?
$$
The goal is to recover the couplings $\vec g$ with prescribed error using as few copies of
$\gamma_{\vec g}$ and as little running time as possible. This is the quantum analogue of learning
an undirected \emph{graphical model} from samples.
Besides being a fundamental inverse
problem, Hamiltonian learning has promising applications in the calibration, verification of quantum simulators and the extraction of effective microscopic models from
thermal data.

One compelling setting for this inverse problem is provided by strongly interacting quantum many-body systems, where the  interactions in the Hamiltonian  
encode highly nontrivial correlations even though the description of the Hamiltonian is
specified by only polynomially many coupling parameters. Among the most prominent models of strongly interacting Hamiltonians is the well-known \emph{Sachdev--Ye--Kitaev} (SYK) model~\cite{SachdevYe93,Kitaev15,MaldacenaStanford16,MaldacenaShenkerStanford16, KKZ19}, which has become central to quantum chaos and strongly interacting fermions, as well as a benchmark
for the limits of classical simulation. In this work,  we initiate  the Hamiltonian learning problem for SYK models by asking
\begin{quote}
\begin{center}
\emph{Can a dense SYK Hamiltonian be learned in polynomial time,\\ given copies of its Gibbs state at constant temperature?}
\end{center}
\end{quote}
Two conceptual reasons motivate this question. 
First, a sequence of works initiated by~\cite{anshu2020sample} showed that geometrically local Hamiltonians can be learned with polynomial sample complexity. Subsequently~\cite{haah2022optimal}  obtained essentially optimal sample and time complexity in the high-temperature regime for a broader class of Hamiltonians. More recently~\cite{BLMT2024_learning,Narayanan25} gave a polynomial-time algorithm at every constant temperature, and then~\cite{CAN25} obtained further local-learning guarantees. Central to these results is some form of locality or bounded interaction degree.
The SYK model is the canonical testbed of whether this obstruction is fundamental. Formally, the SYK Hamiltonian is defined as follows: for even $q$
$$
    H_{\mathrm{SYK}}^{(q)}(\vec g)
      =\sigma_{n,q}\sum_{\substack{S\subseteq[n]\\ |S|=q}}
          g_S\,\Gamma_S^{(q)},\qquad
    \sigma_{n,q}=\sqrt{(q-1)!/n^{q-1}},\qquad
    \Gamma_S^{(q)}=-i^{q/2}\prod_{j\in S}\psi_j ,
$$
where the $\psi_i$ are Majorana operators satisfying $\{\psi_i,\psi_j\}=2\delta_{ij}\Id$ and
$\vec g$ is a realization of $\vec \bsg=(\bsg_S)_S$ with independent $\mc N(0,1)$ coordinates.
The overall phase is chosen so that at $q=4$,
$\Gamma_S^{(4)}=\Gamma_S=\prod_{j\in S}\psi_j$ and
$\sigma_{n,4}=\sigma_n=\sqrt{6/n^3}$. Each summand is $q$-local, but the interaction pattern is all-to-all.
The model combines several features that make the learning problem nontrivial: its terms strongly fail to commute, its Gibbs states are
genuinely non-Gaussian, and its physically relevant behavior occurs at constant inverse temperature. For quartic SYK,  every interaction overlaps
with $\Theta(n^3)$ other interactions. Consequently, tools such as Lieb-Robinson bounds used in \cite{anshu2020sample, BLMT2024_learning, CAN25} do not apply. The high-temperature condition in \cite{haah2022optimal} also depends on the degree, forcing
$\beta= n^{-\Omega(1)}$ which does not cover the  case of fixed inverse temperature $\beta>0$.

Our second motivation is to test whether constructive learning results can be obtained for models that exhibit strong chaotic behavior. The dynamics generated by the SYK Hamiltonian is chaotic, and hence one may expect that it may exhibit pseudo-random behaviour. Learnability of the SYK Hamiltonian would offer a complementary viewpoint and possibly reveal new structures in the model.

\subsection{Main result}
Our first result shows that the dense interaction structure of SYK does not prevent sample-efficient Hamiltonian learning. For quartic SYK, at every fixed inverse temperature \(\beta>0\), the complete coupling vector can be recovered to inverse-polynomial accuracy from
\[
    \operatorname{poly}\!\left(n,1/\varepsilon,
        \log(1/\zeta)\right)
\]
many copies of the Gibbs state, with high probability over the random couplings and conditional success probability at least $1-\zeta$ over the measurement outcomes.
In particular, the sample complexity is polynomial in the parameters, despite the polynomial interaction degree. To our
knowledge, this is the first constant-temperature Hamiltonian-learning
guarantee for a dense, random, noncommuting mean-field~ensemble.

\paragraph{Sample-efficient learning.}
Our first theorem gives polynomial sample complexity at every fixed positive temperature. For technical convenience, throughout the paper we work with the normalized trace
\[
    \Tr(X):=\tr(X)/d_n,
    \qquad
    d_n:=2^{n/2},
\]
where \(d_n\) is the dimension of the Hilbert space of \(n\) Majorana
fermions.\footnote{Here, observe that every pair of Majorana modes forms one complex Fermionic mode, whose occupation number
is either \(0\) or \(1\), so \(n\) Majorana modes correspond to \(n/2\)
Fermionic modes and hence to a \(2^{n/2}\)-dimensional Hilbert space. Hence we assume $n$ is even throughout the paper.} Thus \(\Tr(I)=1\).

\begin{restatable}[Sample-efficient learning at every fixed positive
	temperature]{theorem}{samplealltemp}
\label{sampleall:thm:learning}
Fix a known $\beta>0$, $\varepsilon\in(0,1)$, and $\zeta\in(0,1/2)$. Then, for all sufficiently large $n$ (depending on $\beta$, $\varepsilon$, and $\zeta$), there is a procedure using
\begin{equation}\label{sampleall:eq:samplecomplexity}
 n^{O(1+\beta^2)}\varepsilon^{-2}\log(n/\zeta)
\end{equation}
copies of the Gibbs state $d_n^{-1}\rho_\beta(\vec \bsg)$ with the following guarantee: with probability
at least $1-O_\beta(1/n)$ over the random disorder $\vec \bsg$, the procedure outputs an estimate
$\vec g_\mathrm{meas}$ satisfying
\[
    \|\vec g_\mathrm{meas}-\vec \bsg\|_2\leq \varepsilon
\]
with probability $\geq 1-\zeta$ over the measurement outcomes.
\end{restatable}

\begin{remark}[Even interaction order]
\label{rem:sample_even_q}
The conclusion of \Cref{sampleall:thm:learning} extends to $H_{\mathrm{SYK}}^{(q)}$
for every fixed even $q\geq2$ and known $\beta>0$. For sufficiently large even
$n\geq q$ depending on $\beta,q$, the sample bound becomes
\[
 \begin{aligned}
 r_{\beta,q}&=2q+2eq(2q+1)\beta(1+\beta),\\
 N&\leq C_{\beta,q}\left(n^{4r_{\beta,q}+q+1}
       +n^{2r_{\beta,q}+q+1}\varepsilon^{-2}\right)\log(n/\zeta),
 \end{aligned}
\]
with disorder success probability $1-O_{\beta,q}(n^{-1})$ and the same conditional
measurement guarantee and ranges of $\varepsilon,\zeta$.

Only the following changes to the proof are needed. Set $m=\binom nq$ and use
$\sigma_{n,q}^2=(q-1)!/n^{q-1}$. The coefficient filter in
\eqref{sampleall:eq:filter} averages over the $2^q$ subsets of a $q$-element support.
Using $(q-1)$-element sets in the anticommutation argument of
\Cref{alltemp:lem:elementary} gives
$\sigma_{n,q}\|\sum_Iu_I\Gamma_I^{(q)}\|_{\rm op}\leq\|\vec u\|_2/\sqrt q$;
thus the Hessian Lipschitz constant becomes $6\beta^3/q^{3/2}$.
In \Cref{alltemp:lem:moments}, a nonzero commutator increases the support by at most
$q-2$, so at most $qk\binom{n-1}{q-1}$ labels are available at each of the first $k$
steps. Since $\sigma_{n,q}^2\binom{n-1}{q-1}\leq1$, this yields
\[
 M_{2k}\leq[2qk(1+\beta)]^{2k},\qquad
 s_{\beta,q}=\frac1{2eq(1+\beta)}.
\]
Taking $p=1+s_{\beta,q}/\beta$, the same concentration, fractional-moment, and BKM
arguments give $\nabla^2\log Z\succeq c_{\beta,q}n^{-r_{\beta,q}}I$.
Finally, $m^2\sigma_{n,q}^2=O_q(n^{q+1})$ gives the stated sample bound.
The constants $C_{\beta,q},c_{\beta,q}>0$ depend only on $\beta,q$; at $q=4$ this
recovers the theorem above.
\end{remark}

\paragraph{Quasi-polynomial time learning.}
Our second result gives a quasipolynomial-time learning algorithm at
sufficiently high constant temperature and achieves small $\ell_2$ error in the couplings. 
For the running-time statement, we use the standard qubit representation of the Majoranas and  that $\beta$ is specified to
polynomially many bits of precision to the algorithm.

\begin{restatable}[Quasi-polynomial time learning]{theorem}{timelearning}
\label{thm:time_learning}
Fix a known  constant $0 < \beta \leq \beta_\mathrm{algorithm}$, a constant $0 < c < 1$, and
$\zeta \in (0,1)$. For all sufficiently large $n$, there is an algorithm which uses
\[
	N = \O_\beta\!\Paren{n^{13}\log(n/\zeta)}
\]
copies of the Gibbs state $d_n^{-1} \rho_\beta(\vec \bsg)$ with the following guarantee: with probability
at least $1-\O_\beta(n^{-2(1-c)})$ over the random disorder $\vec \bsg$, the algorithm outputs
$\vec g_\mathrm{meas}$ satisfying
\[
	\frac{1}{\sqrt{\binom{n}{4}}} \|\vec g_\mathrm{meas} - \vec \bsg\|_2 \leq n^{-c},
\]
with probability $\geq 1-\zeta$ over the measurements. The runtime is
\[
	n^{O(\log n)} + \widetilde O_\beta\!\left(n^{14}\log(n/\zeta)\right)\,.
\]
\end{restatable}

\paragraph{Temperature dependence.}
The sample-efficient guarantee holds for every fixed $\beta>0$. For the time-efficient
algorithm, fix an absolute constant $C_\star\geq1$ sufficiently large to dominate the
absolute majorant constants in the background and observable-diagram estimates of
\Cref{sec:degree_proof}, and set
\[
 \beta_\mathrm{algorithm}=\min\left\{
   \frac{1}{768MC_\star},
	\frac{\beta_\mathrm{thermal}}{8\sqrt{15}}\right\}\,,\qquad
 M=\max\{1,c^{(\mathrm{deg})}_1,c^{(\mathrm{deg})}_2\}.
\]

The first restriction comes from multiplying two factors. The first factor $\frac 1 {64C_\star}$ will
give us global control over convergence of various series that will arise when counting polymers via
Wick expansion towards the proof of \Cref{lem:degree_bound}. For the second factor, $c^{(\mathrm{deg})}_1$ and
$c^{(\mathrm{deg})}_2$ relate to the degree bounds of our estimator and the bound $1/(12M)$
will yield stability of our low-degree calibration subroutine. The second restriction $\beta_{\rm
thermal}/(8\sqrt{15})$ supplies the threshold for our improved zero-free result and bounds the error
for the truncation of our proxy for $\Tr[\Gamma_I \rho_\beta(\vec g)]$.

\subsection{Proof overview of sample-efficient learning}

We first give an overview of the algorithm establishing sample-efficient learnability of SYK
Hamiltonians. The overall strategy follows the maximum-entropy framework of~\cite{anshu2020sample}: learning the Hamiltonian from Gibbs states reduces to proving that the log-partition function is \emph{strongly convex}. The main obstacle is that the SYK Hamiltonian has all-to-all connectivity. Every quartic term intersects $\Theta(n^3)$ other terms, so arguments that lose polynomially in the interaction degree do not yield a nontrivial statement at constant temperature. Instead, we exploit  two  features
of SYK: the random mean-field structure of its couplings and the algebraic structure of the Majorana operators. 

\subsubsection{Maximum entropy and strong convexity}
\label{sec:maximumentropy}
Let $\calI=\binom{[n]}4$, $m=|\calI|$, and recall that
\[
 H(\vec \bsg)=\sigma_n\sum_{I\in\calI}\bsg_I\Gamma_I,
 \qquad
 \sigma_n=\sqrt{6/n^3},
 \qquad
 \bsg_I\stackrel{\mathrm{i.i.d.}}{\sim}\mc N(0,1),
 \qquad
 \rho_\beta(\vec \bsg)=\frac{e^{-\beta H(\vec \bsg)}}{Z(\beta,\vec \bsg)}.
\]
We begin by recalling the maximum-entropy reduction of \cite{anshu2020sample}. First, one can observe that the quartic expectation values
\[
    \mu_{\beta,I}(\vec \bsg)
    :=
    \Tr\!\left[\rho_\beta(\vec \bsg)\Gamma_I\right],
    \qquad I\in\mathcal I_4,
\]
form a \emph{sufficient statistic} for the Gibbs family, i.e., the collection
$\{\mu_{\beta,I}(\vec \bsg)\}_{I\in\mathcal I_4}$ \emph{uniquely} determines
the  coupling vector $\vec \bsg$.  Consequently, if the exact expectation vector
$\vec \mu_\beta=(\mu_{\beta,I}(\vec \bsg))_{I\in\calI}$ were known exactly, then $\vec h = \vec \bsg$ would be the unique minimizer of the optimization program
\[
 \mc L(\vec h)
 =
 \Phi(\vec h)+\beta\sigma_n\,\vec h\cdot\vec \mu_\beta,
 \qquad
 \Phi(\vec h)=\log Z(\beta,\vec h).
\]
Indeed, $\nabla\Phi(\vec \bsg) = -\beta\sigma_n \vec \mu_\beta$, and hence $\nabla \mc
L(\vec \bsg)=0$. Of course, we only have empirical estimates
$\vec \mu_{\beta,\mathrm{meas}}$. We therefore minimize
\[
 \mc L_\mathrm{meas}(\vec h)
 =
 \Phi(\vec h)+\beta\sigma_n\,\vec h\cdot\vec\mu_{\beta,\mathrm{meas}}
\]
over a known compact convex set containing $\vec \bsg$.  Thus the optimization problem above reduces to a stability question: how far
can the minimizer move when the linear term in the objective is 
perturbed? Strong convexity of the function $\Phi$ around the true parameter
$\vec \bsg$ gives precisely the required control to convert an
$\ell_2$ error in the measured expectations into an $\ell_2$ error in
the recovered couplings. Suppose
\begin{align}
\label{eq:strongconvexity}
 \nabla^2\Phi(\vec \bsg)\succeq \lambda_n I
\end{align}
for some $\lambda_n>0$. Then $\mc L_{\mathrm{meas}}$ is stable under perturbations of the measured
expectation values, with a bound depending on $\lambda_n$. Ideally, one would like to prove a uniform lower bound on
$\nabla^2\Phi(\vec g)$ throughout the feasible region. In our setting,
however, we are only able to establish such a bound at the random target
$\vec \bsg$ itself. We show that this is nevertheless sufficient. In
particular, \Cref{alltemp:lem:hessian-lipschitz} shows that the Hessian is
Lipschitz:
\[
	\|\nabla^2\Phi(\vec h)-\nabla^2\Phi(\vec h')\|_{\rm op} \leq \poly(\beta)\cdot \|\vec h-\vec h'\|_2,
\]
Consequently, if Eq.~\eqref{eq:strongconvexity} holds, then  $\Phi$ remains $(\lambda_n/2)$-strongly convex in a ball of radius
$\lambda_n/\poly(\beta)$ around $\vec \bsg$. A standard boundary argument then can be used to show that if
\[
 \|\vec\mu_{\beta,\mathrm{meas}}-\vec \mu_\beta\|_2
 \leq
 \min\Brac{\frac{\lambda_n\varepsilon}{8\beta\sigma_n}, \frac{\lambda_n^2}{16\beta\sigma_n \poly(\beta)}}\,,
\]
then every minimizer $\vec g_\mathrm{meas}$ of $\mc L_\mathrm{meas}$ satisfies
$ \|\vec g_\mathrm{meas}-\vec \bsg\|_2\leq\varepsilon.$
Since each $\Gamma_I$ is a $\{\pm1\}$-valued observable, inverse-polynomial accuracy in the moments requires only polynomially many samples. Thus the problem reduces to proving, with high probability over the SYK disorder,
\[
 \nabla^2\Phi(\vec \bsg)\succeq n^{-O_\beta(1)}I.
\]

\subsubsection{From the Hessian to an ordinary variance}
We next reduce strong convexity of the log-partition function to a lower bound on  
thermal variance.  Fix a direction $\vec u\in\R^m$ and associate to it the quartic observable
$    A_{\vec u}:=\sum_{I\in\calI}u_I\Gamma_I.$ 
The Hessian of the log-partition function in the direction $\vec u$ has a particularly clean
interpretation.  Differentiating
$\Phi(\vec g)=\log Z(\beta,\vec g)$ twice gives
\begin{align}
\label{eq:variancebound}
    \vec u^\top\nabla^2\Phi(\vec g)\vec u
    =
    \beta^2\sigma_n^2
    \Var^{\BKM}_{\rho_\beta(\vec g)}(A_{\vec u}),
\end{align}
where the Bogoliubov--Kubo--Mori variance is
\begin{align}
\label{eq:bkmvariance}
    \Var^{\BKM}_{\rho}(A)
    :=
    \int_0^1
    \Tr\!\left[
        \rho^s A_0 \rho^{1-s}A_0
    \right]ds,
    \qquad
    A_0:=A-\Tr[\rho A]I.
\end{align}
This is the non-commutative analogue of the fact that the Hessian of a classical log-partition function is a covariance matrix.
The only difference is
that, because $A_{\vec u}$ need not commute with $\rho_\beta$, the usual variance is replaced by
the BKM variance. It is useful, however, to work with the ordinary variance
\begin{align}
\label{eq:usualvariance}
    \Var_{\rho}(A)
    =
    \Tr[\rho A_0^2] 
\end{align}
and below we relate \Cref{eq:bkmvariance} and \Cref{eq:usualvariance}. Diagonalizing
$H(\vec g)$, and writing $q_a=e^{-\beta\lambda_a}/(d_nZ(\beta,\vec g))$ for the Gibbs weights, one can express these terms as
\[
    \Var^{\BKM}_{\rho_\beta}(A)
    =
    \sum_{a,b}\,|(A_0)_{ab}|^2\cdot \int_0^1 q_a^s q_b^{1-s}\,ds,
    \qquad
    \Var_{\rho_\beta}(A)
    =
    \sum_{a,b}q_a\,|(A_0)_{ab}|^2.
\]
Now, letting
$\Delta_{\vec g}=\lambda_{\max}(H(\vec g))-\lambda_{\min}(H(\vec g))$
denote the spectral diameter, observe that 
$$
 \int_0^1 q_a^s q_b^{1-s}\,ds\geq \frac{q_a}{1+\beta\Delta_{\vec g}},
 $$
hence
\begin{align}
\label{eq:BKMtoclassical}
    \Var^{\BKM}_{\rho_\beta(\vec g)}(A)
    \geq
    \frac{1}{1+\beta\Delta_{\vec g}}
    \Var_{\rho_\beta(\vec g)}(A) \stackrel{\eqref{eq:variancebound}}{\implies}
    \vec u^\top\nabla^2\Phi(\vec g)\vec u
    \geq
    \frac{\beta^2\sigma_n^2}
         {1+\beta\Delta_{\vec g}}
    \Var_{\rho_\beta(\vec g)}(A_{\vec u}).
\end{align}
For a typical SYK instance, $\|H(\vec g)\|_{\op}=O(n)$, and hence
$\Delta_{\vec g}=O(n)$.  Thus the passage from the Hessian to ordinary variance loses only a
polynomial factor.  In particular, it is sufficient to establish
\[
    \Var_{\rho_\beta(\vec g)}(A_{\vec u})
    \geq
    n^{-O_\beta(1)}\|\vec u\|_2^2
\]
 for every direction $\vec u$.
The remaining task is therefore purely a variance lower bound.

\subsubsection{Exact coefficient extraction and local perturbations}

It remains to prove an inverse-polynomial lower bound on 
$\Var_{\rho_\beta}(A_{\vec u})$.  Normalize $\|\vec u\|_2=1$ and~let
\[
    V=A_{\vec u}=\sum_{I\in\calI}u_I\Gamma_I,
    \qquad
    V_0=V-\Tr(\rho_\beta V)\Id.
\]
Our first observation is that, from an arbitrary global quartic observable $V$, we can
isolate any one of its coefficients \emph{exactly}, without paying any factor depending on the
interaction degree. Let $S \in \calI$. Define the Majorana twirl,
\[
	\mc P_S(X) = \frac 1 {16} \sum_{A \subseteq S} (-1)^{|A|} \Gamma_A X \Gamma_A^\dagger\,,
\]
which isolates the coefficient $u_S$. Using $\Gamma_A\Gamma_I\Gamma_A^\dagger = (-1)^{|A\cap I|}\Gamma_I,$ we have that
\begin{align}
\label{eq:termextraction}
	\mc P_S(V_0) = \frac 1 {16} \sum_{A\subseteq S} (-1)^{|A|} \Gamma_A V_0 \Gamma_A^\dagger =
	\frac 1 {16} \sum_{I \in \calI} u_I \Paren{ \sum_{A\subseteq S} (-1)^{|A|+|A\cap I|}
	}\Gamma_I = u_S\Gamma_S\,
\end{align}
where we used that the inner sum vanishes unless $I=S$. Thus one coefficient of an arbitrary quartic observable can be extracted exactly, with no dependence
on the interaction degree. After algebraic manipulation, this isolation procedure allows us to write
\[
	u_S^2 \leq \frac 1 {16} \sum_{A \subseteq S} \Tr[\rho_A V_0^2]\,,\qquad \rho_A = U_A \rho
	U_A^\dagger\,
\]
where $U_A=\Gamma_A$. 
Thus, we have reduced the problem to relating the variance of the locally perturbed Gibbs state
$\rho_A$ to that of the true state $\rho$. The comparison is controlled by a Petz R\'enyi power
\cite{petz1986quasi} between the two states,
\[
	R_{p,A} = \frac{\Tr[e^{-(\beta+s) H} e^{s U_A H U_A^\dagger}]}{Z(\beta, \vec g)} =
	\frac{Z_A(\beta, s, \vec g)}{Z(\beta, \vec g)}\,.
\]
This argument is made in \Cref{sampleall:lem:variance}. Next, we show in \Cref{alltemp:thm:quench} that the corresponding annealed quantity satisfies
\begin{equation}
	\label{eq:annealed_ratio}
	\E_{\vec \bsg} Z_A(\beta, s, \vec \bsg) / \E_{\vec \bsg} Z(\beta, \vec \bsg) \leq \frac 8
	3\,.
\end{equation}
Assuming this bound, we restrict to the event
$Z(\beta,\vec \bsg)\geq n^{-1}\E_{\vec \bsg}Z(\beta,\vec \bsg)$, which has high probability.
On this event, the denominator can be replaced by its annealed value at a polynomial cost;
Markov's inequality then bounds the average of the ratios $R_{p,A}$ with high probability.
Finally, this yields the chain of inequalities,
\begin{equation}
	\label{eq:variance_lower}
	\Var_{\rho_\beta}(V) \geq \Paren{ \binom n 3 m^{p-1} \frac 1 {16} \sum_{|S| = 4} \sum_{A
	\subseteq S} R_{p, A}}^{-1/(p-1)} \|\vec u\|_2^2 \geq \Paren{n^{4p+5}}^{-1/(p-1)} \|\vec u\|_2^2 \geq
	n^{-4 - 9\beta / s} \|\vec u\|_2^2\,,
\end{equation}
where we have taken $p = 1 + s / \beta$. 

\paragraph{Proving \Cref{eq:annealed_ratio}.} We want to control the degree to which a local
operation $\Gamma_A$ can affect the Gibbs state, as measured by the Petz R\'enyi power,
\[
	\Tr[(U_A \rho U_A^\dagger)^p \rho^{1 - p}]\,.
\]
This bound is reminiscent of the ``leakage'' controlled in \cite{arad2016connecting}, and indeed,
bounding such a quantity directly yields a sample-efficient algorithm for bounded-degree graphs
(see \cite{anshu2020sample} and a simplified proof in \Cref{app:bounded_degree}). In an earlier
version of this work, we controlled this quantity by Taylor expanding exponentials $e^{-\beta H}$
with respect to $\beta$ and applying Koteck\'y--Preiss, as in our time-efficient result. However,
generative AI observed that we can instead bound the annealed local-quench ratio by a series of
iterated commutator moments. Writing $H=H(\vec \bsg)$ and $C_k=C_k(\vec \bsg)$ inside expectations,
the terms are, for $k\geq0$,
\[
	\Paren{\frac {s^{2k}}{(2k)!}} M_{2k} = \Paren{\frac {s^{2k}}{(2k)!}}\frac{\E_{\vec \bsg}
	\Tr[e^{-\beta H} C_k C_k^\dagger]}{\E_{\vec \bsg} Z(\beta, \vec \bsg)}\,, \qquad C_k(\vec g) =
	\operatorname{ad}_{H(\vec g)}^k (U_A)\,.
\]
Expanding the commutators and counting the terms which survive when we average over the
SYK couplings yields the bound $M_{2k} \leq [8k ( 1 + \beta)]^{2k}$ in \Cref{alltemp:lem:moments}.
Choosing $s=O((1+\beta)^{-1})$ makes the resulting series summable and yields
\Cref{eq:annealed_ratio}.

\subsubsection{Completing the analysis}
Putting together \Cref{eq:variancebound}, \Cref{eq:BKMtoclassical}, \Cref{eq:variance_lower}, we get
that
\[
	\nabla^2\Phi(\vec \bsg)\succeq\lambda_n\Id, \qquad \lambda_n = \frac{\beta^2 \sigma_n^2}{1 +
	\beta\Delta_{\vec \bsg}} n^{-O_\beta(1)} = n^{-O_\beta(1)}.
\]
Since each $\Gamma_I$ is a
$\{\pm1\}$-valued observable, its expectation can be estimated from repeated measurements
using Hoeffding's inequality.  Estimating every $\mu_{\beta,I}$ to accuracy
$\delta/\sqrt m$ and taking a union bound over all $I$ gives
$\|\vec \mu_{\beta, \mathrm{meas}}-\vec \mu_\beta\|_2\leq\delta$ using
$\textsf{poly}(n,1/\delta)\log(1/\zeta)$ copies.   As discussed in Section~\ref{sec:maximumentropy} this implies
\[
\|\vec g_{\mathrm{meas}}-\vec \bsg\|_2\leq\varepsilon.
\]
with probability at least $1-\zeta$ over the measurements.    Thus, despite the all-to-all interaction structure of SYK, inverse-polynomial accuracy in the couplings can be achieved with only polynomially many copies of the Gibbs state, proving our first result. 

\subsection{Time-efficient learning}
\label{sec:intro-time-efficient}

In the previous section, our main result was an \emph{information-theoretic} result that SYK models are learnable given copies of the Gibbs state. However, the main drawback was that the empirical maximum-entropy principle need not be  solvable in polynomial time. Recall that the sample-efficient algorithm above is based on solving the following optimization problem
\[
    \mathcal L(\vec h)
    =
    \log Z(\beta,\vec h)
    +
    \beta\sigma_n\,\vec h\cdot\vec\mu_\beta,
\]
whose unique minimizer is the \emph{true coupling} vector $\vec g$.  While this
gives an information-theoretically efficient reconstruction procedure, it
does not immediately yield an efficient algorithm for learning the unknown SYK~Hamiltonian.

\subsubsection{First-order inversion}
Unlike the sample-efficient procedure, the time-efficient algorithm does not
solve the maximum-entropy optimization problem.  
Instead, it works with the same measured quartic expectations and seeks to directly invert the
map from the SYK couplings to these observables.
 Recall that
$    H(\vec g)
    =
    \sigma_n\sum_{I\in\mathcal I_4} g_I\Gamma_I,$ 
where $   \sigma_n=\sqrt{{6}/{n^3}}
$. 
Define
\[
    \mu_{\beta,I}(\vec g)
    :=
    \operatorname{Tr}[\rho_\beta(\vec g)\Gamma_I].
\]
At high temperature, these expectations are approximately proportional to
the corresponding couplings.  Indeed, a first-order Taylor expansion of the
Gibbs state around $\beta=0$ gives
\[
    \mu_{\beta,I}(\vec g)
    =
    -\beta\,\operatorname{Tr}[\Gamma_I H(\vec g)]
    +O(\beta^2)
    =
    -\beta\sigma_n g_I+O(\beta^2),
\]
where in the last equality we used the orthogonality relations of the
Majorana operators,
$\operatorname{Tr}[\Gamma_I\Gamma_J]=\delta_{I,J}$.
To recover $g_I$ at first order, we normalize the thermal expectations by $-\beta\sigma_n$, defining
\[
    U_{\beta,I}(\vec g) := -\frac{\mu_{\beta,I}(\vec g)}{\beta\sigma_n} = - \frac{\Tr[\Gamma_I
    \rho_\beta(\vec g)]} {\beta\sigma_n}.
\]
With this normalization, the thermal expectation map is close to the
identity.  More precisely, expanding $\mu_{\beta,I}$ one order further gives
\begin{equation}
    \label{eq:intro-time-expansion}
    U_{\beta,I}(\vec g)
    =
    g_I-\frac{\beta\sigma_n}{2}
    \sum_{J_1\triangle J_2=I}
    \chi(I;J_1,J_2)\,g_{J_1}g_{J_2}
    +O(\beta^2),
\end{equation}
where
\[
    \chi(I;J_1,J_2)
    :=
    \operatorname{sign}
    \Tr[\Gamma_I\Gamma_{J_1}\Gamma_{J_2}].
\]
Thus, after the natural rescaling by $\beta\sigma_n$, the measured expectations agree with the
disorder vector $\vec g$ at linear order, and the remaining task is to correct the nonlinear terms
beginning with the quadratic contribution in \Cref{eq:intro-time-expansion}. This observation almost
gives an algorithm by itself, but not quite.  For a fixed $I$, the quadratic term in
\Cref{eq:intro-time-expansion} contains $N_2=\Theta(n^2)$ essentially orthogonal Gaussian monomials.
Since $\sigma_n=\Theta(n^{-3/2})$ and there are $N_2=\Theta(n^2)$ orthogonal quadratic monomials,
the Gaussian $L^2$ norm of this quadratic correction is
\[
    \Theta_\beta(\sigma_n\sqrt{N_2})
    =
    \Theta_\beta(n^{-1/2}).
\]
Thus, even if one were given the thermal expectations exactly, the naive inversion map 
leaves an $n^{-1/2}$-scale quadratic bias in each coordinate.  The main
challenge in the time-efficient argument which we deal with below is to identify and
cancel this leading nonlinear contribution.

\subsubsection{A polynomial approximation to $U_\beta$}
\label{sec:poly_approximation}
The exact map $\vec U_\beta=(U_{\beta,I}(\vec g))_I$ is still difficult to analyze directly because
it contains the random partition function in the denominator, i.e., 
\[
    U_{\beta,I}(\vec g)
    =
    -
    \frac{
        \operatorname{Tr}[
            \Gamma_Ie^{-\beta H(\vec g)}
        ]
    }{
        \beta\sigma_n Z(\beta,\vec g)
    }.
\]
So, we first approximate this map by an explicit low-degree polynomial in the Gaussian
disorder variables.~Write
\[
    Z(\beta,\vec g)
    =
    Z_{\rm ann}(\beta)
    \bigl(1+\Delta_\beta(\vec g)\bigr),
    \qquad
    Z_{\rm ann}(\beta)
    :=
    \E_{\vec \bsg}[Z(\beta,\vec \bsg)].
\]
Applying the Taylor expansion
$    \frac{1}{1+\Delta_\beta}
    =
    1-\Delta_\beta
    +
    O(\Delta_\beta^2)$ 
suggests replacing the factor $(1+\Delta_\beta)^{-1}$ by its first-order approximation. 
This leads to the untruncated proxy
\[
    V_{\beta,I}^{\mathrm{full}}(\vec g)
    :=
    -
    \frac{
        \operatorname{Tr}[
            \Gamma_Ie^{-\beta H(\vec g)}
        ]
    }{
        \beta\sigma_n Z_{\rm ann}(\beta)
    }
    \left(
        2-
        \frac{Z(\beta,\vec g)}{Z_{\rm ann}(\beta)}
    \right).
\]
We then truncate its Taylor series in $\beta$ at degree
$K=O(\log n)$. Since the coefficient of $\beta^k$ is a polynomial of degree $O(k)$ in
$\vec g$, the resulting proxy $V_{\beta,I}$ is a Gaussian polynomial of
degree $O(\log n)$. There are a couple of subtleties in implementing this idea.

\begin{enumerate}
    \item First, replacing
$(1+\Delta_\beta)^{-1}$ by $1-\Delta_\beta$ produces an error quadratic
in $\Delta_\beta$.  The previously known second-moment estimate~\cite{zlokapa2026rigorous},
\[
    \E_{\vec \bsg}|\Delta_\beta(\vec \bsg)|^2
    =
    O(n^{-1}),
\]
is sufficient to control this quadratic error in $L_1$, but our final
reconstruction uses quadratic combinations of the approximations
$V_{\beta,I}$ to the rescaled expectations $U_{\beta,I}$ and measures the
resulting error in $L^2$.  Concretely, we must control terms of the form
\[
    \bigl\|
        (U_{\beta,J}-V_{\beta,J})\,U_{\beta,K}
    \bigr\|_{L^2}
    \le
    \|U_{\beta,J}-V_{\beta,J}\|_{L^4}
    \|U_{\beta,K}\|_{L^4},
\]
by H\"older's inequality.  Thus we need $L^4$ control of
$U_{\beta,I}-V_{\beta,I}$, rather than only control in expectation.
Moreover, the error incurred by replacing $(1+\Delta_\beta)^{-1}$ by
$1-\Delta_\beta$ contributes a term of the form
$\Delta_\beta^2 U_{\beta,I}$, so controlling the proxy error in $L^4$
requires fixed higher moments of $\Delta_\beta$.  We therefore prove the
stronger weighted estimate
\begin{equation}
    \|\Delta_\beta\|_{2,s}
    \le
    C_s n^{-1},
    \label{eq:intro-strong-thermal}
\end{equation}
for every fixed $s\ge1$.  For the algorithm it suffices to take
$s=\sqrt{15}$, since Gaussian hypercontractivity then controls all
$L^p$ norms up to $p=16$.  This strengthened estimate is one of
the main structural results specific to the time-efficient algorithm, and its
proof requires a more refined analysis of the Wick expansion
and associated polymer expansion than is needed for the
sample-efficient result.

\item Second, we truncate the $\beta$-expansion defining $V_{\beta,I}^{\mathrm{full}}$.
Define
\[
    V_{\beta,I}^{\mathrm{full}}(\vec g)
    =
    \sum_{k\ge 0} v_{I,k}(\vec g)\,\beta^k,
\text{ and }    V_{\beta,I}(\vec g)
    :=
    \sum_{k=0}^{K} v_{I,k}(\vec g)\,\beta^k.
\]
The nontrivial technical step is to obtain a uniform $L^4$ bound on $V_{z,I}^{\rm full}$ throughout a
constant complex neighborhood of the origin.  This follows from the strengthened estimates above,
together with zero-freeness of $Z_{\rm ann}(z)$.  Once this is established, Cauchy's estimate
bounding the derivatives of analytic functions gives
\[
    \|v_{I,k}\|_{L^4} \leq C\,[\beta_{\rm thermal}(\sqrt{15})]^{-k}\,.
\]
In the temperature regime of the algorithm, $|\beta|/\beta_{\rm thermal}(\sqrt{15})\le 1/8$, and hence
\[
    \|V_{\beta,I}^{\mathrm{full}} - V_{\beta,I}\|_{L^4} \leq C\sum_{k > K}8^{-k}.
\]
Taking $K = \Theta(\log n)$ makes this tail inverse-polynomially small; with the choice used in the
proof it is $O(n^{-3})$.  Combining this truncation error with the error from replacing $(1 +
\Delta_\beta)^{-1}$ by $1 - \Delta_\beta$ yields
\begin{equation}
    \label{eq:intro-proxy-error}
    \max_{I\in\mathcal I_4} \|U_{\beta,I}-V_{\beta,I}\|_{L^4} \leq O_\beta(n^{-3/2}).
\end{equation}
\end{enumerate}

\subsubsection{Approximating the inverse map}
Now that we have an explicit polynomial approximation
$\vec V_\beta(\vec g)$ to the rescaled thermal expectations, the remaining
task is to approximately invert the map
$    \vec g \longmapsto \vec V_\beta(\vec g).$  A priori, even correcting only the quadratic part of this map seems
cumbersome: a general quadratic map
$\mathbb R^m\to\mathbb R^m$ is specified by $\Theta(m^3)$ coefficients,
so directly learning or inverting its quadratic correction would require
handling a very large parameter space. 
The key simplification is that the SYK model has a large symmetry group,
which severely restricts the possible low-degree corrections.

The polynomial map $\vec V_\beta$ is covariant under sign changes of the
underlying Majorana modes and under permutations of the $n$ physical
indices.  To see why this is restrictive, consider a fixed output
coordinate $I$.  Under a sign change
$\vec x=(x_1,\ldots,x_n)\in\{\pm1\}^n$, the coupling $g_J$ transforms as
\[
    g_J \longmapsto x_J g_J,
    \qquad
    x_J:=\prod_{j\in J}x_j,
\]
while the $I$-th output must transform as
\[
    V_{\beta,I}(\vec x\cdot\vec g)
    =
    x_I V_{\beta,I}(\vec g).
\]
This already rules out almost all degree-$\le2$ monomials.  A linear
monomial $g_J$ has the required sign only when $J=I$, while a quadratic
monomial $g_{J_1}g_{J_2}$ has the required sign only when
\[
    x_{J_1}x_{J_2}=x_I
    \qquad\text{for every }\vec x,
\]
which is equivalent to
$    J_1\triangle J_2=I.$ 
Thus, for a fixed output coordinate $I$, the only possible degree-one term
is $g_I$, and the only possible degree-two terms are precisely those
indexed by pairs $(J_1,J_2)$ satisfying
$J_1\triangle J_2=I$.

Permutation covariance then forces all such surviving terms to appear with
the same coefficient, up to the fermionic sign coming from the Majorana
ordering.  Consequently, the entire covariant Hermite degree-$\le2$ space
is spanned by only two vector-valued polynomials,
\begin{equation}
	\label{eq:two_hermite_span}
	\vec\Phi_1(\vec g) = \vec g \text{ and } \Phi_{2,I}(\vec g) := \frac{1}{2\sqrt{N_2}}
	\sum_{J_1\triangle J_2=I} \chi(I;J_1,J_2)\, g_{J_1}g_{J_2}\,,
\end{equation}
where $\chi(I; J_1, J_2) = \mathsf{sign}(\Tr[\Gamma_I \Gamma_{J_1} \Gamma_{J_2}])$. Hence,
although the ambient quadratic space is enormous, the symmetries reduce the low-degree correction
problem to a two-dimensional one.
In particular, the quadratic statistic $\Phi_2$ is precisely the structure
of the leading correction appearing in
\Cref{eq:intro-time-expansion}.  This symmetry reduction is what makes it
possible to cancel the nonlinear bias without learning a large collection
of additional parameters. We therefore consider a reconstruction algorithm of the form
\begin{align}
\label{eq:calibcoeff}
    \vec f_{\rm reconstruct}(\vec x)
    =
    c_1\vec\Phi_1(\vec x)
    +
    c_2\vec\Phi_2(\vec x).
\end{align}
Rather than requiring
$\vec f_{\rm reconstruct}(\vec V_\beta(\vec g))=\vec g$ exactly, we only ask
that this identity hold in Hermite degrees at most two:
\begin{equation}
    \Pi_{\le2}
    \vec f_{\rm reconstruct}(\vec V_\beta(\vec g))
    =
    \vec g.
    \label{eq:intro-calibration-goal}
\end{equation}
By the covariance properties above, composition with $\vec V_\beta$ and
projection onto degree at most two defines a map from
$\operatorname{span}\{\Phi_1,\Phi_2\}$
to itself.
Hence \eqref{eq:intro-calibration-goal} reduces to inverting the
$2\times2$ matrix
\[
	M_\beta=M_\mathrm{cal}(\beta)=\begin{bmatrix} m_{1,1} & m_{1,2} \\ m_{2,1} & m_{2,2} \end{bmatrix}\quad \text{where}
	\quad m_{i,j}=\E_{\vec \bsg}[\Phi_{i,I}(\vec \bsg)\Phi_{j,I}(\vec{V}_\beta(\vec \bsg))]
	\quad\text{for any fixed }I\in\mc I_4,
\]
where we have appealed to the coordinatewise orthonormality of $\vec{\Phi}_i$. 
Therefore, choosing $c_1,c_2$ so that the degree-$\le2$ part of the
composition equals $\vec\Phi_1(\vec g)=\vec g$ amounts to solving the
$2\times2$ linear system 
\[
	\begin{bmatrix} c_1 \\ c_2\end{bmatrix} = M_\beta^{-1} \begin{bmatrix} 1 \\
	0\end{bmatrix}\,.
\]
At infinite temperature, $\vec V_0(\vec g)=\vec g$, so $M_0=I$. Since it is analytic in $\beta$, that implies that, for sufficiently small constant $\beta$, $M_\beta$ remains close to the identity and is therefore
 invertible. Thus the coefficients $c_1,c_2$ can be chosen to
recover the linear term exactly and cancel the entire quadratic quantity.

\subsubsection{Why cancelling degree two is enough}
Recall that the two-dimensional reduction applies only to the Hermite degree-$\le 2$ part of the
map.  The full polynomial $\vec V_\beta(\vec g)$ still contains higher-degree terms:
\[
    \vec V_\beta(\vec g) =
    \underbrace{\Pi_{\le 2}\vec V_\beta(\vec g)}_{\text{two-dimensional by symmetry}} +
    \underbrace{\sum_{d\ge 3}\Pi_d\vec V_\beta(\vec g)}_{\text{higher-degree terms}}.
\]
The symmetry argument shows that the first part lies in $\operatorname{span}\{\Phi_1,\Phi_2\}$. By our choice of $c_1,c_2$ in \Cref{eq:calibcoeff},
$\vec f_{\rm reconstruct}$ recovers the linear term exactly and cancels the entire quadratic contribution.

The remaining issue is to show that these higher-degree terms are small.
A second technical ingredient is therefore a quantitative bound on the
higher-degree terms of $\vec V_\beta$. To this end,  we prove
\[
    \|\Pi_d\Phi_{1,I}(\vec V_\beta)\|_{L^2}
    \le
    c_d\,n^{-(d-1)/2},
    \qquad
    \|\Pi_d\Phi_{2,I}(\vec V_\beta)\|_{L^2}
    \le
    c_d\,n^{-(d-2)/2},
\]
where $c_d$ grows only exponentially in $d\log d$.  Since the quadratic
coefficient satisfies $|c_2|=O(n^{-1/2})$, the two linear and quadratic terms contribute on the
same scale, and the degree-$d$ part of the reconstructed error is
essentially
\[
    n^{-(d-1)/2}.
\]
Hence the degree-two contribution, which would be of order $n^{-1/2}$, is
precisely the obstruction that must be removed.  Once it is cancelled, the
first surviving term has degree three and is already of order $n^{-1}$.
Furthermore, since $\vec V_\beta$ has degree only $K=O(\log n)$, the remaining
Hermite tail is a finite rapidly decreasing sum, dominated by its
degree-three contribution.

Finally, the algorithm is given the true rescaled expectations
$\vec U_\beta$, rather than the polynomial approximation $\vec V_\beta$.
Since $\vec f_{\rm reconstruct}$ is quadratic, a direct application of
H\"older's inequality gives a Lipschitz-type bound in terms of
coordinatewise $L^4$ norms.  Combining this stability with
\Cref{eq:intro-proxy-error} transfers the error bound from
$\vec f_{\rm reconstruct}(\vec V_\beta)$ to
$\vec f_{\rm reconstruct}(\vec U_\beta)$ to get
\[
    \frac{1}{\sqrt m}
    \left\|
        \vec \bsg
        -
        \vec f_{\rm reconstruct}(\vec U_\beta)
    \right\|_{L^2}
    \le
    O_\beta(n^{-1}).
\]
Equivalently, for the normalized Hamiltonian parameters $    \vec\theta
    =
    \beta\sigma_n\vec \bsg,$
the estimator
\[
    \vec\theta_{\rm ideal}
    =
    \beta\sigma_n
    \vec f_{\rm reconstruct}(\vec U_\beta)
\]
satisfies
\[
    \frac{1}{\sqrt n}
    \|\vec\theta_{\rm ideal}-\vec\theta\|_{L^2}
    \le
    O_\beta(n^{-1}).
\]
\subsubsection{Overall running time}
It remains to account for the computational cost of the reconstruction. The coefficients $c_1,c_2$
depend only on $n$ and $\beta$, and not on the particular SYK instance, so they may be computed in a
preprocessing step from the truncated polynomial $V_\beta$.  Since the truncation degree is
$K=O(\log n)$, this preprocessing takes time
\[
    n^{O(K)} = 2^{\operatorname{polylog}(n)}\,.
\]
Once $c_1,c_2$ are fixed, evaluating
\[
    \vec g_{\rm meas} = \vec f_{\rm reconstruct}(\vec U_{\beta,\rm meas})
\]
is polynomial-time: for each $I \in \calI$, the quadratic term $\Phi_{2,I}$ contains only $N_2 =
\Theta(n^2)$ summands, and there are $m = \Theta(n^4)$ output coordinates.  Thus the reconstruction
itself can be implemented using $\operatorname{poly}(n)$ arithmetic operations. 

\section{Discussion}

Our work continues a line of papers which develop rigorous algorithms for the SYK model, including
\cite{hastings2022optimizing, basso2024optimizing} and, most recently,
\cite{zlokapa2026rigorous}. The first gives polynomial-time classical certificates and quantum
constant-factor approximations to the optimal energy of quartic SYK, while the second gives
dissipative algorithms with guarantees on the disorder-averaged energy for random local
Hamiltonians, including SYK. The work of \cite{zlokapa2026rigorous} instead estimates thermal
expectations at high temperature in quasi-polynomial classical time. These works develop rigorous
alternatives to heuristics such as the replica method, using different aspects of the random
Hamiltonian. In particular, the Wick-pair expansion of \cite{zlokapa2026rigorous}, which we extend
in our quasi-polynomial time algorithm, exploits Gaussian disorder to control higher-degree
contributions after disorder averaging. We expect that these techniques will lead to increasingly
effective algorithms for simulating, learning, and preparing low-energy states of the SYK model.

Classically, a closely related line of work
~\cite{anari2024universality,gaitonde2024unified,chandrasekaran2025learning} has established
increasingly powerful guarantees for learning the \emph{Sherrington-Kirkpatrick (SK) model}, a
widely studied model for disordered magnets. In this model, the Hamiltonian is classical and
2-local, given by $H(x)\triangleq \frac{1}{\sqrt{n}}\sum_{i<j} g_{ij}x_ix_j$ for $x\in\{-1,1\}^n$
and $g_{ij}$ sampled i.i.d. from $\mathcal{N}(0,1)$. This model exhibits a replica symmetry breaking
phase transition at inverse temperature $\beta = 1$. The most recent work in this
line~\cite{chandrasekaran2025learning} established that one can learn the couplings $(g_{ij})$ given
samples from the Gibbs measure $\propto e^{-\beta H}$ even at $\beta = O(\sqrt{\log n})$, well
beyond this phase transition. Our result is qualitatively different in part because it is
conjectured that the SYK model does not exhibit such a phase transition. Our algorithm does however
bear resemblance to prior structural observations about the SK model below the phase transition:
in~\cite{el2024bounds,brennecke2024two}, bounds on the two-point correlation matrix for the SK model
below the phase transition were proven by relating that matrix to the disorder matrix itself. This
is the same premise behind our time-efficient algorithm, although we need to take higher-order local
observables than merely the sufficient statistics of the model.

\paragraph{The context of our results.} Our sample-efficient result begins by following the proof strategy developed by
\cite{anshu2020sample}; to prove sample-efficient learning, it suffices to show that the log
partition function is strongly convex. In the classical setting, the Hessian is directly
proportional to the covariance matrix of the original Hamiltonian terms; related classical
learning problems are studied in, e.g., \cite{montanari2015, vuffray2016interaction}.
For noncommuting quantum Hamiltonians, the Hessian is no longer the ordinary covariance, and the
two difficulties in \cite{anshu2020sample} are reducing strong convexity to the variance of a
``reasonable'' operator, and relating this variance to the coefficients of the original Hamiltonian. The
authors address these issues by allowing \emph{quasi-local} operators, followed by a local twirl
which effectively restores locality. The effect of this twirl is analyzed using a result of
\cite{arad2016connecting} which controls the ``leakage'' caused by a local operator on eigenspaces
of a Hamiltonian. In our work, we make a simplification to this argument and adapt it to the
setting of SYK. We use the standard identity relating the Hessian to the
Bogoliubov--Kubo--Mori covariance \cite{petz1993bogoliubov} and compare it directly to the ordinary
variance, with a loss controlled by the spectral diameter. This simplifies several aspects of
\cite{anshu2020sample}; see \Cref{app:bounded_degree} for details. The remaining challenge is
adapting the ideas of \cite{arad2016connecting} to SYK. Their bounds depend on the sum of interaction
norms incident to a site, which is typically $\Theta(n^{3/2})$ for quartic SYK even after the
$\sigma_n$ normalization. Consequently, a direct application does not give the constant-temperature
control needed here. Instead, we relate the variance to the local-quench ratio,
\[
	R_{p,A}=\frac{Z_A(\beta,s,\vec g)}{Z(\beta,\vec g)},\qquad p=1+s/\beta\,.
\]
This is a Petz R\'enyi power \cite{petz1986quasi} between the Gibbs state and its conjugate by
$U_A=\Gamma_A$, and plays an analogous role to the energy-leakage bounds of
\cite{arad2016connecting}. Nested-commutator expansions already appear in that work; our analysis
averages over disorder before bounding the resulting terms. Counting the non-zero Gaussian
pairings and balancing this with the SYK normalization $\sigma_n$ bounds our thermal leakage
analogue. To pass from annealed to quenched estimates, we use an elementary, weaker version of the
commutation-index and Gaussian-concentration argument of \cite{anschuetz2025strongly}; their sharper
commutation-index estimates are unnecessary for our bounds.

Our quasi-polynomial time result structurally bears resemblance to the work of
\cite{haah2022optimal}. They begin with forming a series expansion of the log partition function (we
expand the bare partition function instead) and use the assumption of bounded interaction degree to
truncate the expansion at degree logarithmic in the inverse target accuracy (hence $O(\log n)$
for inverse-polynomial accuracy at fixed temperature). At this point,
the authors use Newton's method to solve a system of polynomial equations. Our approach differs in
several crucial ways. First, as mentioned in the introduction, the interaction degree of SYK is
$\Theta(n^3)$, counting quartic terms that overlap a given term, and a direct application of
\cite{haah2022optimal} would not cover constant inverse temperature. As such, after expanding the
series $Z(\beta, \vec g)$, we instead apply Wick--Isserlis \cite{isserlis1918formula} and carefully
count the resulting non-zero contributions, balancing this with the
contribution of $\sigma_n$ from the SYK normalization. Secondly, the result of
\cite{haah2022optimal} implies strong convexity at all points traversed by Newton's method.
Unfortunately, we are only able to show this result with high probability and it is not clear that
gradient descent would succeed in this setting. Rather, we instead appeal to the strong symmetry
exhibited by SYK and show that at Hermite degree $\leq 2$ in the disorder, we need to invert
a \emph{linear} $2\times 2$ system. This Hermite degree is distinct from the Taylor truncation
degree above. The remaining technical challenge is showing that higher Hermite degrees $d\geq3$
do not contribute much; proving this requires generalizing and expanding on the techniques of
\cite{zlokapa2026rigorous}.

More broadly, our partition-function analysis belongs to the zero-free and truncated-expansion
approach to quantum thermal algorithms developed in \cite{harrow2020classical}. For the
disorder-averaged estimates we use the Wick-pair polymer construction of
\cite{zlokapa2026rigorous} and the abstract polymer criterion of Koteck\'y and Preiss
\cite{kotecky1986cluster}, in the form given by \cite[Theorem~1]{ueltschi2004cluster}.

\paragraph{Barriers to time efficiency.} We briefly remark on a barrier to making our
sample-efficient algorithm time-efficient. Recent work of Zlokapa and Kiani~\cite{zlokapa2026syk} gives
a quasi-polynomial algorithm for thermal expectations under a zero-free-region assumption, supported
for SYK at all constant temperatures by large-$n$ arguments. Zlokapa~\cite{zlokapa2026rigorous}
rigorously establishes high-temperature SYK estimates, which we use in our analysis. This might suggest
running gradient descent over the space of disorder parameters and using such algorithms
to estimate the mismatch between the observables under a candidate set of parameters
versus under the true parameters. Unfortunately, it is unclear how to \emph{certify} that the
parameters encountered over the course of gradient descent correspond to Hamiltonians which are
``SYK''-like in the sense that these algorithms would succeed. This is reminiscent of the gap
between the complexity of finding ground states of the classical SK Hamiltonian versus
\emph{certifying} bounds on the ground state energy.
Indeed, Montanari~\cite{montanari2025optimization} showed that the former is achievable in
polynomial time, whereas Ghosh et al.~\cite{ghosh2020sum} gave evidence that certifying any bound
beyond the naive spectral bound is computationally hard by showing a degree-$n^\delta$
sum-of-squares lower bound.

\paragraph{Concurrent work.} During the preparation of this manuscript, we were made aware of a recent work of Swingle~\cite{swingle2026single} which independently studied the relationship between local observables of the SYK model and the disorder parameters. While their work develops the connection using perturbative and heuristic arguments, our focus is on providing rigorous control of the contributions from higher-order terms in the expansion.

\paragraph{Acknowledgments.} SC was supported in part by NSF Award CCF-2430375. Part of this work was done while SC was visiting the Simons Institute
for the Theory of Computing. AA and YH acknowledge
support through NSF Award No 2238836. AA acknowledges support as a QCIS-FF: Quantum Computing \&
Information Science Faculty Fellow at Harvard University (NSF 2013303). YH is supported by an IBM PhD Fellowship.

\paragraph{AI Statement.} Generative AI built on authors' sketch of the sample complexity argument, and came up with the  versions included here, remarkably simplifying the analysis in \cite{anshu2020sample}.  
The central ideas behind the quasi-polynomial time
algorithms were human generated. Generative AI was used extensively for identifying the relevant
formulation of Koteck\'y-Preis and carrying out the combinatorial estimates required to complete the
proof. Prior to submission, a generative AI model suggested replacing our Koteck\'y-Preiss analysis
of the conjugated partition function $\E_{\vec \bsg} Z(\beta, s, \vec \bsg)$ with a simpler iterated
commutator to which Wick-Isserlis is applied. This allowed us to strengthen our result to apply at
all $\beta \geq 0$. We take full responsibility for correctness, exposition, and attribution of our
paper.

\section{Preliminaries}

\subsection{Definitions}
Throughout this work $n\geq 6$ is even, $\mc I_4=\binom{[n]}{4}$ is the set of $4$-tuples with elements in $[n]=\{1,\ldots,n\}$,
$m=|\mc I_4|$, and
\[
    d_n=2^{n/2}
\]
is the dimension of the Hilbert space carrying the irreducible representation of the
$n$ Majorana operators.  We always use the normalized trace $\Tr(X)=d_n^{-1}\operatorname{tr}_{\rm ordinary}(X)$, so that $\Tr(\Id)=1$. We define an instance of the SYK
model with disorder $\vec g$ as,
\[
  H_\mathrm{SYK}(\vec g)=H(\vec g) =\sigma_n\sum_{I\in\mc I_4}g_I\Gamma_I,\qquad \sigma_n =
  \sqrt{6/n^3}\,.
\]
Here, $\mc I_4 = \binom{[n]}4$ and $\Gamma_I$ are quartic Majorana terms, see \Cref{def:majorana}.

We use $\vec g=(g_I)_{I\in\mc I_4}$ for a fixed disorder vector and boldface for random
disorder: $\vec \bsg=(\bsg_I)_{I\in\mc I_4}$. Thus a disorder expectation is written
$\E_{\vec \bsg}f(\vec \bsg)$, where each coordinate $\bsg_I$ is drawn
independently from $\mc N(0,1)$. The scalar $\mu_\beta=\Tr[\rho_\beta H]$ is the mean energy; the vector
$\vec\mu_\beta=(\Tr[\rho_\beta\Gamma_I])_I$ consists of quartic expectations.
Define
\[
  Z(\beta,\vec g)=\Tr[e^{-\beta H(\vec g)}],\qquad
  \rho_\beta(\vec g)=\frac{e^{-\beta H(\vec g)}}{Z(\beta,\vec g)}.
\]
We write the log-partition function as $\Phi(\vec g) = \log Z(\beta, \vec g)$ when $\beta$ is fixed, and $F_{\vec g}(\beta) = \log Z(\beta, \vec g)$ when we are treating $\beta$ as the variable.

We work with the normalized trace so that $Z(0,\vec g)=1$ and $F_{\vec g}(0) = \log Z(0, \vec g) = 0$.
Because we use the normalized trace, $\rho_\beta$ satisfies
$\Tr(\rho_\beta)=1$, but its ordinary matrix trace is $d_n$.
The corresponding physical density matrix is therefore $\rho_\beta/d_n$. We write $Z_\mathrm{ann}(\beta)=\E_{\vec \bsg}Z(\beta,\vec \bsg)$ for the annealed
partition function. The subscript $\mathrm{meas}$ denotes measured moments or
the parameters recovered from them, for example $\vec g_\mathrm{meas}=(g_{\mathrm{meas},I})_{I\in\mc
I_4}$ and $\vec\mu_{\beta,\mathrm{meas}}=(\mu_{\beta,\mathrm{meas},I})_{I\in\mc I_4}$.

\begin{definition}[Multivariate Hermite polynomials]
	Multivariate Hermite polynomials are given by a multi-index $\vec d \in \N_0^m$, where the total degree is $|\vec d| =
	\sum_i d_i$. We use the normalized probabilists' polynomials
$\mathsf{Herm}_k=\mathrm{He}_k/\sqrt{k!}$, so
$\mathsf{Herm}_0=1$, $\mathsf{Herm}_1(x)=x$, and
$\mathsf{Herm}_2(x)=(x^2-1)/\sqrt2$. The corresponding multivariate Hermite polynomial is,
	\[
		\mathsf{Herm}_{\vec d}(\vec g) = \prod_{I \in \mc I_4} \mathsf{Herm}_{d_I}(g_I)\,.
	\]
\end{definition}

\begin{definition}[Majorana fermions]
	\label{def:majorana}
  Let $\{\psi_i\}_{i\in[n]}$ be Hermitian operators on the full $d_n$-dimensional
  Clifford space, satisfying $\{\psi_i,\psi_j\}=2\delta_{ij}\Id$. For
  $A=\{a_1<\cdots<a_k\}\subseteq[n]$, write
  \[
    \Gamma_A=\psi_{a_1}\cdots\psi_{a_k},\qquad \Gamma_\varnothing=\Id.
  \]
  These operators are unitary and
  $\Gamma_A^\dagger=(-1)^{k(k-1)/2}\Gamma_A$. In particular, a quartic term
  $\Gamma_I$, $I\in\mc I_4=\binom{[n]}4$, is Hermitian, squares to the identity,
  and satisfies $\Tr[\Gamma_I\Gamma_J]=\1_{I=J}$ and $\Tr[\Gamma_I]=0$.
  Conjugations by a general $\Gamma_A$ always include the adjoint.
\end{definition}

\begin{definition}[Gaussian norm]
	\label{def:gaussian_norm}
	Let $\vec{F}:\R^{d_1}\to\C^{d_2}$. We define the Gaussian norm,
	\[
		\|\vec{F}\|_{L^p} = \Paren{\E_{\vec \bsg \sim \mc N(\vec 0, \Id_{d_1 \times d_1})}
		\|\vec{F}(\vec \bsg)\|_2^p}^{1/p}\,.
	\]
	where $\|\vec{F}(\vec \bsg)\|_2$ is the Euclidean norm on $\C^{d_2}$.
\end{definition}

For scalar functions (and scalar--vector multiplication), this norm satisfies
H\"older's inequality: for $p,r,s\geq1$ with $\tfrac 1 p = \tfrac 1 r + \tfrac 1 s$,
\begin{equation}
	\label{eq:holder}
	\|FG\|_{L^p} \leq \|F\|_{L^r} \|G\|_{L^s}\,.
\end{equation}

Additionally, we'll need a weighted Gaussian norm, which penalizes higher-degree terms.

\begin{definition}[Weighted Gaussian norm]
	\label{def:weighted_norm}
	For a function $\vec{F}(\vec g)$ with Hermite expansion $\sum_{d \geq 0} \Pi_d \vec{F}$,
	we write,
	\[
		\|\vec{F}\|_{2,s} = \sqrt{\sum_{d \geq 0} |s|^{2d} \|\Pi_d \vec{F}\|^2_{L^2}}\,.
	\]
	The value may be $+\infty$. Note by orthonormality of Hermite polynomials, $\|\vec{F}\|_{2,1} = \|\vec{F}\|_{L^2}$. 
\end{definition}
\noindent This norm can be related to the Gaussian norm via the following lemma. The proof is via
hypercontractivity.
\begin{lemma}
	\label{lem:norm_conversion}
	Let $\vec{F} : \R^{d_1}\rightarrow\C^{d_2}$ have finite weighted norm, let $p\geq2$, and set $s=\sqrt{p-1}$. Then,
	\[
		\|\vec{F}\|_{L^p} \leq \|\vec{F}\|_{2, s}\,.
	\]
\end{lemma}
\begin{proof}
For a scalar function, this is Gaussian hypercontractivity, applied to the
Ornstein--Uhlenbeck operator. For a vector-valued function, Minkowski's inequality in
$L^{p/2}$ gives $\|\vec{F}\|_{L^p}^2\leq\sum_j\|F_j\|_{L^p}^2$.
Apply the scalar bound to each real and imaginary component and sum.
\end{proof}

\subsection{Standard results}

\begin{lemma}[Estimating quartic expectations]
\label{lem:measurement_complexity}
Let $n\geq8$ be even, $m=\binom n4$, $\beta\in\R$, and $\vec g\in\R^m$. For any $\delta, \zeta \in
(0,1)$, there is an algorithm using
\begin{equation}
	\label{eq:measurement_copies}
	N = m \left\lceil \frac{2}{\delta^2} \log\frac{2m}{\zeta} \right\rceil = \O
	\Paren{\frac{n^4}{\delta^2} \log \frac{n}{\zeta}}
\end{equation}
copies of the Gibbs state $d_n^{-1}\rho_\beta(\vec g)$ which returns estimates
$\vec\mu_{\beta,\mathrm{meas}}\in[-1,1]^m$ satisfying
\begin{equation}
\label{eq:measurement_accuracy}
 \max_{I\in\mc I_4}\left|\mu_{\beta,\mathrm{meas},I}
       -\Tr[\Gamma_I\rho_\beta(\vec g)]\right|\leq\delta
\end{equation}
with probability at least $1-\zeta$. In particular, $\|\vec\mu_{\beta,\mathrm{meas}} - \vec
\mu_\beta\|_2 \leq \sqrt m\, \delta$. The algorithm runs in time $\O(N \cdot n)$. 
\end{lemma}

\begin{proof}
Each quartic observable $\Gamma_I$ is Hermitian and satisfies $\Gamma_I^2=\Id$, so its measurement
outcome is in $\{\pm1\}$. Let the expected outcome be $\mu_{\beta,I} = \Tr[\Gamma_I \rho_\beta]$.
For each $I \in \calI$, we measure $\Gamma_I$ on $N_0=\lceil2\delta^{-2}\log(2m/\zeta)\rceil$ copies
and write the empirical mean as $\mu_{\beta,\mathrm{meas},I}$. Hoeffding's inequality and a union
bound give
\[
	\Pr\Brac{\max_I|\mu_{\beta,\mathrm{meas},I} - \mu_{\beta,I}| > \delta} \leq 2m \exp(-N_0
	\delta^2/2) \leq \zeta\,.
\]
The time complexity bound follows since we can write $\Gamma_I$ as a tensor product of $n/2$
single-qubit operators and measure $\Gamma_I$ by measuring each single-qubit operator and
multiplying the outcomes.
\end{proof}

\begin{lemma}[Derivative identities]
  \label{lem:first_derivative_identity}
  Let $H(\vec g)=\sigma_n\sum_I g_I\Gamma_I$ and
  $\Phi_\beta(\vec g)=\log Z(\beta,\vec g)$. For real $\beta$,
  \[
    \partial_{g_I}Z(\beta,\vec g)
       =-\beta\sigma_n\Tr[\Gamma_I e^{-\beta H(\vec g)}],\qquad
    \partial_{g_I}\Phi_\beta(\vec g)
       =-\beta\sigma_n\Tr[\rho_\beta(\vec g)\Gamma_I].
  \]
  When $\beta$ is fixed, we abbreviate $\Phi_\beta$ by $\Phi$.
\end{lemma}
\noindent We will also make use of the Duhamel formula for the derivative of an exponential of a noncommuting
matrix.
\begin{lemma}[Duhamel formula]
	\label{lem:duhamel}
	Let $K(t)$ be a differentiable family of matrices. Then,
	\[
		\frac{\dd}{\dd t} e^{K(t)} = \int_0^1 e^{(1-s) K(t)} K'(t) e^{s K(t)} \dd s\,.
	\]
\end{lemma}

\begin{lemma}[Spectral norm and infinite-temperature second moment]
  \label{lem:spectral_bounds}
  There is a universal $C_\mathrm{spec}$ such that, for all sufficiently large $n$, except with
  probability $e^{-\Omega(n)}+e^{-\Omega(m)}$ over $\vec \bsg$,
  \begin{equation}
    \label{eq:spectral}
    \|H(\vec \bsg)\|_\op\leq C_\mathrm{spec}n,\qquad
    \frac n8\leq\Tr[H(\vec \bsg)^2]
      =\sigma_n^2\sum_{I\in\mc I_4}\bsg_I^2\leq\frac n3.
  \end{equation}
\end{lemma}
\begin{proof}
Write $H=H(\vec \bsg)$ throughout this proof. Orthogonality of the quartic monomials gives the identity. Moreover,
\[
  m\sigma_n^2=\frac n4\left(1-\frac1n\right)
       \left(1-\frac2n\right)\left(1-\frac3n\right).
\]
The usual chi-square Chernoff bound puts $\sum_I\bsg_I^2$ in
$[3m/4,4m/3]$ except with probability $e^{-\Omega(m)}$, which gives the displayed
variance bounds for sufficiently large $n$.
For completeness, Wick's formula and $|\Tr[\Gamma_{I_1}\cdots\Gamma_{I_{2k}}]|\leq1$
give $\E_{\vec \bsg}\Tr[H^{2k}]\leq(2k-1)!!(m\sigma_n^2)^k$; odd averaged moments vanish.
Consequently $\E_{\vec \bsg}\Tr[e^{tH}]\leq e^{t^2m\sigma_n^2/2}$ for real $t$.
Since $e^{t\lambda_{\max}(H)}\leq d_n\Tr[e^{tH}]$ for $t>0$, applying the
exponential Markov bound to both extreme eigenvalues and optimizing $t$ gives
\[
  \Pr_{\vec \bsg}[\|H\|_\op\geq u]\leq2d_n\exp\!\left(-\frac{u^2}{2m\sigma_n^2}\right).
\]
Taking $u=C_\mathrm{spec}n$ with a sufficiently large universal $C_\mathrm{spec}$ proves the claim.
\end{proof}

Since all errors in this work are on the order $1 / \cpoly(n)$, the above error is negligible; as
such, beyond pointing out where this assumption is necessary, we will not factor this error into the
calculations.

\begin{lemma}[Elementary quartic bounds]\label{alltemp:lem:elementary}
Put $B_n=\binom n3$. For every positive $\omega$ with $\operatorname{Tr}(\omega)=1$ and every real $\vec u\in\mathbb R^m$,
\begin{align}
 \sum_{|I|=4}\operatorname{Tr}(\omega\Gamma_I)^2&\le B_n/4,\label{alltemp:eq:elementarymoments}\\
 \left\|\sum_Iu_I\Gamma_I\right\|_{\rm op}&\le\frac{\sqrt{B_n}}2\|\vec u\|_2.\label{alltemp:eq:elementaryop}
\end{align}
In particular, $\sigma_n^2 B_n\le1$.
\end{lemma}
\begin{proof}
For each fixed triple $T$, the Hermitian unitaries $\Gamma_{T\cup\{j\}}$, $j\notin T$, pairwise anticommute: two different supports have intersection of size three. Put $x_j=\operatorname{Tr}(\omega\Gamma_{T\cup\{j\}})$ and $Q=\sum_{j\notin T}x_j\Gamma_{T\cup\{j\}}$. Anticommutation gives $Q^2=(\sum_jx_j^2)I$. Thus
\[
 \sum_jx_j^2=\operatorname{Tr}(\omega Q)\le\|Q\|_{\rm op}=\left(\sum_jx_j^2\right)^{1/2},
\]
so $\sum_jx_j^2\le1$. Sum over all triples $T$. Every quartet occurs exactly four times, giving \eqref{alltemp:eq:elementarymoments}. Taking the supremum over states in
\[
 \left|\operatorname{Tr}\left(\omega\sum_Iu_I\Gamma_I\right)\right|
 \le\|\vec u\|_2\left(\sum_I\operatorname{Tr}(\omega\Gamma_I)^2\right)^{1/2}
\]
proves \eqref{alltemp:eq:elementaryop}. Finally, $\sigma_n^2 B_n=(n-1)(n-2)/n^2\le1$.
\end{proof}

The following comparison uses the commutation-index and Gaussian-concentration argument of
\cite[Appendix~C.2]{anschuetz2025strongly}, with the elementary bound above in place of their
sharper commutation-index estimate.
\begin{lemma}[Partition function concentrates about its mean]\label{alltemp:lem:concentration}
For every $\beta>0$ and $\log n>\beta^2/8$,
\begin{equation}\label{alltemp:eq:concentration}
 \Pr_{\vec \bsg}\left[Z(\beta,\vec \bsg)<\frac{\E_{\vec \bsg}Z(\beta,\vec \bsg)}n\right]
 \le\exp\left[-\frac{2(\log n-\beta^2/8)^2}{\beta^2}\right].
\end{equation}
For $\log n\ge2\beta^2$, the right-hand side is at most $1/n$.
\end{lemma}
\begin{proof}
For $\Phi(\vec g)=\log Z(\beta,\vec g)$ and $\rho=\rho_\beta(\vec g)$,
Lemma~\ref{lem:first_derivative_identity} gives $\partial_{g_I}\Phi=-\beta\sigma_n\operatorname{Tr}(\rho\Gamma_I)$. Lemma~\ref{alltemp:lem:elementary} therefore implies
\[
 \|\nabla\Phi(\vec g)\|_2^2\le\beta^2\sigma_n^2 B_n/4\le\beta^2/4.
\]
Apply Lemma~\ref{alltemp:lem:gaussian-concentration} with $L=\beta/2$. Jensen's inequality and the moment-generating-function bound yield
\[
 0\le\log Z_\mathrm{ann}(\beta)-\E_{\vec \bsg}\Phi(\vec \bsg)\le\beta^2/8.
\]
Hence $Z(\beta,\vec \bsg)<Z_\mathrm{ann}(\beta)/n$ implies
$\Phi(\vec \bsg)-\E_{\vec \bsg}\Phi(\vec \bsg)<-(\log n-\beta^2/8)$, and the lower-tail bound
in \eqref{alltemp:eq:gaussian-concentration-tool} proves \eqref{alltemp:eq:concentration}.
\end{proof}

\subsection{Bogoliubov--Kubo--Mori covariance}
The Hessian of the log-partition function has a natural interpretation as a thermal covariance: in
the noncommutative setting, the relevant quantity is the Bogoliubov--Kubo--Mori (BKM) covariance
\cite{petz1993bogoliubov}. We first record this standard identity and then show that the BKM variance is lower bounded, up to a polynomial
factor, by the ordinary thermal variance. This reduces the strong-convexity problem to the ordinary
variance lower bound proved in the next section. To this end, for a full-rank state $\rho$ define
the BKM linear map,
\[
	\Omega_\rho(B) := \int_0^1 \rho^s B \rho^{1-s} \dd s\,.
\]
This leads to the BKM inner product on observables $\langle A, B\rangle_{\rho, \BKM} :=
\Tr[A^\dagger \Omega_\rho(B)] = \int_0^1 \Tr[\rho^s A^\dagger \rho^{1-s} B ] \dd s$. For Hermitian
$A, B$, let $A_0 = A - \Tr[\rho A] I$, $B_0 = B - \Tr[\rho B] I$ be the centered observables.
Then, the BKM variance and covariance are,
\begin{align}
	&\mathrm{Cov}^\BKM_\rho(A,B) := \langle A_0, B_0 \rangle_{\rho, \BKM} = \int_0^1\Tr[\rho^s A\rho^{1-s}B]\dd s
       -\Tr[\rho A]\Tr[\rho B]\,,\\
       &\Var_\rho^\BKM(A) := \mathrm{Cov}^\BKM_\rho(A,A)\,.
\end{align}
Note if $A, B$ commute with $\rho$, then $\mathrm{Cov}^\BKM_\rho(A, B)  =\Tr(\rho AB) -
\Tr(\rho A)\Tr(\rho B)$, so it reduces to the ordinary covariance in the classical or commuting
case.

\begin{restatable}[Hessian equals BKM covariance]{lemma}{hessianbkm}
\label{lem:hessian-bkm}
Fix a real $\beta$. Let $A_{\vec u}=\sum_I u_I\Gamma_I$ and recall $\Phi(\vec g)=\log Z(\beta,\vec g)$ with
$H(\vec g)=\sigma_n\sum_I g_I\Gamma_I$. For all $I,J$,
\begin{equation}
\partial_{g_I}\partial_{g_J}\Phi(\vec g)
=\beta^2\sigma_n^2\operatorname{Cov}^{\BKM}_{\rho_\beta(\vec g)}(\Gamma_I,\Gamma_J).
\label{eq:hessian-bkm}
\end{equation}
where $\rho_\beta(\vec g)=e^{-\beta H(\vec g)}/\Tr[e^{-\beta H(\vec g)}]$. Equivalently,
\begin{equation}
  \vec{u}^\top\nabla^2\Phi(\vec g)\vec{u}
  =\beta^2\sigma_n^2\Var_{\rho_\beta(\vec g)}^{\BKM}(A_{\vec u}),
  \label{eq:hessian-bkm-direction}
\end{equation}
\end{restatable}
See \Cref{sec:deferred_proofs} for a proof. This lemma reduces the Hessian lower bound to a
lower bound on the BKM variance of the observable \(A_{\vec u}\). The estimates developed for SYK,
however, control the ordinary thermal variance \[ \Var_{\rho_\beta(\vec g)}(A_{\vec u}) =
\Tr\!\left[\rho_\beta(\vec g) \bigl(A_{\vec u}-\Tr[\rho_\beta(\vec g)A_{\vec u}]I\bigr)^2\right]. \]
We therefore next compare these two notions of variance. The following lemma shows that passing from
the BKM variance to the ordinary variance costs only a factor controlled by the spectral diameter of
\(H(\vec g)\).
\begin{lemma}[BKM bounds ordinary variance]
\label{lem:bkm-variance}
For any real $\beta$ and Hermitian $A$, if $\Delta_{\vec g}$ is the spectral diameter of $H(\vec g)$, then
\begin{equation}
  \Var_{\rho_\beta(\vec g)}^{\BKM}(A)
  \geq\frac{1}{1+|\beta|\Delta_{\vec g}}\Var_{\rho_\beta(\vec g)}(A).
  \label{eq:bkm-comparison}
\end{equation}
\end{lemma}

\begin{proof}
Diagonalize $H(\vec g)\vec e_a=\lambda_a\vec e_a$ in an orthonormal basis and put
$q_a=e^{-\beta \lambda_a}/(d_nZ(\beta,\vec g))$, so $\sum_aq_a=1$. With the logarithmic mean
\[
  L(p,q)=\int_0^1p^sq^{1-s}\dd s,
\]
write $\widetilde A=A-\Tr[\rho A]\Id$. The two variances are
\[
  \Var_\rho^{\BKM}(A)=\sum_{a,b}L(q_a,q_b)|\widetilde A_{ab}|^2,\qquad
  \Var_\rho(A)=\sum_{a,b}q_a|\widetilde A_{ab}|^2.
\]
Set $R=|\beta|\Delta_{\vec g}$. Since $q_b/q_a\geq e^{-R}$,
\[
  L(q_a,q_b)=q_a\int_0^1(q_b/q_a)^{1-s}\dd s
  \geq q_a\frac{1-e^{-R}}{R}\geq\frac{q_a}{1+R}.
\]
The value at $R=0$ is defined by continuity, and the last inequality is equivalent to
$e^R\geq1+R$.
Applying this term by term in the energy eigenbasis and summing over
\(a,b\) yields
\[
\Var^{\mathrm{BKM}}_{\rho_\beta(\vec g)}(A)
\geq
\frac{1}{1+|\beta|\Delta_{\vec g}}\,
\Var_{\rho_\beta(\vec g)}(A),
\]
which proves the lemma.
\end{proof}

\section{Wick expansion and Koteck\'y--Preiss}
\label{sec:wick_expansion_exposition}

The analysis involving the sample and time complexity will involve understanding how the random SYK partition functions fluctuate with the disorder.  Concretely, this requires controlling quantities such as
\[
    \E_{\vec \bsg_1,\vec \bsg_2}
    \big[
      Z(\beta_1,\vec \bsg_1)Z(\beta_2,\vec \bsg_2)
    \big],
\]
as well as closely related partition-function ratios.  Since the randomness enters only through the Gaussian couplings in the Hamiltonian, the natural first step is to expand the partition functions in powers of these couplings.  Wick's theorem then converts the resulting Gaussian moments into sums over pairings.  This produces a combinatorial expansion that can be controlled using Koteck\'y--Preiss. We discuss both of these below in more detail. More concretely, the basic objects that we are concerned with are
moments of the~form
\[
    M(\beta_1,\beta_2)
    :=
    \E_{\vec \bsg_1,\vec \bsg_2}
    \Big[
      Z(\beta_1,\vec \bsg_1)Z(\beta_2,\vec \bsg_2)
    \Big],
\]
where
$ Z(\beta,\vec g)=\Tr\!\left[e^{-\beta H(\vec g)}\right].$ 
We refer to the two copies of $Z$ as \emph{replicas}. By choosing
$\vec \bsg_1,\vec \bsg_2$ to be either independent or identical, the same
formalism captures both
\[
    \left|\E_{\vec \bsg} Z(\beta,\vec \bsg)\right|^2
    \qquad\text{and}\qquad
    \E_{\vec \bsg}\left|Z(\beta,\vec \bsg)\right|^2,
\]
with conjugate inverse temperatures when needed. We begin by Taylor expanding a single replica:
\begin{equation}
    \label{eq:z_taylor_expansion}
    Z(\beta,\vec g)
    =
    \sum_{r\geq0}
    \frac{(-\beta)^r\sigma_n^r}{r!}
    \sum_{I_1,\ldots,I_r\in\mc I_4}
    \left(\prod_{i=1}^r g_{I_i}\right)
    \Tr\!\left[\prod_{i=1}^r \Gamma_{I_i}\right].
\end{equation}
Expanding both replicas therefore reduces $M(\beta_1,\beta_2)$ to Gaussian
moments of the form
\[
    \E_{\vec \bsg_1,\vec \bsg_2}\!\left[
      \prod_{i=1}^{r_1}(\vec \bsg_1)_{I_i}
      \prod_{j=1}^{r_2}(\vec \bsg_2)_{J_j}
    \right].
\]
At this stage, the Wick--Isserlis theorem allows us to write the inner expectation as a sum over \emph{perfect matchings} over the ``slots''
$$
\mathsf{Slots} = (\{1\} \times [r_1]) \cup (\{2\} \times [r_2]).
$$ Denote the set of all perfect matchings as $\PM(\mathsf{Slots})$ and a single matching as $\pi$. Moreover, for each edge in $\pi$ (which corresponds to two Gaussian variables), we get a contribution of $\Cov((\vec \bsg_c)_I,(\vec \bsg_{c\prime})_J)$. Recalling that any pair $(\vec \bsg_i)_I, (\vec \bsg_i)_J$ are independent if $I \neq J$, this tells us that only non-zero contributions to the sum are those where both endpoints of a matching edge correspond to the \emph{same} Majorana term. Thus, we can rewrite the above sum as over labelings $\mc A : \pi \rightarrow \mc I_4$ from edges of a perfect matching over the slots to a single Majorana term $I \in \mc I_4$. This yields,
\begin{equation}
	\label{eq:informal_mass}
    M(\beta_1, \beta_2) = \sum_{r_1, r_2 \geq 0} C_{\beta_1, \beta_2, r_1, r_2} \sum_{\pi \in \PM(\mathsf{Slots})} \sum_{\mc A : \pi \rightarrow \mc I_4} \tau_1(\mc A) \tau_2 (\mc A) \prod_{\{(c,i),(c',j)\}\in\pi}\kappa_{c,c'},
\end{equation}
where 
$$ C_{\beta_1,\beta_2,r_1,r_2}
   =\frac{(-\beta_1)^{r_1}(-\beta_2)^{r_2}\sigma_n^{r_1+r_2}}{r_1!r_2!},$$
 $\tau_1(\mc A)$ and $\tau_2(\mc A)$ are the Majorana trace factors from 
the two replicas, while $\kappa_{c,c'}$ records the covariance between
disorder variables in replicas $c$ and $c'$. Every summand, which specifies a pair $(r_1, r_2)$, a
matching $\pi$, and a labeling $\mc A$ is called a \textbf{Wick diagram}, denoted $\mc D$.

Note that the labeling from edges to $\calI$ induces another graph. For the two-replica moment, give
an edge its quartic support in each trace color incident to it; two edges are connected when these
colored supports intersect; see \Cref{fig:wick-to-overlap} for the passage from ordered slots to
the overlap graph. Thus the ground set is $X=\{1,2\}\times[n]$. Together with the
observation that the normalized trace factors decompose over disconnected components, we can rewrite
the above as a product of \emph{connected} polymers. For two polymers $\mc P, \mc P'$, we write $\mc
P \sim \mc P'$ to indicate that they correspond to disjoint supports. A key characterization will be a
factorization of \Cref{eq:informal_mass} into polymer ``activities'' $\mathsf{Act}_{\beta_1,
\beta_2}(\mc P)$:
\[
	M(\beta_1, \beta_2) = \sum_{c\geq0}\frac1{c!} \sum_{\substack{\mc P_1, \dots, \mc P_c \\ \mc
	P_i \text{ connected}\\\mc P_i\sim\mc P_j\ (i<j)}} \prod_{i=1}^c \mathsf{Act}_{\beta_1,
	\beta_2}(\mc P_i)\,.
\]
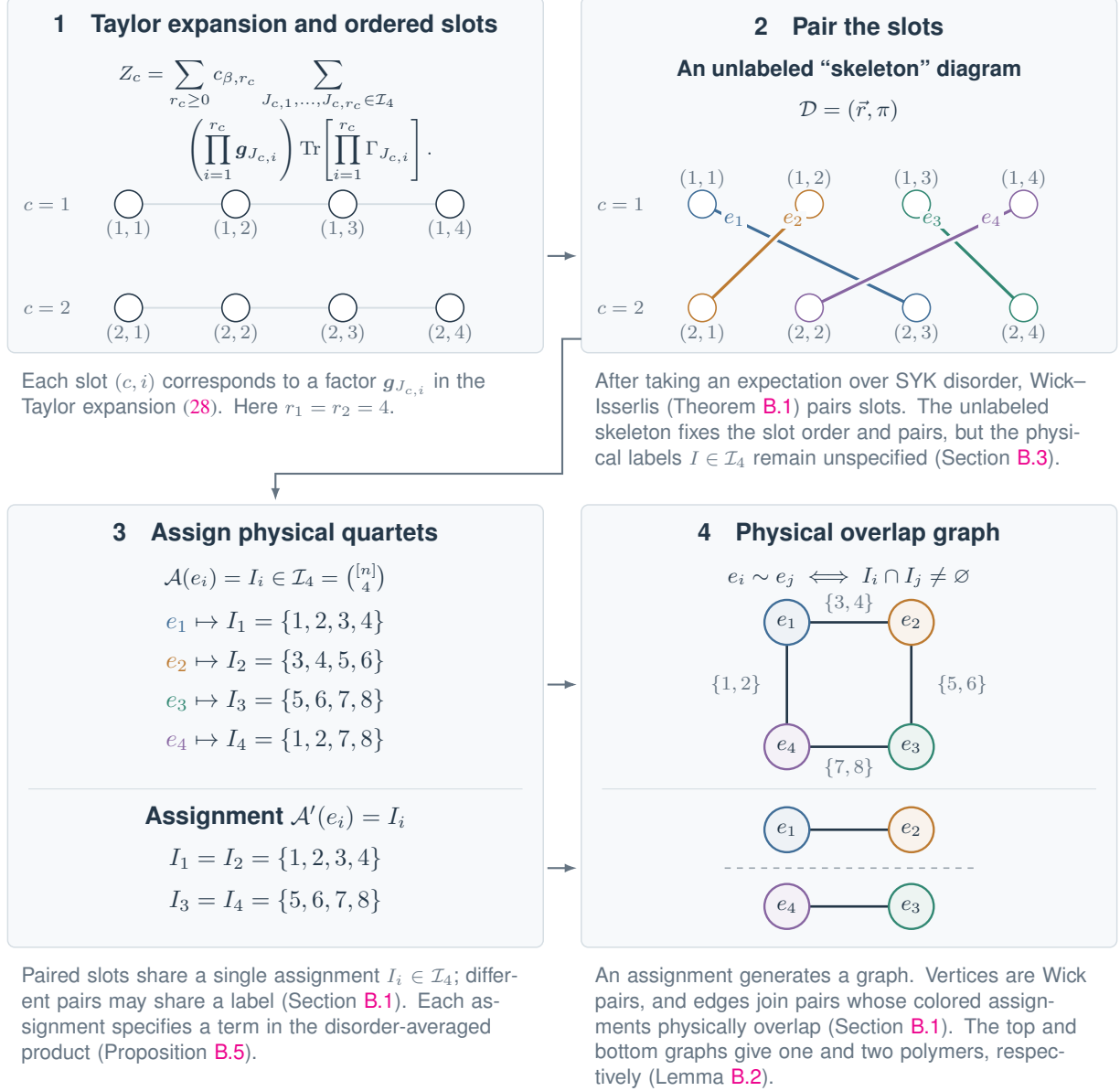
\begin{figure}[tbp]
\centering
\begingroup
\definecolor{wkInk}{HTML}{243748}
\definecolor{wkMuted}{HTML}{657482}
\definecolor{wkRule}{HTML}{D5DDE3}
\definecolor{wkPaper}{HTML}{F7F9FB}
\definecolor{wkBlue}{HTML}{406E98}
\definecolor{wkOchre}{HTML}{BD7B32}
\definecolor{wkGreen}{HTML}{328873}
\definecolor{wkPurple}{HTML}{8765A1}
\begin{tikzpicture}[
  x=1cm,y=1cm,
  font=\sffamily\fontsize{9}{11}\selectfont,
  text=wkInk,
  wkCard/.style={draw=wkRule,fill=wkPaper,rounded corners=3pt,line width=.55pt},
  wkTitle/.style={font=\sffamily\bfseries\fontsize{10}{12}\selectfont},
  wkNote/.style={text=wkMuted,font=\sffamily\fontsize{8.5}{10.5}\selectfont,align=center},
  wkCaption/.style={text=wkMuted,font=\sffamily\fontsize{8.5}{10.5}\selectfont,
    align=left,anchor=north,text width=7.35cm,inner sep=0pt},
  wkArrow/.style={-{Latex[length=1.8mm,width=1.3mm]},draw=wkMuted,line width=.8pt},
  wkSlot/.style={circle,draw=wkInk,fill=white,line width=.6pt,minimum size=4mm,inner sep=0pt},
  wkVertex/.style={circle,draw=wkInk,fill=white,line width=.9pt,minimum size=6.5mm,inner sep=0pt},
  wkEdge/.style={draw=wkInk,line width=1pt},
  wkPair/.style={line width=1.2pt,preaction={draw=wkPaper,line width=3.5pt}},
  wkSet/.style={font=\fontsize{10}{12}\selectfont}
]
\path[wkCard] (0,0) rectangle (7.75,-5.15);
\path[wkCard] (8.3,0) rectangle (16.05,-5.15);
\path[wkCard] (0,-7.35) rectangle (7.75,-13.75);
\path[wkCard] (8.3,-7.35) rectangle (16.05,-13.75);
\node[wkTitle] at (3.875,-.42) {1\quad Taylor expansion and ordered slots};
\node[wkTitle] at (12.175,-.42) {2\quad Pair the slots};
\node[wkTitle] at (3.875,-7.77) {3\quad Assign physical quartets};
\node[wkTitle] at (12.175,-7.77) {4\quad Physical overlap graph};
\draw[wkArrow] (7.81,-3.75) -- (8.24,-3.75);
\draw[wkArrow] (8.3,-4.95) -- (8.025,-4.95)
  -- (8.025,-6.9) -- (3.875,-6.9) -- (3.875,-7.29);

\node[anchor=north,font=\fontsize{8.5}{10.5}\selectfont] at (3.875,-.76) {$\displaystyle
  \begin{aligned}
  Z_c&=\sum_{r_c\geq0}c_{\beta,r_c}
    \sum_{J_{c,1},\ldots,J_{c,r_c}\in\mathcal I_4}\\[-1pt]
  &\qquad\left(\prod_{i=1}^{r_c}\bsg_{J_{c,i}}\right)
    \Tr\!\left[\prod_{i=1}^{r_c}\Gamma_{J_{c,i}}\right].
  \end{aligned}$};
\foreach \y/\c in {-3/1,-4.5/2} {
  \node[wkNote] at (.57,\y) {$c=\c$};
  \draw[draw=wkRule,line width=.7pt] (1.65,\y) -- (6.5,\y);
  \foreach \x/\i in {1.75/1,3.3/2,4.85/3,6.4/4} {
    \node[wkSlot] at (\x,\y) {};
    \node[wkNote] at (\x,\y-.37) {$(\c,\i)$};
  }
}
\node[wkCaption] at (3.875,-5.43)
  {Each slot $(c,i)$ corresponds to a factor $\bsg_{J_{c,i}}$ in the Taylor
   expansion~\eqref{eq:z_taylor_expansion}. Here $r_1=r_2=4$.};

\node[font=\sffamily\bfseries\fontsize{9}{11}\selectfont]
  at (12.175,-1.07) {An unlabeled ``skeleton'' diagram};
\node at (12.175,-1.65) {$\mathcal D=(\vec r,\pi)$};
\node[wkNote] at (8.87,-3) {$c=1$};
\node[wkNote] at (8.87,-4.5) {$c=2$};
\foreach \xt/\xb/\i/\col in {
  10.05/13.15/1/wkBlue,11.6/10.05/2/wkOchre,
  13.15/14.7/3/wkGreen,14.7/11.6/4/wkPurple} {
  \draw[wkPair,draw=\col] (\xt,-3) -- (\xb,-4.5);
}
\foreach \xt/\xb/\i/\col in {
  10.05/13.15/1/wkBlue,11.6/10.05/2/wkOchre,
  13.15/14.7/3/wkGreen,14.7/11.6/4/wkPurple} {
  \node[wkSlot,draw=\col] at (\xt,-3) {};
  \node[wkSlot,draw=\col] at (\xb,-4.5) {};
  \node[fill=wkPaper,inner sep=1pt,text=\col]
    at ({.85*\xt+.15*\xb},-3.225) {$e_{\i}$};
}
\foreach \x/\i in {10.05/1,11.6/2,13.15/3,14.7/4} {
  \node[wkNote] at (\x,-2.63) {$(1,\i)$};
  \node[wkNote] at (\x,-4.87) {$(2,\i)$};
}
\node[wkCaption] at (12.175,-5.43)
  {After taking an expectation over SYK disorder, Wick--Isserlis (\mbox{\Cref{thm:wick}}) pairs
	  slots.
   The unlabeled skeleton fixes the slot order and pairs, but the physical labels $I \in \calI$
   remain unspecified (\mbox{\Cref{sec:shared_counts}}).};

\node at (3.875,-8.4) {$\mathcal A(e_i)=I_i\in\mathcal I_4=\binom{[n]}4$};
\foreach \y/\i/\col/\sites in {
  -9.05/1/wkBlue/{1,2,3,4},
  -9.62/2/wkOchre/{3,4,5,6},
  -10.19/3/wkGreen/{5,6,7,8},
  -10.76/4/wkPurple/{1,2,7,8}} {
  \node[wkSet] at (3.875,\y) {$\textcolor{\col}{e_{\i}}\mapsto I_{\i}=\{\sites\}$};
}
\draw[draw=wkRule] (.3,-11.46) -- (7.45,-11.46);
\node[wkTitle] at (3.875,-11.88) {Assignment $\mathcal A'(e_i)=I_i$};
\node[wkSet] at (3.875,-12.48) {$I_1=I_2=\{1,2,3,4\}$};
\node[wkSet] at (3.875,-13.08) {$I_3=I_4=\{5,6,7,8\}$};
\node[wkCaption] at (3.875,-14.03)
  {Paired slots share a single assignment $I_i \in \calI$; different pairs may share
   a label (\mbox{\Cref{sec:wick_placement}}). Each assignment specifies a term in the disorder-averaged product (\mbox{\Cref{prop:polymer_counting}}).};

\node at (12.175,-8.4) {$e_i\sim e_j\ \Longleftrightarrow\ I_i\cap I_j\ne\varnothing$};
\draw[wkArrow] (7.81,-9.95) -- (8.24,-9.95);
\coordinate (wkG1) at (11.275,-9.05);
\coordinate (wkG2) at (13.075,-9.05);
\coordinate (wkG3) at (13.075,-10.85);
\coordinate (wkG4) at (11.275,-10.85);
\draw[wkEdge] (wkG1) -- (wkG2) -- (wkG3) -- (wkG4) -- cycle;
\node[wkVertex,draw=wkBlue,fill=wkBlue!9] at (wkG1) {$e_1$};
\node[wkVertex,draw=wkOchre,fill=wkOchre!9] at (wkG2) {$e_2$};
\node[wkVertex,draw=wkGreen,fill=wkGreen!9] at (wkG3) {$e_3$};
\node[wkVertex,draw=wkPurple,fill=wkPurple!9] at (wkG4) {$e_4$};
\node[wkNote] at (12.175,-8.77) {$\{3,4\}$};
\node[wkNote] at (13.81,-9.95) {$\{5,6\}$};
\node[wkNote] at (12.175,-11.16) {$\{7,8\}$};
\node[wkNote] at (10.54,-9.95) {$\{1,2\}$};
\draw[draw=wkRule] (8.6,-11.46) -- (15.75,-11.46);
\draw[wkArrow] (7.81,-12.605) -- (8.24,-12.605);
\draw[wkEdge] (11.275,-12.05) -- (13.075,-12.05);
\draw[wkEdge] (11.275,-13.16) -- (13.075,-13.16);
\node[wkVertex,draw=wkBlue,fill=wkBlue!9] at (11.275,-12.05) {$e_1$};
\node[wkVertex,draw=wkOchre,fill=wkOchre!9] at (13.075,-12.05) {$e_2$};
\node[wkVertex,draw=wkGreen,fill=wkGreen!9] at (13.075,-13.16) {$e_3$};
\node[wkVertex,draw=wkPurple,fill=wkPurple!9] at (11.275,-13.16) {$e_4$};
\draw[draw=wkMuted!65,line width=.6pt,dash pattern=on 2.5pt off 2.5pt]
  (10.375,-12.605) -- (13.975,-12.605);
\node[wkCaption] at (12.175,-14.03)
  {An assignment generates a graph. Vertices are Wick pairs, and edges join pairs whose
	 colored assignments physically overlap 
   (\mbox{\Cref{sec:wick_placement}}). The top and bottom graphs give
   one and two polymers, respectively (\mbox{\Cref{lem:component_placement}}).};

\end{tikzpicture}
\endgroup
\caption{From Taylor slots to physical overlap graphs, for two replicas
 $Z_c=Z(\beta,\vec\bsg)$ sharing the same random disorder, with $n\geq8$ and normalized trace
 $\Tr$. Every illustrated pair meets both replicas, so colored-support overlap
 is equivalent to $I_i\cap I_j\ne\varnothing$; the upper graph labels display these
 intersections.}
\label{fig:wick-to-overlap}
\end{figure}

The Koteck\'y--Preiss theorem is the tool that turns this polymer representation
into quantitative bounds.  Roughly speaking, once the total weight of polymers
touching any fixed site is sufficiently small, the full polymer expansion is
controlled by its connected components.  This allows us to bound ratios of
partition functions by summing only over the connected polymers that distinguish
the two quantities being compared.  The precise Koteck\'y--Preiss criterion,
the corresponding expansion, and the required polymer-counting estimates
are deferred to \Cref{sec:shared_wick_kp}. 
To see why this is useful, consider the second-moment ratio
\begin{equation}
    \label{eq:gas_derivation}
    \frac{\E_{\vec \bsg}|Z(\beta,\vec \bsg)|^2}
         {|\E_{\vec \bsg}Z(\beta,\vec \bsg)|^2}
    =
    \frac{M_1(\beta,\overline\beta)}
         {M_0(\beta,\overline\beta)}.
\end{equation}
Here $M_1$ corresponds to two replicas with the same disorder, while $M_0$
corresponds to independent disorder.  The two polymer expansions are identical
except for polymers that connect the two replicas.  Passing to the logarithm
isolates connected clusters, so all contributions that are common to $M_1$ and
$M_0$ cancel.  Consequently,
\[
    \frac{M_1(\beta,\overline\beta)}
         {M_0(\beta,\overline\beta)}
    =
    \exp\!\left(
      \text{total contribution of connected clusters joining the two replicas}
    \right).
\]
Koteck\'y--Preiss bounds precisely this remaining connected contribution.
Thus the second-moment estimate reduces to verifying the Koteck\'y--Preiss
condition and bounding the total weight of polymers that connect the two replicas. 
The formal setup and the required counts to bound \Cref{eq:gas_derivation} are given in \Cref{sec:shared_wick_kp}.

\section{Thermal estimates and zero-freeness}

The first set of ingredients concerns concentration of the partition function $Z(\beta,\vec g)=\Tr[e^{-\beta H(\vec g)}]$ around the annealed value $Z_{\mathrm{ann}}(\beta) \dfn \E_{\vec \bsg}
Z(\beta,\vec \bsg)$.  In
the proof, we want to replace the random denominator
\(Z(\beta,\vec g)\) in
$$
\Tr[\Gamma_I\rho_\beta]
=
\frac{\Tr[\Gamma_Ie^{-\beta H(\vec g)}]}{Z(\beta,\vec g)}
$$
by the deterministic quantity \(Z_{\mathrm{ann}}(\beta)\). Write $\frac{Z(\beta,\vec g)}{Z_{\mathrm{ann}}(\beta)} = 1+\Delta_\beta(\vec g),$ and use the first-order approximation
$\frac{1}{1+\Delta_\beta} \approx 1-\Delta_\beta$. This yields
\begin{equation}
	\label{eq:estimated_quantity}
	\Tr[\Gamma_I\rho_\beta] \approx \frac{\Tr[\Gamma_Ie^{-\beta H(\vec g)}]}
	{Z_{\mathrm{ann}}(\beta)} \Paren{1-\Delta_\beta(\vec g)}\,.
\end{equation}

To obtain our advertised result, we need the error in this approximation to be $\O(n^{-2})$. Since
this step introduces error \(\approx |\Delta_\beta|^2\), this sets the scale for the necessary
second-moment control. Prior to this work, the best estimate was given by \cite[Theorem
4]{zlokapa2026rigorous}:
\begin{restatable}[Second moment bound, \cite{zlokapa2026rigorous}]{lemma}{zlokapamomentbound}
  \label{lem:prior_moment}
  There are constants $\beta_\mathrm{prior}>0$ and $C_2<\infty$ such that
  $Z_\mathrm{ann}(\beta)\neq0$ on a neighborhood of
  $|\beta|\leq \beta_\mathrm{prior}$ and, for every complex
  $|\beta|\leq \beta_\mathrm{prior}$ and all sufficiently large $n$,
  \[
    \Delta_\beta(\vec g)
      :=\frac{Z(\beta,\vec g)}{Z_\mathrm{ann}(\beta)}-1,
    \qquad
    \E_{\vec \bsg}[|\Delta_\beta(\vec \bsg)|^2]
      \leq \frac{C_2}{n}.
  \]
\end{restatable}

Our first contribution is the following quadratically stronger estimate, which holds on a slightly
smaller constant disk of inverse temperatures.
\begin{restatable}[Improved second moment bound]{lemma}{momentbound}
  \label{lem:informal_moment_bound}
  Set the constant radius
  $\beta_\mathrm{improved}=\min\{\beta_\mathrm{prior}/2,\beta_\mathrm{thermal}\}$. Then, for every complex $|\beta|\leq \beta_\mathrm{improved}$,
  \[
    \E_{\vec \bsg}|\Delta_\beta(\vec \bsg)|^2
      \leq O(n^{-2}).
  \]
\end{restatable}
Its proof requires the two-replica Wick expansion together with the weighted Hermite estimates
developed for the time-efficient algorithm. See \Cref{fig:delta-proof} for a high-level view of the
proof. The details are deferred to \Cref{sec:thermal_kp_proof,sec:weighted_thermal_proof}.

The final ingredient is an analogous estimate for partition functions with an inserted
Majorana observable, which will hold at a smaller constant radius $\beta_{\rm thermal}$, set by
\Cref{cor:annealed_overview}.
\begin{cor}[Annealed expectation value]
  \label{cor:annealed_overview}
  For fixed real $|\beta|\leq \beta_\mathrm{thermal}$ and all sufficiently
  large $n$,
  \[
    \frac{\E_{\vec \bsg}
      \big|\Tr[\Gamma_Ie^{-\beta H(\vec \bsg)}]\big|}
         {Z_\mathrm{ann}(\beta)}
    \leq C\sigma_n|\beta|,
    \text{ for all } I\in\mc I_4.
  \]
\end{cor}
Both of the above estimates are stated formally in \Cref{lem:thermal_estimates} and putting them
together in \Cref{eq:estimated_quantity} will control the error of the algorithm.

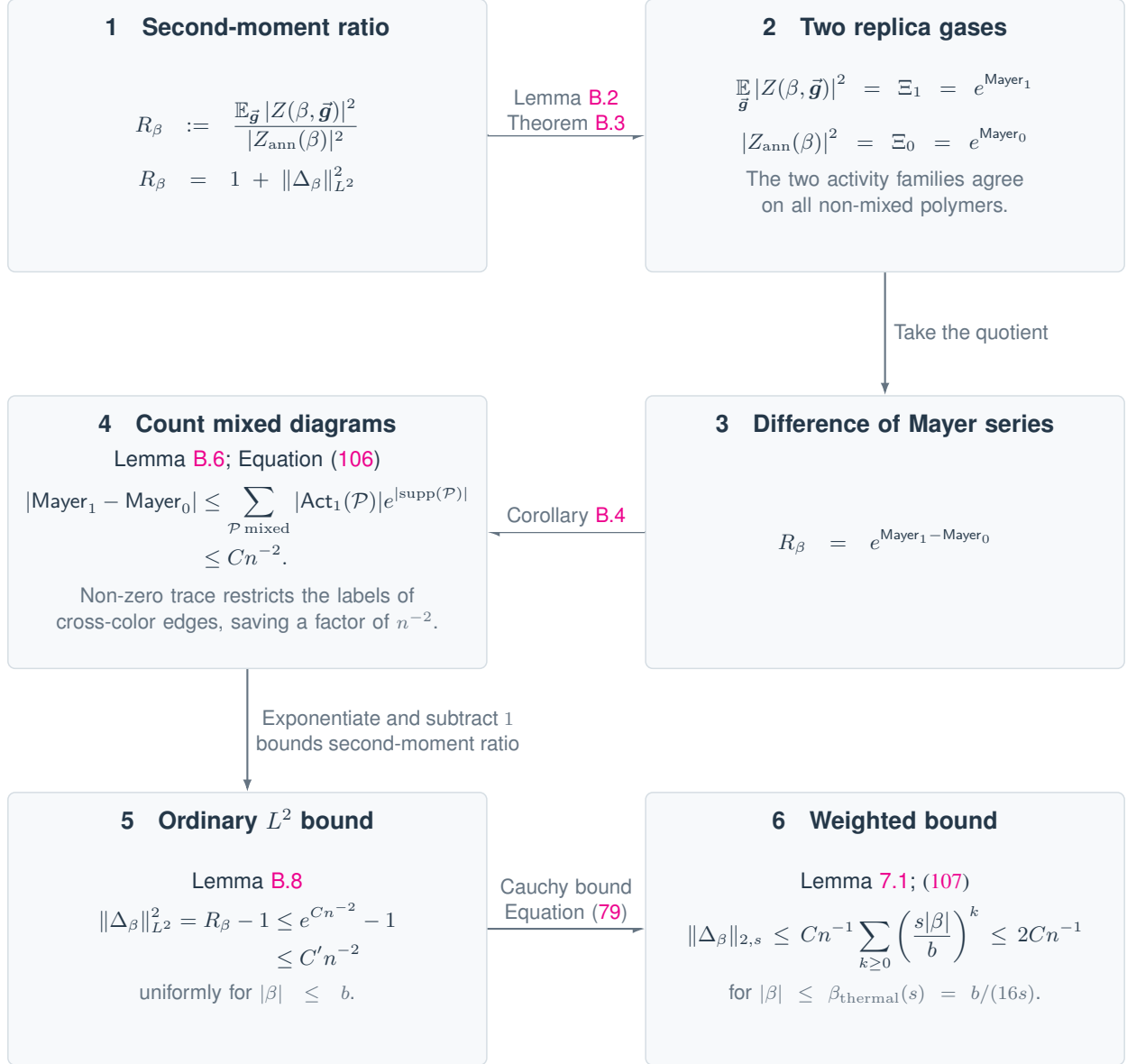
\begin{figure}[htbp]
\centering
\begingroup
\definecolor{mbInk}{HTML}{243748}
\definecolor{mbMuted}{HTML}{657482}
\definecolor{mbRule}{HTML}{D5DDE3}
\definecolor{mbPaper}{HTML}{F7F9FB}
\begin{tikzpicture}[
  font=\sffamily\fontsize{9}{11}\selectfont,
  text=mbInk,
  box/.style={
    draw=mbRule, fill=mbPaper, rounded corners=3pt, line width=.55pt,
    text width=.4\linewidth, minimum height=40mm,
    inner sep=6pt, align=center
  },
  title/.style={
    font=\sffamily\bfseries\fontsize{10}{12}\selectfont,
    text width=.4\linewidth, align=center,
    text height=9pt, text depth=3pt,
    anchor=north, inner sep=0pt, yshift=-6pt
  },
  body/.style={
    text width=.4\linewidth, align=center,
    anchor=center, inner sep=0pt, yshift=-4pt
  },
  arrow/.style={-{Latex[length=1.8mm,width=1.3mm]}, draw=mbMuted, line width=.8pt},
  annotation/.style={
    font=\sffamily\fontsize{8.5}{10.5}\selectfont, text=mbMuted, align=center
  },
  note/.style={
    annotation,
    text width=.13\linewidth, fill=white, inner sep=2pt
  },
  node distance=18mm and .14\linewidth
]

\node[box] (ratio) {};
\node[title] at (ratio.north) {1\quad Second-moment ratio};
\node[body] at (ratio.center) {
  $\displaystyle
    R_\beta:=
    \frac{\E_{\vec\bsg}|Z(\beta,\vec\bsg)|^2}
         {|Z_\mathrm{ann}(\beta)|^2}$\\[5pt]
  $\displaystyle R_\beta=1+\|\Delta_\beta\|_{L^2}^2$
};

\node[box, right=of ratio] (mayer) {};
\node[title] at (mayer.north) {2\quad Two replica gases};
\node[body] at (mayer.center) {
  $\displaystyle
    \E_{\vec\bsg}|Z(\beta,\vec\bsg)|^2
      =\Xi_1=e^{\mathsf{Mayer}_1}$\\[5pt]
  $\displaystyle
    |Z_\mathrm{ann}(\beta)|^2
    =\Xi_0=e^{\mathsf{Mayer}_0}$\\[5pt]
  {\color{mbMuted}\sffamily\fontsize{8.5}{10.5}\selectfont
  The two activity families agree\\
  on all non-mixed polymers.}
};

\node[box, below=of mayer] (difference) {};
\node[title] at (difference.north) {3\quad Difference of Mayer series};
\node[body] at (difference.center) {
  $\displaystyle
    R_\beta=e^{\mathsf{Mayer}_1-\mathsf{Mayer}_0}$
};

\node[box, left=of difference] (mixed) {};
\node[title] at (mixed.north) {4\quad Count mixed diagrams};
\node[body] at (mixed.center) {
  \Cref{lem:shared_counts}; \Cref{eq:global_mixed_mass}\\[5pt]
  $\displaystyle
    \begin{aligned}
	    |\mathsf{Mayer}_1-\mathsf{Mayer}_0| &\leq \sum_{\mc P\,\mathrm{mixed}}
	    |\mathsf{Act}_1(\mc P)|e^{|\supp(\mc P)|}\\
	    &\leq Cn^{-2}.
    \end{aligned}$\\[5pt]
  {\color{mbMuted}\sffamily\fontsize{8.5}{10.5}\selectfont
  Non-zero trace restricts the labels of\\
  cross-color edges, saving a factor of $n^{-2}$.}
};

\node[box, below=of mixed] (ordinary) {};
\node[title] at (ordinary.north) {5\quad Ordinary $L^2$ bound};
\node[body] at (ordinary.center) {
  \Cref{lem:ordinary_replica}\\[5pt]
  $\displaystyle
    \begin{aligned}
      \|\Delta_\beta\|_{L^2}^2
        =R_\beta-1 &\leq e^{Cn^{-2}}-1\\
        &\leq C'n^{-2}
    \end{aligned}$\\[4pt]
  {\color{mbMuted}\sffamily\fontsize{8.5}{10.5}\selectfont
  uniformly for $|\beta|\leq b$.}
};

\node[box, right=of ordinary] (weighted) {};
\node[title] at (weighted.north) {6\quad Weighted bound};
\node[body] at (weighted.center) {
  \Cref{lem:thermal_estimates}; \eqref{eq:weighted_upgrade}\\[5pt]
  $\displaystyle
    \|\Delta_\beta\|_{2,s}
      \leq Cn^{-1}\sum_{k\geq0}
        \left(\frac{s|\beta|}{b}\right)^k
      \leq 2Cn^{-1}$\\[4pt]
  {\color{mbMuted}\sffamily\fontsize{8.5}{10.5}\selectfont
  for $|\beta|\leq\beta_\mathrm{thermal}(s)=b/(16s)$.}
};

\draw[arrow] (ratio.east) --
  node[note, above] {
    \Cref{lem:component_placement}\\
    \Cref{thm:kp}
  } (mayer.west);

\draw[arrow] (mayer.south) --
  node[annotation, right] {Take the quotient}
  (difference.north);

\draw[arrow] (difference.west) --
  node[note, above] {\Cref{cor:kp_comparison}}
  (mixed.east);

\draw[arrow] (mixed.south) --
  node[annotation, right] {
    Exponentiate and subtract $1$\\
    bounds second-moment ratio
  } (ordinary.north);

\draw[arrow] (ordinary.east) --
  node[note, above] {
    Cauchy bound\\
    \Cref{eq:cauchy_lp_coefficients}
  } (weighted.west);

\end{tikzpicture}
\endgroup

\caption{Proof of the $\Delta_\beta$ bound in
\Cref{lem:thermal_estimates}. Here $s\geq1$,
$b=16\beta_\mathrm{thermal}$. The subscripts
$1$ and $0$ correspond to shared and independent disorder parameters. The Mayer series
$\mathsf{Mayer}_1$ differs only from $\mathsf{Mayer}_0$ in cross-color Wick pairs; this corresponds
to mixed polymers. The counting estimates also verify the Koteck\'y--Preiss hypotheses on the disk
$|\beta|\leq b$.}
\label{fig:delta-proof}
\end{figure}

\section{Sample efficiency at all temperatures}
\label{sec:sample_efficiency}
\label{sec:sample_all_temperatures}

In this section, we prove our sample-efficient learning result.

\samplealltemp*

Similar to~\cite{anshu2020sample}, we will use that the marginals $\mu_{\beta, I} = \Tr[\rho_\beta(\vec g)
\Gamma_I]$, for all $I \in \mc I_4$ form a \emph{sufficient statistic} for the Hamiltonian learning
problem. In their work, the authors show the following theorem that we invoke as a black box.
\begin{lemma}[Matching local marginals implies global equivalence, \cite{anshu2020sample}]
  Fix a known real $\beta\neq0$ and $\vec g,\vec h\in\R^m$. Consider two Gibbs states
  \[
    \rho_\beta(\vec g)=\frac{1}{Z(\beta,\vec g)}\cdot e^{-\beta\sigma_n\sum_I g_I\Gamma_I},
    \qquad
    \rho_\beta(\vec h)=\frac{1}{Z(\beta,\vec h)}\cdot e^{-\beta\sigma_n\sum_I h_I\Gamma_I}
  \]
  such that $\Tr[\rho_\beta(\vec g)\Gamma_I]=\Tr[\rho_\beta(\vec h)\Gamma_I]$
  for all $I\in\mc I_4$. Then $\rho_\beta(\vec g)=\rho_\beta(\vec h)$ and
  $\vec g=\vec h$.
\end{lemma}
This suggests that given the marginals $\mu_{\beta, I}$, one should solve the following maximum
entropy problem:
\begin{align*}
	\max_\rho & \quad -\Tr[\rho\log\rho]\\
	\text{s.t.}& \quad\Tr[\rho \Gamma_I] = \mu_{\beta, I}\,, \quad \text{ for all } I \in \mc I_4\\
	& \quad \rho \succeq 0\,, \quad \Tr[\rho] = 1\,.
\end{align*}
Observe that this entropy differs from the physical von Neumann entropy of $\rho/d_n$ by
the constant $-\log d_n$, which does not affect the optimizer. The optimum of this program can be shown to be the Gibbs state $\rho_\beta$ for some choice of
parameters $\vec h$. For the Hamiltonian parameter learning task, it will be more helpful to consider the dual of this program,
\[
	\vec g = \mathop{\mathrm{argmin}}_{\vec h = (h_I)_{I \in \mc I_4}} \mc L(\vec h)\,, \quad
	\mc L(\vec h) = \log Z(\beta, \vec h) +
	\beta\sigma_n \cdot \sum_{I \in \mc I_4} h_I \mu_{\beta, I}\,.
\]
In principle, we can learn the parameters of the Hamiltonian by solving this program. However, we
have access to \emph{noisy estimates} $\mu_{\beta,\mathrm{meas},I}$, rather than the true marginals above. This yields
the approximate version of the program, restricted to the known compact, convex set
$K_n:=\{\vec h\in\R^m:\|\vec h\|_2\leq2\sqrt m\}$,
\begin{equation}
  \vec g_\mathrm{meas}\in\mathop{\mathrm{argmin}}_{\vec h\in K_n}\mc L_\mathrm{meas}(\vec h),
  \qquad
  \mc L_\mathrm{meas}(\vec h)=\log Z(\beta,\vec h)
       +\beta\sigma_n\sum_{I\in\mc I_4}h_I\mu_{\beta,\mathrm{meas},I}.
  \label{eq:estimator}
\end{equation}
As discussed in the introduction, a lower bound on the Hessian of the log-partition function
implies stability of the maximum-entropy estimator, and hence an upper bound on
\(\|\vec g-\vec g_{\mathrm{meas}}\|_2\).  Thus the remaining task is to prove
a lower bound on
\[
    \nabla^2 \log Z(\beta,\vec g).
\]
 To prove strong convexity, we closely follow the structure of \cite{anshu2020sample}. We start by
showing that the Hessian of the log partition function is lower bounded by the ordinary variance.
Rather than constructing quasi-local operators as in the prior work, we show this directly
by passing through the Bogoliubov--Kubo--Mori variance; this argument is quite straightforward. Most of
the technical work is in proving the lower bound on the ordinary variance. We will use an analogue
of the \cite{arad2016connecting} result, although quite different techniques are needed to handle
the SYK model.

\subsection{Strong convexity suffices}

\begin{restatable}[Dimension-independent Hessian Lipschitz bound]{lemma}{hessianlipschitz}
\label{alltemp:lem:hessian-lipschitz}
Fix $\beta>0$ and let $\Phi(\vec h)=\log Z(\beta,\vec h)$. For all $\vec h,\vec h'\in\mathbb R^m$,
\begin{equation}\label{alltemp:eq:lipschitz}
 \|\nabla^2\Phi(\vec h)-\nabla^2\Phi(\vec h')\|_{\rm op}
 \le L_\beta\|\vec h-\vec h'\|_2,
 \qquad L_\beta=3\beta^3/4.
\end{equation}
Consequently, if $\nabla^2\Phi(\vec g)\succeq\lambda I$ for some $\lambda>0$, then
\begin{equation}\label{alltemp:eq:local-curvature}
 \nabla^2\Phi(\vec h)\succeq\frac\lambda2 I
 \quad\text{whenever}\quad
 \|\vec h-\vec g\|_2\le\frac{\lambda}{2L_\beta}.
\end{equation}
\end{restatable}

The proof is deferred to the appendix.

\begin{theorem}[Error bound from local strong convexity]
\label{sampleall:thm:error-from-convexity}
Fix $\beta>0$. Let $K\subseteq\mathbb R^m$ be compact and convex and let $\vec g_\star\in K$ be the true parameter. Suppose
\[
 \nabla^2\Phi(\vec g_\star)\succeq\lambda I,\qquad\lambda>0.
\]
Set $L_\beta=3\beta^3/4$, $R=\lambda/(2L_\beta)$, and let $\mu_{\beta,I}=\operatorname{Tr}(\rho_\beta(\vec g_\star)\Gamma_I)$ be the true quartic moments. For $\varepsilon>0$, suppose the estimates $\vec\mu_{\beta,\mathrm{meas}}$ satisfy
\begin{equation}\label{sampleall:eq:expectation-accuracy}
 \|\vec\mu_{\beta,\mathrm{meas}}-\vec\mu_\beta\|_2\le\delta_\varepsilon
 :=\min\Brac{\frac{\lambda\varepsilon}{8\beta\sigma_n}, \frac{\lambda^2}{16\beta\sigma_nL_\beta}}\,.
\end{equation}
Define the empirical objective
\[
 \mc L_\mathrm{meas}(\vec h)=\Phi(\vec h)+\beta\sigma_n \vec h\cdot\vec\mu_{\beta,\mathrm{meas}}.
\]
Then any $\vec g_\mathrm{meas}\in K$ satisfying
\begin{equation}\label{sampleall:eq:optimization-accuracy}
 \mc L_\mathrm{meas}(\vec g_\mathrm{meas})\le\min_{\vec h\in K}\mc L_\mathrm{meas}(\vec h)+\eta,
 \qquad 0\le\eta\le\min\{\lambda R^2/16,\lambda\varepsilon^2/16\},
\end{equation}
satisfies $\|\vec g_\mathrm{meas}-\vec g_\star\|_2\le\varepsilon$. In particular, the conclusion holds for any exact minimizer.
\end{theorem}
\begin{proof}
Lemma~\ref{lem:hessian-bkm} gives $\nabla^2\Phi(\vec h)\succeq0$ for every $\vec h$, so $\mc L_\mathrm{meas}$ is convex. By Lemma~\ref{alltemp:lem:hessian-lipschitz},
\[
 \nabla^2\Phi(\vec h)\succeq\frac\lambda2I
 \quad\text{whenever}\quad \|\vec h-\vec g_\star\|_2\le R.
\]
The derivative identity in Lemma~\ref{lem:first_derivative_identity} gives
\[
 \nabla \Phi(\vec g_\star)=-\beta\sigma_n\vec\mu_\beta,
 \qquad \nabla\mc L_\mathrm{meas}(\vec g_\star)=\vec\xi
 :=\beta\sigma_n(\vec\mu_{\beta,\mathrm{meas}}-\vec\mu_\beta),
 \qquad\|\vec\xi\|_2\le\beta\sigma_n\delta_\varepsilon.
\]
For $\vec v=\vec h-\vec g_\star$ with $\|\vec v\|_2\le R$, Taylor's formula with integral remainder implies
\begin{equation}\label{sampleall:eq:localobjective}
 \begin{aligned}
 \mc L_\mathrm{meas}(\vec h)-\mc L_\mathrm{meas}(\vec g_\star)
 &=\vec\xi\cdot \vec v+\int_0^1(1-t)\vec v^\top\nabla^2\Phi(\vec g_\star+t\vec v)\vec v\,dt\\
 &\ge-\beta\sigma_n\delta_\varepsilon\|\vec v\|_2+\frac\lambda4\|\vec v\|_2^2.
 \end{aligned}
\end{equation}
Since $\beta\sigma_n\delta_\varepsilon\le\lambda^2/(16L_\beta)=\lambda R/8$, every feasible point at distance $R$ has objective gap at least $\lambda R^2/8$.

Because $\vec g_\star\in K$, \eqref{sampleall:eq:optimization-accuracy} gives $\mc L_\mathrm{meas}(\vec g_\mathrm{meas})-\mc L_\mathrm{meas}(\vec g_\star)\le\eta$. Suppose $d:=\|\vec g_\mathrm{meas}-\vec g_\star\|_2\ge R$. Set $t=R/d$ and $\vec h_R=(1-t)\vec g_\star+t\vec g_\mathrm{meas}$. Then $\vec h_R\in K$ and $\|\vec h_R-\vec g_\star\|_2=R$. Convexity gives
\[
 \mc L_\mathrm{meas}(\vec h_R)-\mc L_\mathrm{meas}(\vec g_\star)
 \le t\bigl(\mc L_\mathrm{meas}(\vec g_\mathrm{meas})-\mc L_\mathrm{meas}(\vec g_\star)\bigr)
 \le t\eta\le\eta<\lambda R^2/8,
\]
contradicting the boundary gap. Therefore $d<R$. If $d>\varepsilon$, \eqref{sampleall:eq:localobjective} and $\beta\sigma_n\delta_\varepsilon\le\lambda\varepsilon/8$ give
\[
 \mc L_\mathrm{meas}(\vec g_\mathrm{meas})-\mc L_\mathrm{meas}(\vec g_\star)
 \ge\frac\lambda4d^2-\frac{\lambda\varepsilon}8d
 >\frac\lambda8d^2>\frac{\lambda\varepsilon^2}8>\eta,
\]
again a contradiction. Thus $\|\vec g_\mathrm{meas}-\vec g_\star\|_2\le\varepsilon$.
\end{proof}

\subsection{Strong convexity}
\label{sec:sample_all_convexity}

Write $B_n=\binom n3$, and let $U_A=\Gamma_A$ denote the ordered Majorana product for $A\subseteq[n]$. Each $U_A$ is unitary, and its adjoint differs from it only by a phase. All states below use the normalized trace. The following comparison turns a bound on conjugated fractional moments into a lower bound on ordinary variance.

\begin{lemma}[Fractional-moment variance comparison]\label{sampleall:lem:variance}
For any positive definite $\rho$ with $\operatorname{Tr}(\rho)=1$ and any $p>1$, define
\[
 R_{p,A}=\operatorname{Tr}\!\left((U_A\rho U_A^\dagger)^p\rho^{1-p}\right),\qquad
 \overline R_p=\frac1{16m}\sum_{|S|=4}\sum_{A\subseteq S}R_{p,A}.
\]
Then for every real coefficient vector $\vec u$ and $V=\sum_Iu_I\Gamma_I$,
\begin{equation}\label{sampleall:eq:variance}
 \operatorname{Var}_\rho(V)\ge
 \left(B_n m^p\overline R_p\right)^{-1/(p-1)}\left\lVert \vec u\right\rVert_2^2.
\end{equation}
\end{lemma}
\begin{proof}
It suffices to consider $\|\vec u\|_2=1$. Put $V_0=V-\operatorname{Tr}(\rho V)I$. For a quartet $I$, Majorana conjugation gives $U_A\Gamma_IU_A^\dagger=(-1)^{|A\cap I|}\Gamma_I$. Thus for each quartet $S$,
\[
 \frac1{16}\sum_{A\subseteq S}(-1)^{|A|+|A\cap I|}
 =\frac1{16}\prod_{a\in S}\left(1-(-1)^{\mathbf1_{\{a\in I\}}}\right)
 =\mathbf1_{\{I=S\}}.
\]
The same signed sum annihilates the identity term in $V_0$, so
\begin{equation}\label{sampleall:eq:filter}
 \frac1{16}\sum_{A\subseteq S}(-1)^{|A|}U_AV_0U_A^\dagger=u_S\Gamma_S.
\end{equation}
For any sixteen self-adjoint matrices $X_1,\ldots,X_{16}$,
\[
 16\sum_{j=1}^{16}X_j^2-\left(\sum_{j=1}^{16}X_j\right)^2
 =\sum_{j<\ell}(X_j-X_\ell)^2\succeq0.
\]
Apply this to $X_A=(-1)^{|A|}U_AV_0U_A^\dagger$, use \eqref{sampleall:eq:filter} and $\Gamma_S^2=I$, and divide by $16^2$. We obtain
\[
 u_S^2I\preceq\frac1{16}\sum_{A\subseteq S}U_AV_0^2U_A^\dagger.
\]
Taking the trace against $\rho$ yields
\[
 u_S^2\le\frac1{16}\sum_{A\subseteq S}\operatorname{Tr}(\rho_A V_0^2),
 \qquad \rho_A=U_A\rho U_A^\dagger.
\]
Here conjugation by $U_A$ and $U_A^\dagger$ is the same.

Let $q=p/(p-1)$. For any positive $B$, Lemmas~\ref{alltemp:lem:schatten-holder} and~\ref{alltemp:lem:alt} give
\begin{equation}\label{sampleall:eq:holder}
 \operatorname{Tr}(\rho_A B)\le R_{p,A}^{1/p}\bigl(\operatorname{Tr}(\rho B^q)\bigr)^{1/q}.
\end{equation}
Explicitly, apply H\"older to $X=\rho^{-1/(2q)}\rho_A\rho^{-1/(2q)}$ and $Y=\rho^{1/(2q)}B\rho^{1/(2q)}$. Applying \eqref{alltemp:eq:alt} with $(D,E,r)=(\rho_A,\rho^{-1/q},p)$ and $(B,\rho^{1/q},q)$ gives $\operatorname{Tr}(X^p)\le R_{p,A}$ and $\operatorname{Tr}(Y^q)\le\operatorname{Tr}(\rho B^q)$, respectively, using $-p/q=1-p$. All negative powers are well defined because $\rho$ has full rank.

Take $B=V_0^2$ and $v=\operatorname{Tr}(\rho B)$. By Lemma~\ref{alltemp:lem:elementary}, $\|V_0\|_{\rm op}\le2\|V\|_{\rm op}\le\sqrt{B_n}$, so $\operatorname{Tr}(\rho B^q)\le B_n^{q-1}v$. Summing the preceding bounds over $S$ and using concavity of $x^{1/p}$ gives
\[
 1\le m\overline R_p^{1/p}B_n^{1/p}v^{(p-1)/p}.
\]
Rearranging proves \eqref{sampleall:eq:variance}.
\end{proof}
This quantity is a Petz R\'enyi power \cite{petz1986quasi}; see also \cite{knill2023quantum}
for this terminology. Letting $\gamma = \rho_\beta/\tr[\rho_\beta]$ and
$\gamma_A = U_A \gamma U_A^\dagger$, it is
\[
	\mc P_p(\gamma_A \mid \gamma) \dfn \tr[\gamma_A^p \gamma^{1-p}]
	= R_{p,A}\,.
\]
We control this quantity through disorder-averaged partition functions.

The above lemma says that to get a bound on the variance, it suffices to bound $R_{p, A} =
\Tr\Brac{(U_A \rho U_A^\dagger)^p \rho ^{1-p}}$. Define
\begin{equation}\label{alltemp:eq:ZA}
	Z(\beta,\vec g) = \Tr[e^{-\beta H(\vec g)}],\qquad Z_A(\beta,s,\vec g) = \Tr\Brac{e^{-(\beta+s)H(\vec g)} e^{s
	U_A H(\vec g) U_A^\dagger}}\,.
\end{equation}

\begin{theorem}[Perturbed partition function bound]\label{alltemp:thm:quench}
For every $\beta\ge0$, every $|A|\le4$, and every real
\begin{equation}\label{alltemp:eq:srange}
 |s|\le s_\beta:=\frac1{8e(1+\beta)},
\end{equation}
one has
\begin{equation}\label{alltemp:eq:main}
 \frac{\E_{\vec \bsg} Z_A(\beta,s,\vec \bsg)}{\E_{\vec \bsg} Z(\beta,\vec \bsg)}\le\frac83.
\end{equation}
The constants and the allowed displacement are independent of $n$.
\end{theorem}
See the appendix for the full proof. We can get the desired bounds on the variance and show strong
convexity of the log partition function as a corollary. To relate convexity of the log partition
function to the variance, we pass through the BKM variance.

\begin{cor}[All-temperature variance lower bound]\label{sampleall:cor:hessian}
Fix $\beta>0$, set $s=s_\beta=1/[8e(1+\beta)]$ and $p=1+s/\beta$, and let $\vec \bsg$ have independent standard Gaussian coordinates. Put $H=H(\vec \bsg)$ and $\rho=\rho_\beta(\vec \bsg)$. With probability $1-O_\beta(n^{-1})$, simultaneously for every real $\vec u$,
\begin{align}
 \operatorname{Var}_\rho\!\left(\sum_Iu_I\Gamma_I\right)
 &\ge n^{-4-9\beta/s}\left\lVert \vec u\right\rVert_2^2,\label{sampleall:eq:finalvar}\\
 \nabla^2\Phi(\vec \bsg)&\succeq c_\beta n^{-r_\beta}I,
 \qquad r_\beta=8+72e\beta(1+\beta),\label{sampleall:eq:hessian}
\end{align}
where $c_\beta=6\beta^2/(1+2\beta)$.
\end{cor}
\begin{proof}
Cyclicity and conjugation by $U_A^\dagger$ being the same as conjugation by $U_A$ imply
\[
 R_{p,A}=\frac{\operatorname{Tr}(e^{-(\beta+s)H}e^{sU_AHU_A^\dagger})}{Z(\beta,\vec \bsg)}
 =\frac{Z_A(\beta,s,\vec \bsg)}{Z(\beta,\vec \bsg)}.
\]
Let $\mathcal E=\{Z(\beta,\vec \bsg)\ge\E_{\vec \bsg}Z(\beta,\vec \bsg)/n\}$. Theorem~\ref{alltemp:thm:quench} and positivity give
\[
 \E_{\vec \bsg}[\mathbf1_{\mathcal E}\overline R_p]\le\frac{8n}3.
\]
Markov's inequality and Lemma~\ref{alltemp:lem:concentration} imply $\overline R_p\le8n^2/3$ except with probability at most $2/n$, for $\log n\ge2\beta^2$. Using $B_n\le n^3/6$ and $m\le n^4$,
\[
 B_n m^p\overline R_p\le\frac49 n^{4p+5}\le n^{4p+5}.
\]
Lemma~\ref{sampleall:lem:variance} therefore gives
\[
 \operatorname{Var}_\rho\left(\sum_Iu_I\Gamma_I\right)
 \ge n^{-(4p+5)/(p-1)}\|\vec u\|_2^2
 =n^{-4-9\beta/s}\|\vec u\|_2^2.
\]

For $V=\sum_Iu_I\Gamma_I$ and spectral diameter $\Delta$ of $H$, Lemmas~\ref{lem:hessian-bkm} and~\ref{lem:bkm-variance} give directly
\[
 \vec u^\top\nabla^2\log Z\,\vec u
 =\beta^2\sigma_n^2\operatorname{Var}_\rho^{\rm BKM}(V)
 \ge\frac{\beta^2\sigma_n^2}{1+\beta\Delta}\operatorname{Var}_\rho(V).
\]
The tail estimate proved in Lemma~\ref{lem:spectral_bounds}, evaluated at $u=n$ and using $\sigma_n^2m\le n/4$, gives
\[
 \Pr_{\vec \bsg}[\|H\|_{\rm op}>n]\le2\exp[-(2-\tfrac12\log2)n].
\]
On the complementary event $\Delta\le2n$, and therefore
\[
 \frac{\beta^2\sigma_n^2}{1+\beta\Delta}
 \ge\frac{6\beta^2}{1+2\beta}\,n^{-4}.
\]
This proves \eqref{sampleall:eq:hessian} with $c_\beta=6\beta^2/(1+2\beta)$ and $r_\beta=8+9\beta/s_\beta$. Its failure probability is at most $2/n+2\exp[-(2-\tfrac12\log2)n]$.
\end{proof}

\subsection{Proof of main theorem}
\label{sec:sample_all_main_proof}

Finally, we prove the main theorem.

\samplealltemp*

\begin{proof}
Take the feasible set $K_n=\{\vec h\in\mathbb R^m:\|\vec h\|_2\le2\sqrt m\}$. Corollary~\ref{sampleall:cor:hessian} and the chi-squared event in the proof of Lemma~\ref{lem:spectral_bounds} imply that, with probability at least $1-2/n-e^{-\Omega(n)}-e^{-\Omega(m)}$,
\[
 \vec \bsg\in K_n,\qquad \nabla^2\Phi(\vec \bsg)\succeq\lambda I,
 \qquad \lambda=c_\beta n^{-r_\beta},\qquad
 c_\beta=\frac{6\beta^2}{1+2\beta}.
\]
Fix a realization $\vec g$ in this event. Set $L_\beta=3\beta^3/4$, $R=\lambda/(2L_\beta)$, and take the moment tolerance from Theorem~\ref{sampleall:thm:error-from-convexity}:
\[
 \delta_\varepsilon=\min\left\{\frac{\lambda\varepsilon}{8\beta\sigma_n},
                              \frac{\lambda^2}{16\beta\sigma_nL_\beta}\right\}.
\]
For sufficiently large $n$, $\delta_\varepsilon/\sqrt m\in(0,1)$. Apply Lemma~\ref{lem:measurement_complexity} with this coordinate precision and failure probability $\zeta$. Using
\[
 N=m\left\lceil\frac{2m}{\delta_\varepsilon^2}\log\frac{2m}{\zeta}\right\rceil
\]
copies, it returns empirical means $\vec\mu_{\beta,\mathrm{meas}}\in[-1,1]^m$ satisfying $\|\vec\mu_{\beta,\mathrm{meas}}-\vec\mu_\beta\|_2\le\delta_\varepsilon$ with probability at least $1-\zeta$.

Choose $0<\eta\le\min\{\lambda R^2/16,\lambda\varepsilon^2/16\}$ and compute an $\eta$-approximate minimizer over $K_n$ of
\[
 \mc L_\mathrm{meas}(\vec h)=\Phi(\vec h)+\beta\sigma_n \vec h\cdot\vec\mu_{\beta,\mathrm{meas}}.
\]
Theorem~\ref{sampleall:thm:error-from-convexity}, applied with $\vec g_\star=\vec g$, gives $\|\vec g_\mathrm{meas}-\vec g\|_2\le\varepsilon$ on the measurement-success event.

For completeness, this approximate minimization can be implemented with finite computation. Lemmas~\ref{lem:first_derivative_identity} and~\ref{alltemp:lem:elementary} give $\|\nabla \Phi(\vec h)\|_2\le\beta/2$, and $\|\vec\mu_{\beta,\mathrm{meas}}\|_2\le\sqrt m$. Thus $\mc L_\mathrm{meas}$ is globally $G_n$-Lipschitz with $G_n=\beta/2+\beta\sigma_n\sqrt m$. Evaluate $\mc L_\mathrm{meas}$ to accuracy $\eta/4$ on a finite $\eta/(4G_n)$-net of $K_n$ and return the point with the smallest computed value. Its objective value is within $3\eta/4$ of the minimum. The evaluations can be obtained from finite-dimensional matrix exponentials, with no runtime bound required. As in the usual numerical-access convention, the known $\beta$ is supplied to the precision needed for these evaluations.

It remains to bound the number of copies. Since $\lambda=c_\beta n^{-r_\beta}$, $L_\beta$ is independent of $n$, and $m^2\sigma_n^2=O(n^5)$,
\[
\begin{aligned}
 N
 &\le m+\frac{2m^2}{\delta_\varepsilon^2}\log\frac{2m}{\zeta}\\
 &\le m+C_\beta m^2\sigma_n^2
       \left(\lambda^{-4}+\lambda^{-2}\varepsilon^{-2}\right)\log\frac{2m}{\zeta}\\
 &\le C_\beta\left(n^{4r_\beta+5}+n^{2r_\beta+5}\varepsilon^{-2}\right)\log(n/\zeta).
\end{aligned}
\]
This proves \eqref{sampleall:eq:samplecomplexity}. The good-disorder event is independent of $\varepsilon$ and $\zeta$, and the measurement guarantee holds conditionally for every realization in that event.
\end{proof}

\section{Time-efficient algorithm}

The sample-efficient algorithm of \Cref{sec:sample_efficiency} reconstructs the Hamiltonian by
solving a maximum-entropy optimization problem.  While this gives polynomial
sample complexity, evaluating the objective requires computing partition
functions of candidate SYK Hamiltonians and does not directly lead to an
efficient reconstruction procedure.  In this section we instead construct
an explicit estimator directly from the measured quartic expectations.

For $I\in\mathcal I_4$, define the rescaled thermal expectation
\[
    U_{\beta,I}(\vec g)
    :=
    -\frac{\operatorname{Tr}[\Gamma_I\rho_\beta(\vec g)]}
    {\beta\sigma_n},
    \qquad
    \vec U_\beta=(U_{\beta,I})_{I\in\mathcal I_4}.
\]
The normalization is chosen so that $\vec U_\beta$ agrees with the disorder
vector $\vec g$ to first order in $\beta$.  Indeed, expanding around
$\beta=0$ gives
\begin{equation}
    U_{\beta,I}(\vec g)
    =
    g_I
    -
    \frac{\beta\sigma_n}{2}
    \sum_{J_1\triangle J_2=I}
        \chi(I;J_1,J_2)\,
        g_{J_1}g_{J_2}
    +O(\beta^2).
    \label{eq:time-U-expansion}
\end{equation}
For fixed $I$, the quadratic term contains
$N_2=\Theta(n^2)$ orthogonal Gaussian monomials.  Since
$\sigma_n=\Theta(n^{-3/2})$, its Gaussian $L^2$ norm is
\[
    \Theta_\beta(\sigma_n\sqrt{N_2})
    =
    \Theta_\beta(n^{-1/2}).
\]
Thus the naive estimator $\vec g_{\rm meas}=\vec U_\beta$ is not accurate
enough for our target error: the leading quadratic bias must be removed.

Our reconstruction proceeds in three steps.  First, we replace $\vec U_\beta$ by an explicit
polynomial proxy $\vec V_\beta(\vec g)$ of degree $K=O(\log n)$ in the Gaussian disorder variables.
Second, we exploit the sign and permutation symmetries of SYK to show that the entire Hermite
degree-$\le2$ part of the reconstruction problem is two-dimensional a spanned by linear polynomial
$\Phi_1$ and quadratic polynomial $\Phi_2$, see \Cref{eq:two_hermite_span}.  This allows us to
construct a quadratic map
\[
    \vec f_{\rm reconstruct} = c_1 \vec\Phi_1 + c_2\vec\Phi_2
\]
whose composition with $\vec V_\beta$ has exactly the desired degree-one component and no degree-two
component:
\[
    \Pi_{\leq 2} \vec f_{\rm reconstruct}(\vec V_\beta(\vec g)) = \vec g\,.
\]
Finally, we prove quantitative bounds on the higher Hermite degrees,
showing that after the quadratic bias is cancelled the first surviving
error occurs at degree three and is already of order $n^{-1}$. The rest of the section implements these three steps.  Section~\ref{eq:expproxy}
constructs the polynomial proxy and proves that it approximates
$\vec U_\beta$ in the required $L^4$ norm.  Section~\ref{eq:2dimcalibration} develops the
two-dimensional calibration and defines the reconstruction map.
Section~\ref{sec:errorfordegree3} bounds the remaining degree-$\ge3$ contribution and transfers
the resulting estimate from $\vec V_\beta$ back to the true thermal
expectations $\vec U_\beta$.

\subsection{Expectation proxy}
\label{eq:expproxy}

We formally define our expectation proxy $\vec{V}_\beta$. Rewrite $\Tr[\Gamma_I\rho_\beta]$ as
\[
  \Tr[\Gamma_I\rho_\beta(\vec g)]
    =\frac{\Tr[\Gamma_Ie^{-\beta H(\vec g)}]}{Z_\mathrm{ann}(\beta)}
       \cdot\frac{Z_\mathrm{ann}(\beta)}{Z(\beta,\vec g)}.
\]
Next, we use the Taylor approximation $1/x=1-(x-1)+O((x-1)^2)$ around $x=1$
to approximate the second term by $2-Z(\beta,\vec g)/Z_\mathrm{ann}(\beta)$.
This has the advantage of bringing the disorder vector $\vec g$ from the denominator
into the numerator, where it is more amenable to averaging.
Finally, we treat the rescaled expression as a Taylor series in $\beta$ and truncate
at order $K=O(\log n)$.

\begin{definition}[Expectation proxy $\vec{V}_\beta$]
  \label{def:expec_proxy}
  Set $K=\lceil3\log n/\log8\rceil$ and define
  \[
    \vec{V}_\beta=(V_{\beta,I})_{I\in\mc I_4},\qquad
    V_{\beta,I}(\vec g)=\left[
       -\frac{\Tr[\Gamma_Ie^{-\beta H(\vec g)}]}{\beta\sigma_n Z_\mathrm{ann}(\beta)}
        \left(2-\frac{Z(\beta,\vec g)}{Z_\mathrm{ann}(\beta)}\right)
          \right]_{\leq K,\beta}.
  \]
  The notation $[f(\beta)]_{\leq K,\beta}$ denotes truncation of the Taylor series in $\beta$ at
  order $K$. The singularity at $\beta=0$ is removable, with $V_{0,I}(\vec g)=g_I$ and thus whenever
  $\beta$ is in the zero-free region of $Z_\mathrm{ann}$, $V_{\beta, I}$ is analytic. As a
  polynomial in $\vec g$, $V_{\beta,I}$ has degree at most $K+1$; the extra one comes from division
  by $\beta$.
\end{definition}

\begin{restatable}[Thermal estimates]{lemma}{thermalestimates}
  \label{lem:thermal_estimates}
  For every fixed $s\geq1$ there are constants $\beta_\mathrm{thermal}(s)>0$ and
  $C_s<\infty$, independent of $n$, with the following properties for all sufficiently
  large $n$. The annealed partition function has no zero on a neighborhood of
  $|\beta|\leq \beta_\mathrm{thermal}(s)$. The functions
  \[
	\Delta_\beta(\vec g)=\frac{Z(\beta,\vec g)}{Z_\mathrm{ann}(\beta)}-1, \qquad
	U_{\mathrm{annealed}, I}(\beta, \vec g) = -\frac{\Tr[\Gamma_I e^{-\beta H(\vec g)}]}{\beta
	\sigma_n Z_\mathrm{ann}(\beta)}
  \]
  are analytic there in the weighted Gaussian norm, with the removable singularity
  $U_{\mathrm{annealed}, I}(0, \vec g) = g_I$, and obey
  \[
	\sup_{|\beta| \leq \beta_\mathrm{thermal}(s)}\|\Delta_\beta\|_{2,s} \leq C_s n^{-1},\qquad
	\sup_{|\beta| \leq \beta_\mathrm{thermal}(s)} \max_{I \in \mc I_4} \|U_{\mathrm{annealed},
	I}(\beta, \cdot)\|_{2,s} \leq C_s.
  \]
\end{restatable}
The proof in \Cref{sec:weighted_thermal_proof} fixes the family as
$\beta_\mathrm{thermal}(s)=\beta_\mathrm{thermal}/s$.
Its base value is independent of the choice of $C_\star$ in the degree bounds below.

Recall that we will set $\fdec = c_1 \vec{\Phi}_1 + c_2 \vec{\Phi}_2$. In order to control both the degree $d\geq3$ contribution of $\fdec(\vec{V}_\beta)$ and show that $M_\mathrm{cal}$ is invertible, we require a bound on
$\|\Pi_d \vec{\Phi}_a (\vec{V}_\beta)\|$ for $a \in [2]$.
\begin{restatable}[Degree bounds on calibrated proxy]{lemma}{degreebound}
  \label{lem:degree_bound}
  For the absolute constant $C_\star\geq1$ fixed above, there is a constant $c'>0$,
  independent of $n$, such that, for the proxy $V_{\beta, I}$ in \Cref{def:expec_proxy}, all sufficiently large $n$,
  every $I\in\mc I_4$, every $1\leq d\leq2(K+1)$, and every complex
  $|\beta|\leq 1/(64C_\star)$,
  \begin{align*}
    \|\Pi_d\Phi_{1,I}(\vec{V}_\beta)\|_{L^2}
       &\leq c_d^{(\mathrm{deg})} n^{-(d-1)/2},\\
    \|\Pi_d\Phi_{2,I}(\vec{V}_\beta)\|_{L^2}
       &\leq c_d^{(\mathrm{deg})} n^{-(d-2)/2},
       \qquad c_d^{(\mathrm{deg})}=(c'd)^{5d/2}.
  \end{align*}
  Here $\vec{\Phi}_2$ has the normalized definition given above.
\end{restatable}
The proof is given in \Cref{sec:degree_proof}.
For the reconstruction, fix
\begin{equation}
 \label{eq:algorithm_threshold}
 M=\max\{1,c^{(\mathrm{deg})}_1,c^{(\mathrm{deg})}_2\},\qquad
 \beta_\mathrm{algorithm}=\min\left\{
   \frac{1}{768MC_\star},
   \frac{\beta_\mathrm{thermal}(\sqrt{15})}{8}\right\}\,,
\end{equation}
with $c_1^{(\rm deg)}$, $c_2^{(\rm deg)}$ set by \Cref{lem:degree_bound} and will control the
entry-wise maximum of the calibration matrix. The two restrictions give calibration stability and
the Taylor-tail bound, respectively; the subsequent reconstruction lemmas use this single threshold.
The fixed weight $s=\sqrt{15}$ supplies the $L^4$, $L^8$, and $L^{16}$ bounds needed inside the
proxy proof, while the final reconstruction error is estimated in $L^2$.

\subsection{Low-degree calibration}
\label{eq:2dimcalibration}
\subsubsection{Calibration is two-dimensional}
\label{sec:calibration_two_dimensional}
In this section we show that permutation and sign covariance imply that the low-degree component
of $\fdec(\vec{V}_\beta)$ is spanned by two Hermite polynomials $\vec{\Phi}_1$ and $\vec{\Phi}_2$.

\begin{lemma}[Sign covariance of $\vec{V}_\beta$]
	\label{lem:sign_covariance}
	Let $\vec x \in \npo^n$ and for $\vec g \in \R^m$, let $(\vec x \cdot \vec g)_I = x_I g_I$, where $x_I \dfn \prod_{i \in I} x_i$. Then,
	\[
		V_{\beta, I}(\vec x \cdot \vec g) = x_I V_{\beta, I}(\vec g)\,.
	\]
\end{lemma}

\begin{proof}
	Fix $\vec x \in \npo^n$. Let $T_\pm = \{i : x_i = \pm1\}$ and
	\[
		T = \begin{cases}
			T_- & |T_-| \text{ even}\,,\\
			T_+ & |T_-| \text{ odd}\,.
		\end{cases}
	\]
	Define the unitary
	\[
		U_{\vec x} = \prod_{i \in T} \psi_i\,.
	\]
	Then, one can check that $U_{\vec x} \psi_i U_{\vec x}^\dagger = x_i \psi_i$. Since $U_{\vec
	x}$ is unitary, this implies that $U_{\vec x} \Gamma_I U_{\vec x}^\dagger = x_I \Gamma_I$.
	As a consequence,
	\[
		e^{-\beta H(\vec x \cdot \vec g)} = e^{-\beta U_{\vec x} H(\vec g) U_{\vec x}^\dagger}
		= U_{\vec x} e^{-\beta H(\vec g)} U_{\vec x}^\dagger\,.
	\]
	This implies that,
	\[
		\Tr[\Gamma_I e^{-\beta H(\vec x \cdot \vec g)}] = x_I \Tr[\Gamma_I e^{-\beta H(\vec g)}] \quad\text{and}\quad Z(\beta, \vec x \cdot \vec g) = Z(\beta, \vec g)\,.
	\]
	From \Cref{def:expec_proxy}, we can see that this shows the desired claim for the
	\emph{untruncated} proxy. However, since the Taylor expansion is in terms of $\beta$, this
	easily implies the same result for $\vec{V}_{\beta}$.
\end{proof}

\begin{lemma}[Permutation covariance of $\vec{V}_\beta$]
\label{lem:permutation_covariance}
	Let $\pi \in \mc S_n$ be a permutation on $[n]$ and for any subset $S \subseteq [n]$, let
	$\sgn_S(\pi)$ denote the sign needed to sort the list $(\pi(i))_{i\in S}$, with $S$ initially in increasing order. Additionally, let $(\pi
	\cdot \vec g)_{\pi(I)} = \sgn_I(\pi) g_I$. Then,
	\[
		V_{\beta, \pi(I)}(\pi \cdot \vec g) = \sgn_I(\pi) V_{\beta, I}(\vec g)\,.
	\]
\end{lemma}

\begin{proof}
	We show that the map $\Gamma_I \rightarrow \sgn_I(\pi) \Gamma_{\pi(I)}$ can be implemented
	by a unitary operator. Let $\Psi = \I^{n/2} \psi_1 \psi_2 \dots \psi_n$. Since $n$ is even,
	we have that,
	\[
		\Psi^\dagger = \Psi\,,\quad \Psi^2 = \Id\,,\quad \Psi \psi_k \Psi = - \psi_k\,.
	\]
	For any transposition $\tau = (i,j)$, define $U_\tau = \Psi \frac{\psi_i - \psi_j}{\sqrt
	2}$. One can verify that, $U_\tau \psi_k U_\tau^\dagger = \psi_{\tau(k)}$. Given a permutation
	$\pi$, write it as a composition of transpositions $\tau_1 \circ \dots \circ
	\tau_L$ and define
	\[
		U_\pi = U_{\tau_1} \dots U_{\tau_L}\,.
	\]
	Then $U_\pi\Gamma_IU_\pi^\dagger=\sgn_I(\pi)\Gamma_{\pi(I)}$. This implies
	that $e^{-\beta H(\pi \cdot \vec g)} = U_\pi e^{-\beta H(\vec g)} U_\pi^\dagger$ and using
	\Cref{def:expec_proxy} yields the desired conclusion.
\end{proof}

\begin{cor}[Covariance of calibrated $\vec{V}_\beta$]
	\label{cor:q_covariance}
	Let $\vec{Q} = c_1 \vec{\Phi}_1 + c_2 \vec{\Phi}_2$. Then, $\vec{Q}(\vec{V}_\beta)$ is also sign and permutation
	covariant.
\end{cor}
\begin{proof}
Since $\vec{\Phi}_1=\Id$, this does not affect the sign or permutation covariance of $\vec{V}_\beta$.
For $\vec{\Phi}_2$, we have
\begin{align}
  \Phi_{2,I}(\vec{V}_\beta(\vec x\cdot\vec g))
    &=\frac1{2\sqrt{N_2}}\sum_{J_1\triangle J_2=I}
       \chi(I;J_1,J_2)V_{\beta,J_1}(\vec x\cdot\vec g)V_{\beta,J_2}(\vec x\cdot\vec g)\\
    &=\frac1{2\sqrt{N_2}}\sum_{J_1\triangle J_2=I}
       \chi(I;J_1,J_2)x_{J_1}x_{J_2}V_{\beta,J_1}(\vec g)V_{\beta,J_2}(\vec g)
      =x_I\Phi_{2,I}(\vec{V}_\beta(\vec g)),
\end{align}
proving sign covariance.

For permutation covariance, conjugating the three Clifford monomials by $U_\pi$
and taking the trace gives
\begin{equation}
  \chi(\pi(I);\pi(J_1),\pi(J_2))\sgn_{J_1}(\pi)\sgn_{J_2}(\pi)
      =\sgn_I(\pi)\chi(I;J_1,J_2).
  \label{eq:chi-covariance}
\end{equation}
Consequently,
\begin{align*}
  \Phi_{2,\pi(I)}(\vec{V}_\beta(\pi\cdot\vec g))
    &=\frac1{2\sqrt{N_2}}\sum_{J_1\triangle J_2=I}
       \chi(\pi(I);\pi(J_1),\pi(J_2))
       V_{\beta,\pi(J_1)}(\pi\cdot\vec g)V_{\beta,\pi(J_2)}(\pi\cdot\vec g)\\
    &=\frac1{2\sqrt{N_2}}\sum_{J_1\triangle J_2=I}
       \chi(\pi(I);\pi(J_1),\pi(J_2))\sgn_{J_1}(\pi)\sgn_{J_2}(\pi)
       V_{\beta,J_1}(\vec g)V_{\beta,J_2}(\vec g)\\
    &=\sgn_I(\pi)\Phi_{2,I}(\vec{V}_\beta(\vec g)).
\end{align*}
\end{proof}

Moreover, any function which is sign and permutation covariant has a particularly simple
decomposition over its degree $\leq 2$ component.
\begin{lemma}[Degree-$\leq 2$ characterization]
	Let $\Pi_{\leq 2}$ be the projector onto the span of degree at most $2$ Hermite polynomials.
	Let $\vec{F}:\R^{\mc I_4}\to\C^{\mc I_4}$ have finite Gaussian $L^2$ norm and be
sign and permutation covariant, in the sense of \Cref{lem:sign_covariance,lem:permutation_covariance}.
Then,
	\begin{equation}
		\label{eq:canonical_form}
		\Pi_{\leq 2} \vec{F} = c^{\tout}_1 \vec{\Phi}_1 + c^{\tout}_2 \vec{\Phi}_2\,,
	\end{equation}
	where
	\[
		\vec{\Phi}_1(\vec g) = \vec g \quad \text{and} \quad \Phi_{2,I}(\vec g) =\frac1{2\sqrt{N_2}}
		\sum_{\substack{J_1, J_2 \in \mc I_4\\ J_1 \triangle J_2 = I}} \chi(I; J_1, J_2)
		g_{J_1} g_{J_2}\,,
	\]
	and $\chi(I; J_1, J_2) = \mathsf{sign}(\Tr[\Gamma_I \Gamma_{J_1} \Gamma_{J_2}])$.
\end{lemma}

\begin{proof}
	We consider multivariate Hermite polynomials of degree at most $2$. We list the possible
	options below, along with their sign properties.
	\begin{align*}
		& p_0(\vec g) = 1 & &p_0(\vec x \cdot \vec g) = p_0(\vec g)\\
		& p^{(1)}_I(\vec g) = \mathsf{Herm}_1(g_I) & &p^{(1)}_I(\vec x \cdot \vec g) = x_I
		p^{(1)}_I(\vec g)\\
		& p_{I,J}(\vec g) = \mathsf{Herm}_1(g_I) \mathsf{Herm}_1(g_J) &
		&p_{I,J}(\vec x \cdot \vec g) = x_{I \triangle J} p_{I, J}(\vec g)\\
		& p^{(2)}_I(\vec g) = \mathsf{Herm}_2(g_I) & &p^{(2)}_I(\vec x \cdot \vec g) =
		p^{(2)}_I(\vec g)\\
	\end{align*}
	In general, we may write,
	\[
  \Pi_{\leq2}F_{I^\star}=c_0^{I^\star}p_0
     +\sum_Ic_I^{I^\star,1}p_I^{(1)}
     +\sum_Jc_J^{I^\star,2}p_J^{(2)}
     +\sum_{I\ne J}b_{I,J}^{I^\star}p_{I,J},
\]
where the last coefficients are chosen symmetric in $I,J$. The Gaussian law is invariant
under the signed coordinate permutations in question, so its degree projectors commute
with these actions.
	By the assumed sign covariance we must have $F_{I^\star}(\vec x \cdot \vec g) =
	x_{I^\star} F_{I^\star}(\vec g)$ for every $\vec x$; comparing with the second column
	above, we see that all components except for $p^{(1)}_{I^\star}$ and $p_{I,J}$ with $I
	\triangle J = I^\star$ are zero. Thus,
	\[
		\Pi_{\leq 2} F_{I^\star} = c_{I^\star} p^{(1)}_{I^\star} +\frac1{2\sqrt{N_2}}\sum_{I\triangle J=I^\star} c^{I^\star}_{I, J} \chi(I^\star; I, J) p_{I, J}\,.
	\]
	To get our claimed form \Cref{eq:canonical_form}, we need to show for all $I^\star, I,
	J$, that $c_{I^\star} = c^{\tout}_1$ and $c^{I^\star}_{I, J} = c^{\tout}_2$.

	By permutation covariance, we have that $F_{\pi(I^\star)}(\pi \cdot \vec g) =
	\sgn_{I^\star}(\pi) F_{I^\star}(\vec g)$. At degree one, this implies
	\[
		c_{\pi(I^\star)} p^{(1)}_{\pi(I^\star)}(\pi \cdot \vec g) = \sgn_{I^\star}(\pi)
		c_{I^\star} p^{(1)}_{I^\star}(\vec g)\,.
	\]
	But by construction $p^{(1)}_{\pi(I^\star)}(\pi \cdot \vec g) = \sgn_{I^\star}(\pi)g_{I^\star} =
	\sgn_{I^\star}(\pi) p^{(1)}_{I^\star}(\vec g)$. Cancelling like terms on both sides yields that
	$c_{\pi(I^\star)} = c_{I^\star}$.

	At degree two, we can write,
	\begin{align}
  F^{(2)}_{\pi(I^\star)}(\pi\cdot\vec g)
    &=\frac1{2\sqrt{N_2}}\sum_{I'\triangle J'=\pi(I^\star)}
       c_{I',J'}^{\pi(I^\star)}\chi(\pi(I^\star);I',J')
           (\pi\cdot\vec g)_{I'}(\pi\cdot\vec g)_{J'}\\
    &=\frac1{2\sqrt{N_2}}\sum_{I\triangle J=I^\star}
       c_{\pi(I),\pi(J)}^{\pi(I^\star)}
       \chi(\pi(I^\star);\pi(I),\pi(J))\sgn_I(\pi)\sgn_J(\pi)g_Ig_J\\
    &=\frac1{2\sqrt{N_2}}\sum_{I\triangle J=I^\star}
       c_{\pi(I),\pi(J)}^{\pi(I^\star)}
       \sgn_{I^\star}(\pi)\chi(I^\star;I,J)g_Ig_J.
\end{align}
	Comparing the above expression term by term with the corresponding equation for
	$\sgn_{I^\star}(\pi) F^{(2)}_{I^\star}(\vec g)$ yields that
	$c^{\pi(I^\star)}_{\pi(I), \pi(J)} = c^{I^\star}_{I, J}$. Now, for any pair $I \triangle J =
	I' \triangle J' = I^\star$, there exists some permutation $\pi$ satisfying $\pi(I^\star) =
	I^\star$ yet $\pi(I) = I'$ and $\pi(J) = J'$. Indeed, a permutation fixing $I^\star$ can send one ordered partition of
$I^\star$ into two pairs to the other, and a permutation of its complement can send
$I\cap J$ to $I'\cap J'$. Thus the stabilizer is transitive on these pairs. But then,
	\[
		c^{I^\star}_{I', J'} = c^{\pi(I^\star)}_{\pi(I), \pi(J)} = c^{I^\star}_{I, J}\,,
	\]
	which implies that $c^{I^\star}_{I,J} = c^{I^\star}$ for all $I, J$.

	Finally, for any $I^\star, J^\star$, there is a permutation $\pi$ such that $\pi(I^\star) =
	J^\star$. Then,
	\[
		F^{(2)}_{\pi(I^\star)}(\pi \cdot \vec g) = \sgn_{I^\star}(\pi)
		F^{(2)}_{I^\star}(\vec g) = F^{(2)}_{J^\star}(\pi \cdot \vec g)
	\]
	Expanding the second and third terms and comparing them term-wise yields that $c^{I^\star} =
	c^{J^\star}$, and thus $c^{I^\star} = c^{\tout}_2$, as desired.
\end{proof}

Given this canonical form, we describe the low-degree calibration procedure. By
\Cref{cor:q_covariance}, $\vec{Q}(\vec{V}_\beta)$ is sign and permutation covariant, so the above lemma yields,
\[
	\Pi_{\leq 2} \vec{Q}(\vec{V}_\beta) = c^{\tout}_1 \vec{\Phi}_1 + c^{\tout}_2 \vec{\Phi}_2\,.
\]
The goal is to pick $\vec{Q} = c_1 \vec{\Phi}_1 + c_2 \vec{\Phi}_2$ so that $c^{\tout}_1 = 1$ and $c^{\tout}_2 =
0$; if this were possible, this would imply that,
\[
	\Pi_{\leq 2} \vec{Q}(\vec{V}_\beta(\vec g)) = \vec g\,,
\]
giving us recovery of the model parameters $\vec g$ (supposing the high-degree component was
negligible). Notice that $\vec{Q}$ is a linear combination of $\vec{\Phi}_1, \vec{\Phi}_2$, and consequently for any
$I \in \mc I_4$, $Q_I$ is a linear combination of $\Phi_{1, I}, \Phi_{2, I}$. Fixing any $I$, we can
write this transformation as a $2\times2$ calibration matrix (all functions inside the expectation are evaluated at $\vec \bsg$),
\[
	M_\mathrm{cal}(\beta) = \begin{bmatrix} \E_{\vec \bsg}[ \Phi_{1, I} \cdot \Pi_{\leq 2} \Phi_{1,
		I}(\vec{V}_\beta)] & \E_{\vec \bsg}[ \Phi_{1, I} \cdot \Pi_{\leq 2} \Phi_{2, I}(\vec{V}_\beta)]\\
		\E_{\vec \bsg}[ \Phi_{2, I} \cdot \Pi_{\leq 2} \Phi_{1, I}(\vec{V}_\beta)] &
		\E_{\vec \bsg}[ \Phi_{2, I} \cdot \Pi_{\leq 2} \Phi_{2, I}(\vec{V}_\beta)]\\
		\end{bmatrix} =
		\begin{bmatrix} m_{1,1}(\beta) & m_{1,2}(\beta) \\ m_{2,1}(\beta) & m_{2,2}(\beta)
	\end{bmatrix}\,.
\]
We will write $m_{i,j} = m_{i,j}(\beta)$ when the context is clear. Notice that the matrix entries
$m_{i,j}$ do not depend on the instance $\vec g$ and therefore can be computed ahead of time. Then,
the ideal solution is to set
\begin{equation}
	\label{eq:ideal_solution}
	\begin{bmatrix} c_1 \\ c_2 \end{bmatrix} = M_\mathrm{cal}^{-1}
	\begin{bmatrix} 1 \\ 0\end{bmatrix} = \frac 1 {\det M_\mathrm{cal}} \begin{bmatrix}
	m_{2,2} & -m_{1,2} \\ - m_{2,1} & m_{1,1}\end{bmatrix} \begin{bmatrix} 1 \\ 0
	\end{bmatrix}\,,
\end{equation}
which requires that $M_\mathrm{cal}$ be invertible.

\subsubsection{Invertibility of the calibration matrix}

We show that $M_\mathrm{cal}$ is invertible and obtain bounds on the coefficients $c_1, c_2$.
For ease of presentation, we define a rescaled matrix,
\[
	\tilde M = \begin{bmatrix} m_{1,1} & \frac 1 {\sqrt n} m_{1,2}\\
		\sqrt n m_{2,1} & m_{2,2}
	\end{bmatrix}\,,
\]
chosen so that each entry can be bounded by a constant $c$. Notice that $\det[\tilde M] =
\det[M_\mathrm{cal}]$ so it suffices to show invertibility of $\tilde M$.

At infinite temperature, $\vec{V}_0=\vec g$ and no calibration is necessary; using coordinatewise
orthonormality of $\vec{\Phi}_1, \vec{\Phi}_2$, we see that this corresponds to $\tilde M(0)=\Id$.
For sufficiently small constant $\beta>0$, Cauchy's integral formula in
\Cref{lem:cauchy_integral_formula} shows that every entry of
$\tilde M-\tilde M(0)$ is small, yielding invertibility.

\begin{lemma}[Invertibility of $M_\mathrm{cal}$]
	\label{lem:invertibility}
 For all sufficiently large $n$, every complex
 $|\beta|\leq\beta_\mathrm{algorithm}$, and $i,j\in[2]$,
	\[
		|\tilde m_{i,j} - \tilde m_{i,j}(0)| \leq \frac 1 6 \quad
		\text{ and } \quad |\det[\tilde M]| \geq \frac 2 3\,.
	\]
\end{lemma}

\begin{proof}
	To apply \Cref{lem:cauchy_integral_formula}, we bound $|\tilde m_{i,j}|$. This
	follows from \Cref{lem:degree_bound}.
	\begin{align*}
		|\tilde m_{1,1}| &= |\E_{\vec \bsg} \Phi_{1, I}(\vec \bsg) \cdot \Phi_{1,
				I}(\vec{V}_\beta(\vec \bsg))| \leq \|\bsg_I\|_{L^2} \cdot \|\Pi_1 \Phi_{1,I}(\vec{V}_\beta)\|_{L^2} \leq 1 \cdot c^{(\mathrm{deg})}_1\,.
	\end{align*}
	And,
	\begin{align*}
		|\tilde m_{2,1}| = \sqrt n |m_{2,1}| &=\sqrt n |\E_{\vec \bsg} \Phi_{2,
				I}(\vec \bsg) \cdot \Phi_{1, I}(\vec{V}_\beta(\vec \bsg))|\\
			&\leq \sqrt n \|\Phi_{2, I}\|_{L^2} \cdot \|\Pi_2
				\Phi_{1,I}(\vec{V}_\beta)\|_{L^2}\\
			&\leq \sqrt n \cdot 1 \cdot c^{(\mathrm{deg})}_2 \frac 1 {\sqrt n}\\
			&\leq c^{(\mathrm{deg})}_2\,.
	\end{align*}
	Similarly, $|\tilde m_{1,2}| = \frac 1 {\sqrt n} |m_{1,2}| \leq c^{(\mathrm{deg})}_1$ and
	$|\tilde m_{2,2}| \leq c^{(\mathrm{deg})}_2$. These bounds hold on the complex disk
	$|\beta|\leq1/(64C_\star)$. Applying \Cref{lem:cauchy_integral_formula} on this disk yields
	\begin{equation}
		\label{eq:coeff_bound}
		|\tilde m_{i,j} - \tilde m_{i,j}(0)| \leq
		M\frac{64C_\star|\beta|}{1-64C_\star|\beta|}\,.
	\end{equation}
 By \Cref{eq:algorithm_threshold}, $64C_\star|\beta|\leq1/(12M)\leq1/12$.
 Thus the right-hand side of \Cref{eq:coeff_bound} is at most
 \[
   128MC_\star|\beta|\leq\frac16,
 \]
 proving the first claim.

	For the second, write,
	\[
		|\det[\tilde M]| = |m_{1,1} m_{2,2} - m_{1,2}
		m_{2,1}| \geq (1  - 1/6)^2 - (1/6)^2 \geq \frac 2 3\,.
	\]
\end{proof}

\begin{cor}[Coefficient bounds on $\fdec$]
	Under the hypotheses of \Cref{lem:invertibility}, let
	$\fdec=c_1\vec{\Phi}_1+c_2\vec{\Phi}_2$ be as set by \Cref{eq:ideal_solution}.
	Then,
	\[
		|c_1| \leq 2\,, \qquad |c_2| \leq \frac2{\sqrt n}\,.
	\]
\end{cor}

\begin{proof}
	From \Cref{eq:ideal_solution}, we have that
	\[
		c_1 = \frac 1 {\det[M_{\mathrm{cal}}]} m_{2,2}\, \qquad c_2 = -\frac
		1 {\det[M_{\mathrm{cal}}]} m_{2,1}\,.
	\]
	Then, the determinant and matrix entry bounds from \Cref{lem:invertibility} give the result.
\end{proof}

\subsection{Error estimates}
\label{sec:errorfordegree3}

In this section, we account for the error sources in the exact-expectation analysis.
The low-degree calibration produces $\vec g_\mathrm{ideal}=\fdec(\vec{U}_\beta)$ and
$\vec\theta_\mathrm{ideal}=\beta\sigma_n\vec g_\mathrm{ideal}$, whereas the analysis is done on
$\fdec(\vec{V}_\beta)$. We compute the error $\|\vec{U}_\beta-\vec{V}_\beta\|$, then use a Lipschitz
property of the map $\fdec$ to translate this error bound to one on the final estimate.
Here the subscript $\mathrm{ideal}$ denotes reconstruction from exact expectations;
finite-measurement error is treated in the next section.
All norms and expectations in this section are over $\vec \bsg$; when arguments are suppressed, each function is evaluated at $\vec \bsg$. Throughout this section (except in the lemma statements), we write
$\|\cdot\|_q:=\|\cdot\|_{L^q}$ for $q\in\{2,4,8,16\}$.

\begin{lemma}[Expectation proxy error]
  \label{lem:proxy_error}
  Set $s=\sqrt{15}$ in \Cref{lem:thermal_estimates}, and let
  $\vec V_\beta$ be defined as in \Cref{def:expec_proxy}.
  For fixed real $0<|\beta|\leq \beta_\mathrm{algorithm}$ and sufficiently large $n$,
  \begin{equation}
    \label{eq:proxy_error_bound}
    \max_{I\in\mc I_4}\|U_{\beta,I}-V_{\beta,I}\|_{L^4}
       \leq O_\beta(n^{-3/2}).
  \end{equation}
  Moreover,
  \begin{equation}
    \label{eq:individual_bound}
    \max_{I\in\mc I_4}\|U_{\beta,I}\|_{L^4}
      +\max_{I\in\mc I_4}\|V_{\beta,I}\|_{L^4}
         \leq O_\beta(1).
  \end{equation}
\end{lemma}

\begin{proof}
Fix $I\in\mc I_4$. Let $\vec{V}_{\beta,\mathrm{full}}$ be the untruncated proxy. Then
\[
  \|U_{\beta,I}-V_{\beta,I}\|_4
    \leq\|U_{\beta,I}-V_{\beta,I}^{\mathrm{full}}\|_4
          +\|V_{\beta,I}^{\mathrm{full}}-V_{\beta,I}\|_4.
\]
We bound each term. For convenience, write
\[
  V_{\beta,I}^{\mathrm{full}}
    =\underbrace{-\frac{\Tr[\Gamma_Ie^{-\beta H(\vec g)}]}
                         {\beta\sigma_n Z_\mathrm{ann}(\beta)}}_{U_{\mathrm{annealed},I}(\beta,\vec g)}
       \left(1-\underbrace{\left(\frac{Z(\beta,\vec g)}{Z_\mathrm{ann}(\beta)}-1\right)}
                         _{\Delta_\beta(\vec g)}\right).
\]

\paragraph{Second term (truncation error).} Write
\[
  V_{\beta,I}^{\mathrm{full}}=\sum_{k\geq0}v_{I,k}\beta^k.
\]
We require a uniform $L^4$ bound on $V_{\beta,I}^{\mathrm{full}}$.
By the definition of the proxy and H\"older's inequality,
\begin{equation}
  \|V_{\beta,I}^{\mathrm{full}}\|_4
    \leq\|U_{\mathrm{annealed},I}\|_4
        +\|U_{\mathrm{annealed},I}\|_8\|\Delta_\beta\|_8.
\end{equation}
We use \Cref{lem:thermal_estimates,lem:norm_conversion} to bound the right-hand side.
The weighted norm with $s=\sqrt{15}$ also bounds the norms needed at the smaller
moment orders. Hence, uniformly on this complex disk,
\begin{equation}
  \label{eq:v_bound}
  \|V_{\beta,I}^{\mathrm{full}}\|_4\leq C+C^2n^{-1}\leq C_V.
\end{equation}
Since $|\beta| \leq \beta_\mathrm{thermal}(\sqrt{15})$, $Z_\mathrm{ann}(\beta)$ is zero-free and
$V_{\beta, I}$ is analytic. As such, we can apply \Cref{cor:cauchy_lq_bound} with $q = 4$. This
bounds $\|v_{I,k}\|_4$ by $C_V\beta_\mathrm{thermal}(\sqrt{15})^{-k}$.
For $|\beta|\leq\beta_\mathrm{algorithm}$, \Cref{eq:algorithm_threshold} gives $|\beta|/\beta_\mathrm{thermal}(\sqrt{15})\leq1/8$, so
\begin{align}
  \|V_{\beta,I}-V_{\beta,I}^{\mathrm{full}}\|_4
     &\leq\sum_{k>K}\|v_{I,k}\|_4|\beta|^k\\
     &\leq C_V\sum_{k>K}8^{-k}\leq\frac{8C_V}{7}\,8^{-K}
       \leq O(n^{-3}).
       \label{eq:proxy-truncation}
\end{align}
The last line follows from $K=\lceil3\log n/\log8\rceil$.

\paragraph{First term (proxy error).} Recall that
$V_{\beta,I}^{\mathrm{full}}=U_{\mathrm{annealed},I}(1-\Delta_\beta)$, whereas
$U_{\beta,I}=U_{\mathrm{annealed},I}/(1+\Delta_\beta)$. Therefore,
\[
  V_{\beta,I}^{\mathrm{full}}=U_{\beta,I}(1-\Delta_\beta^2),\qquad
  U_{\beta,I}-V_{\beta,I}^{\mathrm{full}}=\Delta_\beta^2U_{\beta,I}.
\]
By H\"older's inequality,
\begin{align}
  \|U_{\beta,I}-V_{\beta,I}^{\mathrm{full}}\|_4
     &\leq\|\Delta_\beta^2\|_8\|U_{\beta,I}\|_8\\
     &=\|\Delta_\beta\|_{16}^2\|U_{\beta,I}\|_8.
       \label{eq:proxy-holder}
\end{align}
The first factor is $O(n^{-2})$ by
\Cref{lem:thermal_estimates,lem:norm_conversion}. The second term is a Gibbs
observable scaled by $(\beta\sigma_n)^{-1}$; a priori, its absolute value could
be as large as $(|\beta|\sigma_n)^{-1}=O_\beta(n^{3/2})$.
Nonetheless, $U_{\beta,I}(1+\Delta_\beta)=U_{\mathrm{annealed},I}$ and we have
strong bounds on both factors on the right-hand side.

First consider the event $\mc E=\{|\Delta_\beta|\leq1/2\}$.
On this event, $|1+\Delta_\beta|\geq1/2$, and
$|U_{\beta,I}|\leq2|U_{\mathrm{annealed},I}|$. Therefore,
\begin{align*}
  \E_{\vec \bsg}|U_{\beta,I}|^8
    &\leq2^8\E_{\vec \bsg}|U_{\mathrm{annealed},I}|^8
        +( |\beta|\sigma_n)^{-8}\Pr_{\vec \bsg}[\mc E^c]\\
    &\leq C_\beta+O_\beta(n^{12})\Pr_{\vec \bsg}[\mc E^c].
\end{align*}
Markov's inequality yields
\[
  \Pr_{\vec \bsg}[\mc E^c]
     =\Pr_{\vec \bsg}[|\Delta_\beta|>1/2]
     \leq2^{16}\|\Delta_\beta\|_{16}^{16}
     \leq O(n^{-16}).
\]
This establishes
\begin{equation}
  \|U_{\beta,I}\|_8\leq O_\beta(1).
  \label{eq:true-proxy-moment}
\end{equation}
Plugging this bound into \Cref{eq:proxy-holder} gives
$\|U_{\beta,I}-V_{\beta,I}^{\mathrm{full}}\|_4\leq O_\beta(n^{-2})$.
Combining it with \Cref{eq:proxy-truncation} proves the stated (weaker)
$n^{-3/2}$ bound. The same estimates and \Cref{eq:v_bound,eq:true-proxy-moment}
prove \Cref{eq:individual_bound}.
\end{proof}

To translate the above results into error bounds on the final estimator, we need a
Lipschitz-like property of $\fdec$. We consider a general function $\vec{Q}=c_1\vec{\Phi}_1+c_2\vec{\Phi}_2$.

\begin{lemma}[Lipschitz property for quadratic reconstruction]
  \label{lem:reconstruction_lipschitz}
  Let $\vec{Q}=c_1\vec{\Phi}_1+c_2\vec{\Phi}_2$, and let $\vec{x},\vec{y}:\R^m\to\C^m$ be measurable
  functions of the same standard Gaussian vector, with finite coordinate $L^4$ norms.
  Then
  \begin{equation}
  \begin{aligned}
    \|\vec{Q}(\vec{x})-\vec{Q}(\vec{y})\|_{L^2}
      &\leq\sqrt m\max_I\|x_I-y_I\|_{L^4}\\[-2pt]
      &\quad\cdot\left[|c_1|+|c_2|\sqrt{N_2}
           \left(\max_I\|x_I\|_{L^4}+\max_I\|y_I\|_{L^4}\right)\right].
  \end{aligned}
  \label{eq:lipschitz_reconstruction}
  \end{equation}
\end{lemma}

\begin{proof}
Write $\delta=\max_I\|x_I-y_I\|_4$ and
$B_{\vec x}=\max_I\|x_I\|_4$, $B_{\vec y}=\max_I\|y_I\|_4$.
For each pair, the identity
$x_Jx_K-y_Jy_K=(x_J-y_J)x_K+y_J(x_K-y_K)$ and H\"older's inequality give
\[
  \|x_Jx_K-y_Jy_K\|_2\leq\delta(B_{\vec x}+B_{\vec y}).
\]
There are $2N_2$ ordered pairs in each coordinate of $\vec{\Phi}_2$, each with
coefficient of absolute value $(2\sqrt{N_2})^{-1}$. By the triangle inequality,
\[
  \|Q_I(\vec{x})-Q_I(\vec{y})\|_2
     \leq\delta\bigl(|c_1|+|c_2|\sqrt{N_2}(B_{\vec x}+B_{\vec y})\bigr).
\]
Finally, summing the coordinate second moments gives
\[
  \|\vec{Q}(\vec{x})-\vec{Q}(\vec{y})\|_{L^2}^2
     =\E_{\vec \bsg}\sum_I|Q_I(\vec{x})-Q_I(\vec{y})|^2
     =\sum_I\|Q_I(\vec{x})-Q_I(\vec{y})\|_{L^2}^2,
\]
which proves the claim. No independence between coordinates of $\vec{x}$ and $\vec{y}$ is required.
\end{proof}

\begin{lemma}[Final error bound]
  \label{lem:final_error}
  Use the fixed threshold $\beta_\mathrm{algorithm}$ from
  \Cref{eq:algorithm_threshold}.
  Let $0<\beta\leq \beta_\mathrm{algorithm}$ be fixed and real, let $n$ be
  sufficiently large and even, and use $K$ from \Cref{def:expec_proxy} and the exact
  calibration coefficients from \Cref{eq:ideal_solution}. For the exact rescaled
  Gibbs observables $U_{\beta,I}=-\Tr[\Gamma_I\rho_\beta]/(\beta\sigma_n)$,
  \[
    \frac1{\sqrt m}\|\vec \bsg-\fdec(\vec{U}_\beta)\|_{L^2}
          \leq C_\beta n^{-1}.
  \]
  Thus, setting $\vec{\theta}=\beta\sigma_n\vec \bsg$ and
  $\vec\theta_\mathrm{ideal}=\beta\sigma_n\fdec(\vec{U}_\beta)$,
  \[
    \frac1{\sqrt n}\|\vec\theta_\mathrm{ideal}-\vec{\theta}\|_{L^2}
       \leq O_\beta(n^{-1}).
  \]
\end{lemma}

\begin{proof}
First, by construction, $\Pi_{\leq2}\fdec(\vec{V}_\beta)=\vec \bsg$. Therefore,
\begin{align}
  \|\vec \bsg-\fdec(\vec{U}_\beta)\|_2
     &\leq\|\Pi_{\leq2}\fdec(\vec{V}_\beta)-\fdec(\vec{V}_\beta)\|_2
            +\|\fdec(\vec{V}_\beta)-\fdec(\vec{U}_\beta)\|_2.
\end{align}
To bound the first term, we decompose the degree-$d$ component of $\fdec(\vec{V}_\beta)$
coordinate by coordinate, then apply \Cref{lem:degree_bound}.
Recall that $\vec{V}_\beta$ has degree at most $K+1$, so $\fdec(\vec{V}_\beta)$ has degree
at most $2(K+1)$. The two degree-$d$ terms need not be orthogonal, so we use
$|a+b|^2\leq2|a|^2+2|b|^2$:
\begin{align*}
  \|\Pi_d\fdec(\vec{V}_\beta)\|_2^2
    &=\sum_I\|c_1\Pi_d\Phi_{1,I}(\vec{V}_\beta)+c_2\Pi_d\Phi_{2,I}(\vec{V}_\beta)\|_2^2\\
    &\leq2\sum_I\left(|c_1|^2\|\Pi_d\Phi_{1,I}(\vec{V}_\beta)\|_2^2
                  +|c_2|^2\|\Pi_d\Phi_{2,I}(\vec{V}_\beta)\|_2^2\right)\\
    &\leq2m(c_d^{(\mathrm{deg})})^2
          \left(|c_1|^2n^{-(d-1)}+|c_2|^2n^{-(d-2)}\right)\\
    &\leq m(c''d)^{5d}n^{-(d-1)},
\end{align*}
where the last line uses $|c_1|\leq2$ and $|c_2|\leq2/\sqrt n$.
Different Hermite degrees are orthogonal in $L^2$. Hence
\begin{align}
  \|\Pi_{>2}\fdec(\vec{V}_\beta)\|_2^2
    &=\sum_{d=3}^{2(K+1)}\|\Pi_d\fdec(\vec{V}_\beta)\|_2^2\notag\\
    &\leq m\sum_{d=3}^{2(K+1)}
          \frac{(c''d)^{5d}}{n^{d-1}}.
          \label{eq:tail-sum}
\end{align}
Let $a_d$ be the $d$-th summand without $m$. Then
\[
  \frac{a_{d+1}}{a_d}
    =\frac{(c''(d+1))^5}{n}\left(1+\frac1d\right)^{5d}
    \leq\frac{(c''e)^5(2(K+1))^5}{n}
    =O\left(\frac{(K+1)^5}{n}\right).
\]
Since $K=O(\log n)$, this ratio is at most $1/2$ for sufficiently large $n$.
Continuing from \Cref{eq:tail-sum},
\[
  \|\Pi_{>2}\fdec(\vec{V}_\beta)\|_2^2\leq2ma_3\leq Cmn^{-2}.
\]
For the second term, \Cref{lem:proxy_error} gives
\[
  \max_I\|U_{\beta,I}-V_{\beta,I}\|_4\leq C_\beta n^{-3/2},\qquad
  \max_I\|U_{\beta,I}\|_4+\max_I\|V_{\beta,I}\|_4\leq C'_\beta.
\]
Plugging these estimates into \Cref{lem:reconstruction_lipschitz},
\begin{align}
  \|\fdec(\vec{V}_\beta)-\fdec(\vec{U}_\beta)\|_2
    &\leq\sqrt m\,O_\beta(n^{-3/2})
       \left(2+\frac2{\sqrt n}\sqrt{3\binom{n-4}{2}}\,O_\beta(1)\right)\\
    &\leq\sqrt m\,O_\beta(n^{-1}).
\end{align}
Together these two bounds prove the claimed parameter error.
Finally, $\vec{\theta}=\beta\sqrt{6/n^3}\,\vec \bsg$ and $\sqrt m=\sqrt{\binom n4}$, so
\[
  \|\vec\theta_\mathrm{ideal}-\vec{\theta}\|_2
     \leq\beta\sqrt{6/n^3}\sqrt m\,O_\beta(n^{-1})
     \leq O_\beta(n^{-1/2}).
\]
The additional $1/\sqrt n$ scaling gives the desired result.
\end{proof}
\subsection{Proof of the main theorem}
\label{sec:time_main_proof}

We now combine the exact-expectation error bound with finite measurements and account for the
computation of the calibration coefficients.

\timelearning*

\begin{proof}[Proof of \Cref{thm:time_learning}]
Fix $\beta,c,\zeta$ as in the statement and take $K=\lceil3\log n/\log8\rceil$. For exact
expectations, write $\vec g_\mathrm{ideal}=\fdec(\vec U_\beta)$, with the coefficients from
\Cref{eq:ideal_solution}. We first describe how to compute these coefficients and implement the
measurements, then bound the resulting reconstruction error.

\paragraph{Precomputation.}
Set $X(\vec g)=\sum_Ig_I\Gamma_I$ and $z=\beta\sigma_n$. For $0\leq t\leq K+1$, compute the
trace polynomials
\[
 P_t(\vec g)=\Tr[X(\vec g)^t],\qquad
 Q_{I,t}(\vec g)=\Tr[\Gamma_I X(\vec g)^t].
\]
Enumerating ordered $t$-tuples of quartic labels computes these polynomials for all $I$ in
$m^{K+2}\operatorname{poly}(n,K)=n^{O(K)}$ operations. An ordered Clifford word is represented
by its symmetric-difference set and sign, so its normalized trace is computable without forming
$d_n\times d_n$ matrices.

The coefficients of the annealed partition function are computed using exact Gaussian moments:
for a Gaussian monomial $\prod_I\bsg_I^{\alpha_I}$, its expectation over $\vec \bsg$ is zero if any $\alpha_I$ is odd, and
otherwise is $\prod_I(\alpha_I-1)!!$. If its truncated series is
$\sum_{t=0}^{K+1}\bar z_tz^t$, then $\bar z_0=1$, and the reciprocal coefficients are given by
\[
 d_0=1,\qquad d_t=-\sum_{j=1}^t\bar z_jd_{t-j}.
\]
Truncated polynomial multiplication now computes $\vec V_\beta$ from \Cref{def:expec_proxy}.
Its numerator is needed through order $K+1$ because of the division by $z$.

Each of the four calibration entries
$\E_{\vec \bsg}[\Phi_{i,I}(\vec \bsg)\Phi_{j,I}(\vec V_\beta(\vec \bsg))]$ is the expectation of a
polynomial of Gaussian degree at most $2K+4$. There are at most
$\binom{m+2K+4}{2K+4}=n^{O(K)}$ monomials, so the same Gaussian-moment rule computes all four
entries in $n^{O(K)}$ operations. By \Cref{lem:invertibility}, the resulting $2\times2$
calibration system is invertible and gives $c_1,c_2$.

For the fixed rational $\beta$, all these quantities lie in the number field generated by
$\sqrt{6/n^3}$, $\sqrt{N_2}$, and $\sqrt n$, whose degree is at most eight. The factorials,
Gaussian moments, and rational coefficients involved have bit lengths polynomial in $K$ and
$\log n$; the sums, truncated products, and reciprocal recursion preserve a polynomial bit-length
bound. Exact arithmetic therefore gives a preprocessing bound of
$n^{O(K)}=n^{O(\log n)}$ bit operations as well.

\paragraph{Measurements and reconstruction.}
Apply \Cref{lem:measurement_complexity} with
\[
	\delta = \beta \sigma_n n^{-3},\qquad \widehat U_I = -\frac{\mu_{\beta, \mathrm{meas},
	I}}{\beta\sigma_n}\,.
\]
For all sufficiently large $n$, $\delta < 1$, and the number of copies is
\[
	N = m \left\lceil\frac{2n^6}{\beta^2\sigma_n^2} \log\frac{2m}{\zeta}\right\rceil = \O_\beta
	\Paren{n^{13}\log\frac{n}{\zeta}}.
\]
The time complexity is $\widetilde \O(N \cdot n) = \widetilde \O_\beta(n^{14}\log(n/\zeta))$. There
are $N_2=\O(n^2)$ unordered pairs in each coordinate of $\vec\Phi_2$, so evaluating
$\fdec(\widehat{\vec U})$ requires $O(mN_2)=O(n^6)$ operations.

\paragraph{Error bound.}
On the measurement-success event,
$\|\widehat{\vec U}-\vec U_\beta\|_\infty\leq n^{-3}$. Both the exact expectations and
empirical means lie in $[-1,1]$, so, putting $B=(\beta\sigma_n)^{-1}$,
\[
 \max_I|U_{\beta,I}|,\ \max_I|\widehat U_I|\leq B=O_\beta(n^{3/2}).
\]
The identity $x_Jx_K-y_Jy_K=(x_J-y_J)x_K+y_J(x_K-y_K)$ used in
\Cref{lem:reconstruction_lipschitz}, now applied pointwise, together with
$|c_1|\leq2$ and $|c_2|\leq2/\sqrt n$, gives
\[
 \frac1{\sqrt m}\|\fdec(\widehat{\vec U})-\fdec(\vec U_\beta)\|_2
 \leq\left(2+\frac{4\sqrt{N_2}B}{\sqrt n}\right)n^{-3}
 =O_\beta(n^{-1}).
\]
Rounding adds at most $n^{-2}$ to the normalized error. Consequently,
\[
 \frac1{\sqrt m}\|\vec g_\mathrm{meas}-\vec g_\mathrm{ideal}\|_2
 \leq C'_\beta n^{-1}
\]
with probability at least $1-\zeta$, for every disorder realization.

It remains to go from mean-square error to high-probability recovery. By \Cref{lem:final_error},
\[
 \left\|\frac1{\sqrt m}(\vec g_\mathrm{ideal}-\vec \bsg)\right\|_{L^2}
 \leq C_\beta n^{-1}.
\]
Thus, for $X(\vec g)=m^{-1/2}\|\vec g_\mathrm{ideal}(\vec g)-\vec g\|_2$, Markov's inequality
 gives
\[
 \Pr_{\vec \bsg}\!\left[X(\vec \bsg)>\tfrac12n^{-c}\right]
 \leq4n^{2c}\E_{\vec \bsg}[X(\vec \bsg)^2]
 \leq4C_\beta^2 n^{-2(1-c)}.
\]
For sufficiently large $n$, $C'_\beta n^{-1}\leq\tfrac12n^{-c}$. On the complementary disorder
event and the measurement-success event, the triangle inequality gives
$m^{-1/2}\|\vec g_\mathrm{meas}-\vec \bsg\|_2\leq n^{-c}$. This proves the claimed conditional
success probability; jointly over disorder and measurements, it is at least
$1-O_\beta(n^{-2(1-c)})-\zeta$.
\end{proof}

\printbibliography

\appendix

\section{Math tools}

\begin{lemma}[Cauchy's integral formula]
	\label{lem:cauchy_integral_formula}
	Let $S\subseteq\C$ be open and let $f:S\to\C$ be analytic. Suppose the disk $D_r=\{w\in\C:|w|\leq r\}$ is fully contained in $S$. Then for any $|x_0|<r$, we have
	\begin{equation}
		\label{eq:standard_cauchy}
		f(x_0) = \frac 1 {2\pi \I} \int_{|x| = r} \frac{f(x)}{x - x_0}\,\mathrm d x\,.
	\end{equation}
	Moreover, suppose $\sup_{x \in D_r} |f(x)| \leq c$. Then,
	\[
		|f(x_0) - f(0)| \leq c \frac {|x_0|}{r - |x_0|}\,.
	\]
\end{lemma}
\begin{proof}
	\Cref{eq:standard_cauchy} is the standard Cauchy integral formula. For the second, write,
	\begin{align}
		|f(x_0) - f(0)| = \Abs{\frac 1 {2 \pi \I} \int_{|x| = r} f(x) \Paren{\frac 1 {x -
				x_0} - \frac 1 x} \, \mathrm d x} &= \Abs{\frac 1 {2 \pi \I}
					\int_{|x| = r} f(x) \frac {x_0} {x(x-
		x_0)}\,\mathrm d x}\\
		&\leq \frac 1 {2\pi} \int_{|x| = r} |f(x)| \frac {|x_0|}{|x| \cdot |x -
		x_0|}\,|\mathrm d x|\\
		&\leq \frac {c}{2\pi} \frac{|x_0|}{r(r-|x_0|)}\int_{|x|=r}|\mathrm d x|\\
		&= \frac {c}{2\pi} \frac{|x_0|}{r(r-|x_0|)}2\pi r\\
		&\leq c \frac{|x_0|}{r - |x_0|}\,.
	\end{align}
\end{proof}

\begin{cor}[{Cauchy coefficient and derivative estimates;
\cite[Chapter~III, \S7]{Lang1999ComplexAnalysis}}]
\label{cor:cauchy_estimate}
Let $f$ be complex-valued and analytic on a neighborhood of
$D_r=\{z:|z|\leq r\}$. For $|z_0|<r$ and $k\geq0$,
\[
 f^{(k)}(z_0)=\frac{k!}{2\pi\I}
   \int_{|z|=r}\frac{f(z)}{(z-z_0)^{k+1}}\,\mathrm dz.
\]
If $f(z)=\sum_{k\geq0}f_kz^k$ and
$\sup_{|z|\leq r}|f(z)|\leq M$, then
\begin{equation}
 \label{eq:taylor_coeff_bound}
 |f_k|\leq Mr^{-k}.
\end{equation}
In particular, if $f$ is bounded by $M$ on a disk of radius $R$
centered at $z_0$, then $|f^{(k)}(z_0)|\leq k!MR^{-k}$.
The coefficient and derivative bounds also hold for functions analytic 
and bounded by $M$ on the corresponding open disk.

More generally, let $h:\{\vec u\in\C^m:\|\vec u\|_1<R\}\to\C$
be analytic with $|h(\vec u)|\leq M$. For any ordered tuple
$(B_1,\ldots,B_d)$ of coordinate indices, including repeated indices,
\begin{equation}
 \label{eq:cauchy_mixed_derivatives}
 |D_{B_1}\cdots D_{B_d}h(0)|
    \leq M(d/R)^d,\qquad d\geq1.
\end{equation}
\end{cor}
\Cref{eq:cauchy_mixed_derivatives} does not follow directly from the reference, but it is simple to
show. Apply the scalar coefficient bound successively to $H(t_1,\ldots,t_d)=h(\sum_{j=1}^dt_j\vec
e_{B_j})$ on $|t_j|<R/d$. The coefficient of $t_1\cdots t_d$ is the indicated ordered derivative,
even when indices repeat.

\begin{cor}[Cauchy $L^q$ coefficient bound]
	\label{cor:cauchy_lq_bound}
Let $1\leq q<\infty$ and
	\[
	f(z,\omega)=\sum_{k\geq0}f_k(\omega)z^k
	\]
	be analytic in $z$ on $|z|<R$ with $\sup_{|z|<R}\|f(z,\cdot)\|_{L^q}\leq M$. Then,
	\begin{equation}
		 \label{eq:cauchy_lp_coefficients}
		 \|f_k\|_{L^q}\leq MR^{-k}.
	\end{equation}
\end{cor}
\begin{proof}
Fix $0<r<R$. For each $\omega$, scalar Cauchy's
coefficient formula applies on the circle $|z|=r$, which lies strictly
inside the analytic disk. Parametrizing the circle by $z=re^{\I t}$ gives
\[
 f_k(\omega)
   =\frac{1}{2\pi\I}\int_{|z|=r}
       \frac{f(z,\omega)}{z^{k+1}}\,\mathrm dz
   =\frac{r^{-k}}{2\pi}\int_0^{2\pi}
       f(re^{\I t},\omega)e^{-\I kt}\,\mathrm dt.
\]
We now take the $L^q$ norm with respect to $\omega$. Minkowski's integral
inequality, valid for $q\geq1$, bounds the norm of an integral by the
integral of the norms, so
\[
\begin{aligned}
 \|f_k\|_{L^q}
 &=\frac{r^{-k}}{2\pi}
      \left\|\int_0^{2\pi}f(re^{\I t},\cdot)e^{-\I kt}\,\mathrm dt\right\|_{L^q}\\
 &\stackrel{\mathrm{Minkowski}}{\leq}
      \frac{r^{-k}}{2\pi}\int_0^{2\pi}
       \|f(re^{\I t},\cdot)e^{-\I kt}\|_{L^q}\,\mathrm dt\\
 &=\frac{r^{-k}}{2\pi}\int_0^{2\pi}
       \|f(re^{\I t},\cdot)\|_{L^q}\,\mathrm dt\\
 &\leq\frac{r^{-k}}{2\pi}\int_0^{2\pi}M\,\mathrm dt
   =Mr^{-k}.
\end{aligned}
\]
The second equality uses $|e^{-\I kt}|=1$, and the assumed
uniform $L^q$ bound ensures that the integral of the norms is finite.
Since this holds for every $r<R$, letting $r \rightarrow R$ gives the desired conclusion.
\end{proof}

\begin{lemma}[Hermite coefficients from complex shifts]
\label{lem:shift_to_hermite}
Let $\mc I$ be a finite nonempty index set, let
$\vec \bsg=(\bsg_B)_{B\in\mc I}$ have independent standard
real Gaussian coordinates, and let
$F:\R^{\mc I}\to\C$ be a polynomial. For complex arguments,
use its polynomial extension. Let $\Pi_d$ denote the orthogonal
projection onto Gaussian Hermite degree $d$.

For an ordered tuple $\vec B=(B_1,\ldots,B_d)\in\mc I^d$,
define
\[
  R_{\vec B}(F)
    :=\E_{\vec \bsg}[D_{B_1}\cdots D_{B_d}F(\vec \bsg)],
  \qquad D_B=\frac{\partial}{\partial g_B},
\]
and let $\vec\alpha(\vec B)$ be its multiplicity vector:
\[
  \alpha_B(\vec B)=|\{j:B_j=B\}|,\qquad
  \vec\alpha!=\prod_{B\in\mc I}\alpha_B!.
\]
Then, for every $d\geq1$,
\begin{equation}
  R_{\vec B}(F)
    =\sqrt{\vec\alpha(\vec B)!}\,
       \E_{\vec \bsg}\left[
         F(\vec \bsg)\mathsf{Herm}_{\vec\alpha(\vec B)}(\vec \bsg)
       \right],
  \label{eq:shift_hermite_coefficient}
\end{equation}
and
\begin{equation}
  \|\Pi_dF\|_{L^2}^2
    =\frac1{d!}
       \sum_{\vec B\in\mc I^d}|R_{\vec B}(F)|^2.
  \label{eq:ordered_hermite_shared}
\end{equation}
\end{lemma}

\begin{proof}
Let $\varphi$ be the standard Gaussian density on
$\R^{\mc I}$. The normalized probabilists' Hermite
polynomials satisfy
\[
  D^{\vec\alpha}\varphi(\vec g)
    =(-1)^{|\vec\alpha|}
       \sqrt{\vec\alpha!}\,
       \mathsf{Herm}_{\vec\alpha}(\vec g)\varphi(\vec g).
\]
Integration by parts in each coordinate therefore gives
\[
  \E_{\vec \bsg}[D^{\vec\alpha}F(\vec \bsg)]
    =
    \sqrt{\vec\alpha!}\,
      \E_{\vec \bsg}
        [F(\vec \bsg)\mathsf{Herm}_{\vec\alpha}(\vec \bsg)].
\]
There are no boundary terms, since $F$ and all its derivatives
are polynomials and the Gaussian density decays faster than
any polynomial. Taking $\vec\alpha=\vec\alpha(\vec B)$ proves
\Cref{eq:shift_hermite_coefficient}, including repeated indices.

A multiplicity vector $\vec\alpha$ of total degree $d$ has
exactly $d!/\vec\alpha!$ ordered representatives. Hence
\begin{align*}
  \frac1{d!}
    \sum_{\vec B\in\mc I^d}|R_{\vec B}(F)|^2
  &=
    \frac1{d!}
    \sum_{|\vec\alpha|=d}
      \frac{d!}{\vec\alpha!}\,
      \vec\alpha!\,
      \left|
        \E_{\vec \bsg}
          [F(\vec \bsg)\mathsf{Herm}_{\vec\alpha}(\vec \bsg)]
      \right|^2\\
  &=
    \sum_{|\vec\alpha|=d}
      \left|
        \E_{\vec \bsg}
          [F(\vec \bsg)\mathsf{Herm}_{\vec\alpha}(\vec \bsg)]
      \right|^2
    =\|\Pi_dF\|_{L^2}^2.
\end{align*}
The last equality is Hermite orthonormality, proving
\Cref{eq:ordered_hermite_shared}.
\end{proof}

\begin{lemma}[Schatten H\"older inequality {\cite[Theorems~6.32 and~7.2]{tropp2022matrix}}]\label{alltemp:lem:schatten-holder}
Use the normalized Schatten norms $\|X\|_p=(\operatorname{Tr}|X|^p)^{1/p}$ for $1\le p<\infty$, and $\|X\|_\infty=\|X\|_{\rm op}$. For matrices $X_1,\ldots,X_\ell$ on the same space and $p_j\in[1,\infty]$ with $\sum_jp_j^{-1}=1$,
\begin{equation}\label{alltemp:eq:schatten-holder}
 \left|\operatorname{Tr}(X_1\cdots X_\ell)\right|
 \le\prod_{j=1}^{\ell}\|X_j\|_{p_j}.
\end{equation}
\end{lemma}
The multi-factor form follows by iterating the product inequality in the cited result. The normalized-trace convention preserves the inequality because both sides acquire the same factor $d_n^{-1}$.

\begin{lemma}[Gaussian concentration {\cite[Theorem~3.25]{vanhandel2016probability}}]\label{alltemp:lem:gaussian-concentration}
Let $G\in\mathbb R^d$ be a standard Gaussian vector and let $f:\mathbb R^d\to\mathbb R$ be globally $L$-Lipschitz with respect to the Euclidean norm, where $L>0$. Then
\begin{equation}\label{alltemp:eq:gaussian-concentration-tool}
 \begin{aligned}
  \mathbb E e^{t(f(G)-\mathbb E f(G))}&\le e^{t^2L^2/2}
   &&(t\in\mathbb R),\\
  \mathbb P[f(G)-\mathbb E f(G)\le-u]&\le e^{-u^2/(2L^2)}
   &&(u\ge0).
 \end{aligned}
\end{equation}
\end{lemma}
The cited theorem gives $L^2$-subgaussianity in the sense of \cite[Definition~3.5]{vanhandel2016probability}, which is the displayed moment-generating-function bound. Applying the exponential Markov inequality to $-f$ gives the lower tail.

\begin{lemma}[Araki--Lieb--Thirring inequality {\cite{araki1990inequality}}]\label{alltemp:lem:alt}
For positive semidefinite matrices $D,E$ on the same space and every real $r\ge1$,
\begin{equation}\label{alltemp:eq:alt}
 \operatorname{Tr}\!\left[(E^{1/2}DE^{1/2})^r\right]
 \le\operatorname{Tr}\!\left(E^{r/2}D^rE^{r/2}\right).
\end{equation}
\end{lemma}

\newcommand{\Conn}{\operatorname{Conn}}
\section{Wick expansion and Koteck\'y--Preiss}
\label{sec:shared_wick_kp}

In this section, we introduce the definitions and objects needed for our Wick expansion, extending
the Wick-pair polymer construction of \cite{zlokapa2026rigorous}. The key
technical tool used for both the thermal estimates in \Cref{lem:thermal_estimates}
and the high-degree bound \Cref{lem:degree_bound} is the Koteck\'y--Preiss theorem
\cite{kotecky1986cluster}, in the form of \cite[Theorem~1]{ueltschi2004cluster}. We use one
version of that theorem, but keep the three applications separate. The relevant difference is
whether we are changing Wick covariances, deleting mixed polymers, or excluding the support of an
observable.

Recall that $n \geq 8$ is even and $\sigma_n=\sqrt{6/n^3}$. Constants denoted by $C$ are absolute
and may increase from line to line. A subscript, as in $C_A$, indicates the allowed dependence.

\subsection{Wick diagrams and component placement}
\label{sec:wick_placement}

As a first step, expand each partition function as
\begin{equation}
  Z(\beta,\vec g)=\sum_{r\geq0}\frac{(-\beta)^r\sigma_n^r}{r!}
     \sum_{I_1,\ldots,I_r\in\mc I_4}
       \left(\prod_{j=1}^r g_{I_j}\right)
       \Tr\left[\prod_{j=1}^r\Gamma_{I_j}\right].
  \label{eq:z_taylor_appendix}
\end{equation}
Note that for a single summand above, there are $r$ ``slots'' which are labeled by a quartic
Majorana $I \in \mc I_4$. In our application, we will have as many as four replicas; we index the
replicas by an integer, referring to the index as a ``color'' $c$ and the corresponding partition
function as $Z_c(\beta_c, \vec g_c)$. Then, a summand from each replica is indexed by $\vec r =
(r_c)_c$, where $r_c$ is the number of slots for color $c$. We denote the total set of slots by 
\[
	\mathsf{Slots}_C(\vec r)=\bigsqcup_{c\in[C]}(\{c\}\times[r_c])\,.
\]
The Gaussian expectation is computed by the Wick--Isserlis identity \cite{isserlis1918formula}.

\begin{theorem}[Wick--Isserlis]
\label{thm:wick}
Let $X_1, \dots, X_r$ be centered jointly Gaussian real random variables. Then
\[
	\E\Brac{\prod_{j=1}^rX_j} =\sum_{\pi\in\PM([r])} \prod_{\{i,j\} \in \pi} \E[X_i X_j]\,.
\]
$\PM([r])$ denotes the set of perfect matchings on $[r]$, with $\PM([r]) = \emptyset$ when $r$ is
odd, and $\PM(\emptyset)$ contains the empty matching.
\end{theorem}

In the replicated Gaussian disorder, coordinates with different quartic labels are independent, so a covariance between two slots labeled by $I,J \in \mc I_4$ is $0$ unless $I=J$. Therefore,
after choosing a matching, each matched pair of slots is assigned the same label. A slot vector
$\vec r$, matching $\pi \in \PM(\mathsf{Slots}_C)$, and the labeling $\mc A$ of matching edges to
$\mc I_4$ define a \emph{Wick diagram}.

\begin{definition}[Wick diagram and polymer]
A Wick diagram is a tuple $\mc D = (\vec r, \pi, \mc A)$, where $\pi \in \PM(\mathsf{Slots}_C(\vec
r))$ and $\mc A: \pi \rightarrow \mc I_4$. For an edge $e=\{(c,i),(c',j)\}$, its \emph{colored}
support is
\[
	S_e=\Paren{\{c\} \times \mc A(e)} \cup \Paren{\{c'\} \times \mc A(e)}\,.
\]
The Wick diagram defines a graph, with vertices given by the edges of the matching $\pi$ and edges
between $e, e' \in \pi$ when $S_e \cap S_{e'} \neq \emptyset$. A \emph{polymer} $\mc P$ is a
nonempty connected Wick diagram. Its support is $\supp(\mc P)=\bigcup_{ e \in \pi} S_e$. We write,
\[
	\mc P_1 \sim \mc P_2
\]
to indicate that $\mc P_1, \mc P_2$ have disjoint supports; we say such polymers are
\emph{compatible}.
\end{definition}

Given a diagram $\mc D$ and a color $c$, define the ordered word
$\mathsf{Mon}_c(\mc D)=\prod_{i=1}^{r_c}\Gamma_{\mc A(e_{\mc D}(c,i))}$,
where $e_{\mc D}(c,i)$ is the matching edge incident to slot $(c,i)$. For a pair of colors $c, c'$, we write $\kappa_{c, c'}$ to
denote the covariance between any pair of matched slots. This will either be $0$ if the replicas
are independent or $1$ otherwise. Then the contribution of a diagram with slot vector $\vec r$ to the disorder average of
$\prod_{c \in [C]} Z_c(\beta_c, \vec\bsg_c)$ is given by the activity
\begin{savedequation}{diagramactivity}
	\label{eq:diagram_activity}
	\mathsf{Act}(\mc D; \beta_1, \dots, \beta_C) = \Paren{\prod_c \frac{(-\beta_c)^{r_c}}{r_c!}}
	\sigma_n^{2|\pi|} \Paren{\prod_{e=\{(c,i), (c',j)\} \in \pi} \kappa_{cc'}}
	\prod_c \Tr[\mathsf{Mon}_c(\mc D)]\,.
\end{savedequation}
Comparing this definition with \Cref{eq:z_taylor_appendix} and using \Cref{thm:wick} we have that,
\[
	\E_{(\vec\bsg_c)_{c\in[C]}} \prod_{1 \leq c \leq C} Z_c(\beta_c, \vec\bsg_c) =
	\sum_{\text{Diagrams } \mc D} \mathsf{Act}(\mc D; \beta_1, \dots, \beta_C)\,.
\]
The trace factors obey the normalization and symmetric-difference rule established in
\Cref{prop:polymer_counting} below.

Next, we will show how to decompose diagrams into connected components, each of which is a polymer.
View the expression
\[
	M(\vec \beta) = \E_{(\vec\bsg_c)_{c\in[C]}} \prod_{c \in [C]} Z_c(\beta_c, \vec\bsg_c)
\]
as a power series in formal variables $\vec \beta = (\beta_c)_c$.  The next lemma shows that
coefficient-wise, $M(\vec \beta)$ is equivalent to the \emph{hard-core gas} $\Xi(\vec \beta)$.

\begin{lemma}[Decomposition into connected polymers]
\label{lem:component_placement}
The $C$-replica partition function $M(\vec \beta)$ is, coefficientwise, the hard-core gas
\begin{equation}
	\label{eq:shared_gas}
	\Xi(\beta_1, \dots, \beta_C)= \sum_{q \geq 0}\frac1{q!} \sum_{\substack{\mc P_1, \dots, \mc
	P_q \in \mathsf P\\ \mc P_i\sim\mc P_j\ (i < j)}} \prod_{j=1}^q \mathsf{Act}(\mc P_j;
	\beta_1, \dots, \beta_C)\,.
\end{equation}
\end{lemma}

\begin{proof}
Fix a diagram $\mc D=(\vec r,\pi,\mc A)$ and recall the formula for its activity
\Cref{eq:diagram_activity}. Let $E_1, \dots, E_q$ be the vertex sets of the connected components of
the graph induced by $\mc D$. Recall that the vertices of this graph are the edges of the matching
$\pi$, so the sets $E_i$ partition $\pi$. Each of these connected components will induce a polymer
$\mc P_i$. Note that since the components are disconnected, $\mc P_i \sim \mc P_j$, for all distinct $i,j\in[q]$.

For each component $j$, let $r_{c,j}$ be the total number of slots of color $c$ belonging to component $j$. Let
$\mathsf{ord}_{j,c} : [r_{c,j}] \rightarrow \mathsf{Slots}_C$ list these slots in inherited increasing order, so that
$\mathsf{ord}_{j,c}(i)$ indicates the $i$-th slot in the ordering. In the context of a single
connected polymer, we will abuse notation and write $(c,i)$ to refer to the slot $\mathsf{ord}_{j,
c}(i)$.

Conversely, fix an ordered tuple of compatible polymers
$\mc P_1,\ldots,\mc P_q$. We describe how to place the slots for each individual polymer into a
global slot vector. Since $r_{c,j}$ is the number of color-$c$ slots for polymer $\mc P_j$, we
define $r_c = \sum_{j \in [q]} r_{c,j}$. A placement is a collection of strictly increasing maps
$\mathsf{Place}_{c, j} : [r_{c,j}] \rightarrow [r_c]$ with,
\[
	\bigsqcup_{j \in [q]} \mathsf{Image}(\mathsf{Place}_{c,j}) = [r_c]\,.
\]
Then, we can replace every local slot $(c,k)$ of $\mc P_j$ by $(c, \mathsf{Place}_{c,j}(k))$ which
yields a full Wick diagram.

These two operations are inverses: the maps $\mathsf{ord}_{j,c}$ partition the diagram into
connected polymers and $\mathsf{Place}$ places each polymer in the context of a global diagram.

\paragraph{Factorization of the weights.}
Next we show that the activities themselves factorize across the connected polymers. Suppress the
inverse temperatures $\beta_1, \dots, \beta_C$ in the activities for notational convenience. We
first consider all expressions not including the factorial in the diagram activity above, and write
\[
	\mathsf{Act}(\mc D) = \frac 1 {\prod_c r_c!} \widetilde{\mathsf{Act}}(\mc D)
\]
By the above argument, any diagram is a combination of a set of polymers $\mc P_1, \dots \mc P_q$ as
well as a set of placement functions $\mathsf{Place}_{c,j}$ for all $(c,j) \in [C] \times [q]$.
Moreover, observe that all terms in $\widetilde{\mathsf{Act}}(\mc D)$ depend only on the set $\mc
P_1, \dots, \mc P_q$ and not the placement. We argue these coefficients decompose into contributions
from each $\mathsf{Act}(\mc P_i)$.

Start with $\mathsf{Mon}_c$ for a fixed color $c$. We want to show
\[
	\Tr[\mathsf{Mon}_c(\mc D)] =\prod_{j=1}^q \Tr[\mathsf{Mon}_c(\mc P_j)]\,.
\]
By definition, the supports of each $\mc P_j$ are disjoint and thus each of their contributions
commute (two disjoint quartic Majoranas commute). We may therefore group $\mathsf{Mon}_c$ by each
polymer.
\[
	\mathsf{Mon}_c(\mc D)=\prod_{j=1}^q \mathsf{Mon}_c(\mc P_j)\,.
\]
To see that the traces factorize, let $Q_j$ be the symmetric difference of the quartic labels in
$\mathsf{Mon}_c(\mc P_j)$. The sets $Q_j$ are pairwise disjoint, because the component supports
are disjoint in color $c$. The global trace is nonzero only when
$\mathop{\triangle}_jQ_j=\emptyset$, which is equivalent to every $Q_j$ being empty.
In that case each component word is proportional to $\Id$, and its scalar contribution factors
out of the normalized trace.

Next, covariance contributions are always internal to a connected component so these factorize.
Similarly, the identities $|\pi| = \sum_j |E_j|$ and $(-\beta_c)^{r_c} = \prod_j
(-\beta_c)^{r_{c,j}}$  show that these scalar terms factorize.

Finally, we count the contribution from the placements, which will account for the factorial
term. The diagram associated with this set of polymers has factorial weight $\prod_c \frac 1
{r_c!}$, but each placement is a new diagram with the same activity $\widetilde{\mathsf{Act}}(\mc
D)$, of which there are $(\#\text{ possible placements})$. Therefore the ``weight'' that should be
associated with the set of polymers $\mc P_1, \dots, \mc P_q$ is,
\[
	(\#\text{ possible placements}) \cdot \prod_c \frac 1 {r_c!} = \Paren{\prod_c \binom{r_c}{r_{c,1},\ldots,r_{c,q}}} \cdot \prod_c \frac 1 {r_c!} = \prod_{c,j} \frac 1 {r_{c,j}!}\,.
\]
This is exactly equal to the ``weight'' contributed by each polymer activity $\mathsf{Act}(\mc P_j;
\beta_1, \dots, \beta_C)$. Finally, allowing for all orderings yields the remaining $\frac 1 {q!}$
term in \Cref{eq:shared_gas}.
\end{proof}

\subsection{Koteck\'y--Preiss theorem}
\label{sec:kp_shared}

We state the theorem for scalar activities. Any temperature or Taylor parameters
are fixed in its statement. A common majorant on a parameter domain will also
give normal convergence there.

\begin{theorem}[Koteck\'y--Preiss theorem]
\label{thm:kp}
Let $\mathsf P$ be a set of polymers with nonempty supports contained in a finite set $X$, and let
$\mathsf{Act}:\mathsf P \rightarrow \C$ be the polymer activities. Suppose that there exists a
majorant $w : \mathsf P \rightarrow \mathbb R$ such that for some $0 \leq \delta \leq a_\mathrm{KP}$ with $a_\mathrm{KP} >
0$, 
\begin{equation}
	\label{eq:shared_kp}
	|\mathsf{Act}(\mc P)| \leq w(\mc P),\qquad \sup_{x\in X}\sum_{\mc P:x\in\supp(\mc P)}
	     w(\mc P)e^{a_\mathrm{KP}|\supp(\mc P)|}\leq\delta\,.
\end{equation}
Let $\mc G_\mathrm{conn}$ be the set of connected graphs with vertex set $[q]$. Define the
connected-graph coefficient
\[
	\Conn^{\#}(\mc P_1, \dots, \mc P_q) = \sum_{G \in \mc G_\mathrm{conn}}
	     \prod_{\{i,j\}\in E(G)}(-\mathbf 1_{\mc P_i \not\sim \mc P_j})\,,
\]
with $1$ for $q=1$. Then the gas $\Xi$ defined in \Cref{eq:shared_gas}, and the Mayer series
\begin{equation}
	\label{eq:shared_mayer}
	\mathsf{Mayer}(\vec\beta)= \sum_{q\geq1}\frac1{q!}\sum_{\mc P_1, \dots, \mc
	P_q \in \mathsf P} \Conn^{\#}(\mc P_1,\ldots,\mc P_q) \prod_{j=1}^q\mathsf{Act}(\mc P_j;
	\beta_1, \dots, \beta_C)
\end{equation}
converge absolutely and satisfy $\Xi=e^\mathsf{Mayer} \neq 0$. The rooted estimate is
\begin{equation}
  \sum_{q\geq1}\frac1{(q-1)!}\sum_{\mc P_2,\ldots,\mc P_q\in\mathsf P}
    |\Conn^{\#}(\mc P,\mc P_2,\ldots,\mc P_q)|
      \prod_{j=2}^q w(\mc P_j)
     \leq e^{a_\mathrm{KP}|\supp(\mc P)|}.
  \label{eq:shared_pinned}
\end{equation}
For any $\mc D\subseteq\mathsf P$, the absolute Mayer subseries containing
at least one polymer from $\mc D$ is at most
$\sum_{\mc P\in\mc D}w(\mc P)e^{a_\mathrm{KP}|\supp(\mc P)|}$.
\end{theorem}
\begin{proof}
The site condition gives
\begin{equation}
	\label{eq:site_conditions}
	\sum_{\mc Q \not\sim \mc P} w(\mc Q)e^{a_\mathrm{KP}|\supp(\mc Q)|} \leq \delta |\supp(\mc P)| \leq
	a_\mathrm{KP}|\supp(\mc P)|, \qquad \sum_{\mc P} w(\mc P) e^{a_\mathrm{KP}|\supp(\mc P)|} \leq \delta|X| < \infty\,.
\end{equation}
These verify the hypotheses in
\cite[Theorem 1]{ueltschi2004cluster}. Equation (4) therein with the substitutions
\[
	\mu = \sum_{\mc Q} w(\mc Q) \delta_{\mc Q}\,, \qquad a(\mc P)=a_\mathrm{KP}|\supp(\mc P)|\,, \qquad \phi
	= \frac 1 {q!} \Conn^\#
\]
yields \Cref{eq:shared_pinned}. To get absolute convergence let $M$ be the sum of the absolute
values of the terms of the Mayer series. Then,
\begin{align*}
	M &= \sum_{q \geq 1}\frac 1 {q!}\sum_{\mc P_1, \dots, \mc P_q \in \mathsf P} |\Conn^{\#}(\mc
	P_1,\ldots,\mc P_q)| \prod_{j=1}^q|\mathsf{Act}(\mc P_j; \beta_1, \dots, \beta_C)|\\
	&\leq \sum_{q \geq 1}\frac 1 {q!}\sum_{\mc P_1, \dots, \mc P_q \in \mathsf P} |\Conn^{\#}(\mc
	P_1,\ldots,\mc P_q)| \prod_{j=1}^q w(\mc P_j)\\
	&\leq \sum_{\mc P \in \mathsf P} w(\mc P)\sum_{q \geq 1} \frac 1 {(q-1)!}\sum_{\mc P_2,
	\dots, \mc P_q \in \mathsf P} |\Conn^{\#}(\mc P,\mc P_2,\ldots,\mc P_q)| \prod_{j=2}^q w(\mc
	P_j)\\
	&\leq \sum_{\mc P \in \mathsf P} w(\mc P) e^{a_\mathrm{KP} | \supp(\mc P)|} \leq \delta |X| <
	\infty\,.
\end{align*}
In the last line we used \Cref{eq:shared_pinned} and \Cref{eq:site_conditions}. Additionally,
we can bound the absolute gas series:
\begin{align*}
	\Xi_{\mathrm{abs}} &:= \sum_{q \geq 0} \frac 1 {q!} \sum_{\substack{\mc P_1, \dots, \mc P_q
	\in \mathsf P\\ \mc P_i \sim \mc P_j\ (i\ne j)}} \prod_{i = 1}^q |\mathsf{Act}(\mc P_i)|\\
	&\leq \sum_{q \geq 0} \frac 1 {q!} \sum_{\mc P_1, \dots, \mc P_q \in \mathsf P} \prod_{i =
	1}^q |\mathsf{Act}(\mc P_i)|\\
	&=\sum_{q \geq 0} \frac 1 {q!} \Paren{\sum_{\mc P \in \mathsf P} |\mathsf{Act}(\mc P)|}^q\\
	&= \exp\Paren{\sum_{\mc P \in \mathsf P} |\mathsf{Act}(\mc P)|}
	\leq \exp\Paren{\sum_{\mc P \in \mathsf P} w(\mc P)}<\infty\,.
\end{align*}

\cite[Theorem 1]{ueltschi2004cluster} applied to the complex activities $\mathsf{Act}$ also gives
$\Xi = e^{\mathsf{Mayer}} \neq 0$. 

Finally, consider the Mayer subseries,
\[
	\mathsf{Mayer}_{\mc D} = \sum_{q \geq 1} \frac 1 {q!} \sum_{\substack{\mc P_1, \dots, \mc
	P_q\\ \exists j \st \mc P_j \in \mc D}} \Conn^{\#}(\mc P_1, \dots, \mc P_q) \prod_{i=1}^q
	\mathsf{Act}(\mc P_i)\,.
\]

The absolute subseries is bounded as follows:
\begin{align*}
	(\mathsf{Mayer}_{\mc D})_\mathrm{abs} &= \sum_{q \geq 1} \frac 1 {q!} \sum_{\substack{\mc
	P_1, \dots, \mc P_q\\ \exists j \st \mc P_j \in \mc D}} |\Conn^{\#}(\mc P_1, \dots, \mc
	P_q)| \prod_{i=1}^q |\mathsf{Act}(\mc P_i)|\\
	&\leq \sum_{q \geq 1} \frac 1 {q!} \sum_{j = 1}^q \sum_{\mc P_1, \dots, \mc P_q} \mathbf
	1_{\mc P_j \in \mc D} |\Conn^{\#}(\mc P_1, \dots, \mc P_q)| \prod_{i=1}^q w(\mc P_i)\\
	&\leq \sum_{q \geq 1} \frac q {q!} \sum_{\mc P \in \mc D} w(\mc P)\sum_{\mc P_2, \dots, \mc
	P_q} |\Conn^{\#}(\mc P, \mc P_2, \dots, \mc P_q)| \prod_{i=2}^q w(\mc P_i)\\
	&= \sum_{\mc P \in \mc D}w(\mc P) \sum_{q \geq 1} \frac 1 {(q-1)!}  \sum_{\mc P_2, \dots,
	\mc P_q} |\Conn^{\#}(\mc P, \mc P_2, \dots, \mc P_q)| \prod_{i=2}^q w(\mc P_i)\\
	&\leq \sum_{\mc P \in \mc D} w(\mc P) e^{a_\mathrm{KP} |\supp(\mc P)|}\,.
\end{align*}
In the last line, we used that the inner series is bounded by the pinned
estimate \Cref{eq:shared_pinned}.
\end{proof}

As stated in the introduction, all of our applications of Koteck\'y--Preiss involve ratios of
two gases $\Xi_0, \Xi_1$. After verifying the KP conditions, we have that $\Xi_0 =
e^{\mathsf{Mayer}_0}$ and $\Xi_1 = e^{\mathsf{Mayer}_1}$ and thus bounding the ratio reduces to
bounding $|\mathsf{Mayer}_0 - \mathsf{Mayer}_1|$. The following lemma gives the required bounds.
\begin{cor}[Comparison and exclusion]
\label{cor:kp_comparison}
Suppose two activity families $\mathsf{Act}_0,\mathsf{Act}_1$ have the same majorant
$w$ in \Cref{eq:shared_kp}. The corresponding Mayer series $\mathsf{Mayer}_0,\mathsf{Mayer}_1$
satisfy
\begin{equation}
	\label{eq:kp_comparison}
	|\mathsf{Mayer}_1 - \mathsf{Mayer}_0| \leq \sum_{\mc P} |\mathsf{Act}_1(\mc P) -
	\mathsf{Act}_0(\mc P)| e^{a_\mathrm{KP}|\supp(\mc P)|}\,.
\end{equation}
Deleting a family $\mc D$ therefore changes the logarithm by at most $\sum_{\mc P\in\mc
D}w(\mc P)e^{a_\mathrm{KP}|\supp(\mc P)|}$. If $\Xi[S]$ is the gas restricted to polymers disjoint from
$S\subseteq X$, then
\begin{equation}
	\label{eq:kp_exclusion}
	\Abs{\frac{\Xi[S]}{\Xi}} \leq e^{\delta|S|}\,.
\end{equation}
\end{cor}
\begin{proof}
For a tuple of length $q$, telescope the difference of its two activity products.
Its absolute value is bounded by
\[
	\sum_{j=1}^q |\mathsf{Act}_1(\mc P_j)-\mathsf{Act}_0(\mc P_j)| \prod_{i \neq j}w(\mc P_i)\,.
\]
Multiply by $|\Conn^{\#}|/q!$ and sum. Symmetry gives the factor $q/q!=1/(q-1)!$, and
\Cref{eq:shared_pinned} proves \Cref{eq:kp_comparison}. The resulting sum is finite, since the
difference is bounded by $2w$. Deletion sets the corresponding activities to zero. For the exclusion
ratio, delete polymers meeting $S$. Their weighted mass is at most $\delta|S|$ by the site condition
in \Cref{eq:site_conditions}. Writing $\mathsf{Mayer}[S]$ for the restricted cluster logarithm gives
$\Xi[S]/\Xi=\exp(\mathsf{Mayer}[S]-\mathsf{Mayer})$, proving \Cref{eq:kp_exclusion}.
\end{proof}

\subsection{The shared label counts}
\label{sec:shared_counts}

To verify the KP conditions, we sum nonnegative majorants of polymer activities.
The common bookkeeping step is to separate the ordered slots and Wick matching from their
physical labels. We call the former data an \emph{unlabeled diagram} or \emph{skeleton}:
only the physical labels are forgotten, not the slot positions or their order. In particular,
we do not identify skeletons up to graph isomorphism. Connectedness is a condition on the
labeling, so a skeleton is not itself a polymer.

\begin{proposition}[Counting polymers by skeletons and labels]
\label{prop:polymer_counting}
Fix a finite set of Taylor blocks $\mc B$, with nonnegative weights $(b_c)_{c\in\mc B}$,
and a set of allowed slot vectors $\mc R\subseteq\mathbb Z_{\geq0}^{\mc B}$.
Write $\mathsf{Slots}(\vec r)=\bigsqcup_{c\in\mc B}(\{c\}\times[r_c])$.
For $h,j\geq0$, let $\mathsf{Skel}_{h,j}(\mc R)$ consist of the triples
$\mc D=(\vec r,M,\pi)$ with
\[
  \vec r\in\mc R,\qquad \sum_cr_c=2h+j,\qquad
  M\subseteq\mathsf{Slots}(\vec r),\quad |M|=j,\qquad
  \pi\in\PM(\mathsf{Slots}(\vec r)\setminus M).
\]
Here $M$ records the perturbation slots. The assignment of blocks to traces, and any distinguished
observable-slot positions, are fixed in advance; observable slots are not Taylor slots. A labeling assigns a quartic
label to each pair and each perturbation slot, and also specifies any observable-slot labels that are not
fixed in advance. Every trace monomial $W=\Gamma_{I_1}\cdots\Gamma_{I_q}$ in these diagrams is a
product of quartic Majorana monomials so $\Tr[W]\in\{0,1,-1\}$ and $ |\Tr[W]|=\mathbf
1_{I_1\triangle \dots \triangle I_q=\emptyset}$. Thus every nonzero trace contributes a factor of $1$ to an absolute activity.

For each skeleton choose a set $\mathsf{Lab}(\mc D)$ of admissible physical labelings and a
nonnegative label weight $v(\mc D,\lambda)$, and put
\begin{equation}
  \mathsf{LabSum}_v(\mc D)
    :=\sum_{\lambda\in\mathsf{Lab}(\mc D)}v(\mc D,\lambda).
  \label{eq:shared_label_sum}
\end{equation}
The admissibility conditions may include connectedness, a pinned site, nonzero traces, or
restrictions on observable slots. If observable-slot labels are also summed, they are part of $\lambda$ and
their weights are included in $v$.

Use the physical or colored support prescribed by the application. For $t\geq0$,
let $\mathcal M_{h,j}$ be the sum over these labeled diagrams of
\[
  \left(\prod_{c\in\mc B}\frac{b_c^{r_c}}{r_c!}\right)
  \sigma_n^{2h}v(\mc D,\lambda)e^{t|\supp(\mc D,\lambda)|}.
\]
If every such diagram has support size at most $s_{h,j}$, then
\begin{equation}
\begin{aligned}
  \mathcal M_{h,j}
    &={}\sigma_n^{2h}
      \sum_{\mc D\in\mathsf{Skel}_{h,j}(\mc R)}
        \left(\prod_c\frac{b_c^{r_c}}{r_c!}\right)
        \sum_{\lambda\in\mathsf{Lab}(\mc D)}
          v(\mc D,\lambda)e^{t|\supp(\mc D,\lambda)|}\\
    &\leq \sigma_n^{2h}e^{ts_{h,j}}
      \sum_{\mc D\in\mathsf{Skel}_{h,j}(\mc R)}
        \left(\prod_c\frac{b_c^{r_c}}{r_c!}\right)
        \mathsf{LabSum}_v(\mc D).
\end{aligned}
\label{eq:shared_skeleton_decomposition}
\end{equation}
Consequently, a uniform bound $\mathsf{LabSum}_v(\mc D)\leq L_{h,j}$ on the allowed
skeletons gives
\begin{equation}
  \mathcal M_{h,j}
    \leq
    \underbrace{\frac{(\sum_cb_c)^{2h+j}}{j!\,2^hh!}}_
      {\textnormal{Taylor weights and skeletons}}
    \quad\underbrace{\sigma_n^{2h}}_{\textnormal{Wick covariances}}
    \quad\underbrace{L_{h,j}}_{\textnormal{weighted label sum}}
    \quad\underbrace{e^{ts_{h,j}}}_{\textnormal{support weight}}.
  \label{eq:shared_counting_bound}
\end{equation}
\end{proposition}
\begin{proof}
First, the Majorana relations give $W=\varepsilon\Gamma_Q$ with
$\varepsilon\in\{1,-1\}$ and $Q=I_1\triangle\cdots\triangle I_q$. If $Q=\varnothing$,
then $\Tr[W]=\varepsilon$ because $\Tr[\Id]=1$. Otherwise, $|Q|$ is positive and even.
For any $i\in Q$, conjugation by $\psi_i$ gives
$\psi_i\Gamma_Q\psi_i=(-1)^{|Q|-1}\Gamma_Q=-\Gamma_Q$. Invariance of the trace under
conjugation therefore implies $\Tr[\Gamma_Q]=0$. This proves the trace rule and its
consequence for absolute activities.

Forgetting the physical labels partitions the labeled diagrams by their unique skeleton,
which gives the equality in \Cref{eq:shared_skeleton_decomposition}. The support bound gives
its inequality. Replace each label sum by $L_{h,j}$; only then enlarge the remaining sum over
slot vectors to all nonnegative $\vec r$ with $\sum_cr_c=2h+j$. For each such vector, there are
$\binom{2h+j}{j}(2h-1)!!$ choices of $M$ and $\pi$. The slot identity
\Cref{eq:shared_slot_count}, proved below, gives \Cref{eq:shared_counting_bound}.
\end{proof}

The label sums may depend on the skeleton; only their upper bound must be uniform.
Thus the proposition leaves the connectivity and parity estimates to each application.
For a sum over compatible \emph{families} of polymers, first use the component-placement
identity to express the mass as a sum over diagrams with global Taylor factorials
$\prod_c1/r_c!$. We make this step explicit for the observable-slot configurations in
\Cref{lem:observable_diagram_reduction}. The following lemma supplies the slot identity and two label
bounds used in these applications.

\begin{lemma}[Slots and quartic labels]
\label{lem:shared_counts}
For real, nonnegative $b_c$ and integers $h,j\geq0$,
\begin{equation}
  \sum_{\substack{\vec r\geq0\\\sum_cr_c=2h+j}}
     \left(\prod_c\frac{b_c^{r_c}}{r_c!}\right)
        \binom{2h+j}{j}(2h-1)!!
    =\frac{(\sum_cb_c)^{2h+j}}{j!\,2^h h!}.
  \label{eq:shared_slot_count}
\end{equation}
Fix at most four colors, each corresponding to a trace replica. Each $e \in [h]$ corresponds to an
edge with endpoints on slots with colors $c_{e,1}, c_{e,2}$. Let $C_e = \{c_{e, 1}, c_{e,2}\}$. We
bound the number of labelings $\mc A : [h] \rightarrow \mc I_4$.
\begin{enumerate}
	\item \emph{(Pinned physical site)} Let $h \geq 1$. For a polymer with a fixed site $x_0$
		with color $c_0$, the number of connected labelings $\mc A$ containing this site is
		at most
		\begin{equation}
			\label{eq:pinned_physical_site}
			C^h h! n^{3h}\,.
		\end{equation}
	\item \emph{(Pinned edges)} Fix a nonempty edge set $\varnothing\ne B\subseteq[h]$. The number of connected labelings
		$\mc A$ which satisfy $\mathop{\triangle}_{e\in B}\mc A(e)=\varnothing$ is zero for
		$|B|=1$, and otherwise at most
		\begin{equation}
			\label{eq:shared_xor_count}
			C^h h^{h+4}n^{3h-2}\,.
		\end{equation}
\end{enumerate}
\end{lemma}
\begin{proof}
We first prove \Cref{eq:shared_slot_count}. Put $s=2h+j$. By the multinomial theorem,
\[
  \sum_{\substack{\vec r\geq0\\\sum_cr_c=s}}
    \prod_c\frac{b_c^{r_c}}{r_c!}
  =
  \frac{(\sum_cb_c)^s}{s!}.
\]
For each slot vector, there are $\binom{s}{j}$ choices of the
unpaired positions. The remaining $2h$ positions have
\[
  (2h-1)!!=\frac{(2h)!}{2^h h!}
\]
perfect matchings: order these positions and group consecutive
positions into pairs, then divide by the two orientations of each
pair and the $h!$ orderings of the pairs. Therefore,
\[
  \frac{(\sum_cb_c)^s}{s!}
  \binom{s}{j}\frac{(2h)!}{2^h h!}
  =
  \frac{(\sum_cb_c)^{2h+j}}{j!\,2^h h!},
\]
which proves \Cref{eq:shared_slot_count}. The convention
$(-1)!!=1$ includes $h=0$.

Finally, we count labelings. Recall that the labeling $\mc A$ induces a graph on $[h]$ where $e, e' \in [h]$ are connected when
$\paren{C_e\times\mc A(e)} \cap \paren{C_{e'} \times \mc A(e')} \neq \emptyset$. We
restrict to $\mc
A$ which generate a connected graph and thus we can associate a spanning tree $T$. Let $\mc T$ be the set of possible
spanning trees of $[h]$, of which there are at most $|\mc T| \leq h^{h-2}$. There are at most $h
\cdot h^{h-2}$ possible \emph{rooted} trees. We count labelings by starting with a tree $T
\in \mc T$ and counting possible assignments to its vertices.

\paragraph{Pinned physical site. } In this case, fix an edge $e_0 \in [h]$ whose labeling $\mc
A(e_0)$ contains $(c_0, x_0)$ as physical colored site. There are at most $h$ choices for $e_0$ and four possible positions for $x_0$
within $I_0 = \mc A(e_0)$. The remaining three sites have at most $n^3$ choices. We now traverse the spanning tree in a fixed order (e.g., by picking a connected edge with
smallest indices in $[h]$). For each subsequent edge, at least one out of the four physical sites is
fixed (by the connectivity condition), and thus there are a total of $n^{3h}$ remaining choices
across all edges. The last freedom we have is in the positions of the index shared by $\mc A(e), \mc A(e')$ for an
edge $(e, e')$. There are $16$ possible indices in $[4] \times [4]$, and the total count is thus,
\[
	4h |\mc T| 16^{h-1} n^{3h} \leq C^h h! n^{3h}\,,
\]
where we used that $(h/e)^h \leq h!$ and the constants $e$ and $16$ absorbed into $C$.

Lastly, if $h=1$, the count is zero if $c_0 \not\in C_1$, and otherwise equals $\binom{n-1}{3}\leq n^3$. This
proves the pinned estimate in all cases.

\paragraph{Pinned edges. } Here no physical site is fixed at the root. A naive tree exposure has $n^4$ choices for
the root label and at most $n^3$ choices for each subsequent label, giving $n^{3h+1}$ choices
for fixed overlap positions. The symmetric-difference condition saves three powers of $n$,
giving a label-count bound with a factor $n^{3h-2}$. Multiplication by the covariance factor $\sigma_n^{2h}$ then leaves $6^h n^{-2}$.

Fix a nonempty $B \subseteq [h]$ and $k = |B|$. Note that since $\triangle_{e \in B}
I_e = \emptyset$, each physical site appears an even number of times in $\bigcup_{e \in B}
\supp(I_e)$. There are at most $4k$ possible physical sites, so there are at most $2k$ unique
physical sites; this also implies $k > 1$, since otherwise the single label would be proportional to
the identity operator. This implies a bound on the total number of distinct sites in the labeling.
In addition to the initial set of at most $2k$ from $B$, each subsequent labeling must intersect a
prior edge and thus contributes at most $3$ new sites, leading to a total of
\[
	2k + 3(h-k) = 3h - k \leq 3h - 2\,.
\]
For a fixed spanning tree, root, traversal, and overlap positions, the tree exposure has
$3h+1$ initially free site entries. Since at most $3h-2$ distinct physical sites occur, at least
three entries repeat an earlier free entry. Record three such positions and their earlier
occurrences, using at most $(3h+1)^6$ choices. Assign arbitrary sites to the remaining
$3h-2$ entries, using at most $n^{3h-2}$ choices. This records every valid labeling, possibly
more than once, and gives the required upper bound.
The final count is thus,
\[
	|\mc T| (16)^{h-1} (3h + 1)^6 n^{3h - 2} = h^{h-2} (16)^{h-1} (3h+1)^6 n^{3h-2} \leq C^h
	h^{h+4} n^{3h-2}\,.
\]
\end{proof}

\paragraph{Separate parity and connectedness bounds.}
The next lemma separates the parity and connectedness arguments used in \Cref{lem:shared_counts},
allowing arbitrary set sizes, a fixed arbitrary symmetric difference, and fixed root labels. For $k$
quartic labels and an empty symmetric difference, its two estimates recover the power $n^{3h-k}\leq
n^{3h-2}$ for connected extensions, but with a larger prefactor. Jointly encoding the parity and
connectedness in \Cref{lem:shared_counts} is needed for its uniform $C^h h^{h+4}$ prefactor, which
permits summation over all pair counts in the thermal estimates.

\begin{lemma}[Parity and rooted exposure]
\label{lem:parity_exposure}
Let $r_1,\ldots,r_b$ be positive integers, $D=\sum_i r_i$, and $Q\subseteq[n]$
with $s=|Q|$. The number of ordered sets $A_i\in\binom{[n]}{r_i}$ satisfying
$A_1\triangle\cdots\triangle A_b=Q$ is zero unless $s\leq D$ and $D-s$ is even.
Otherwise it is at most
\begin{equation}
  \left(\frac{D+s}{2}\right)^D n^{(D-s)/2}.
  \label{eq:parity_count_shared}
\end{equation}
Consider also a skeleton $\mc D=(\vec r,M,\pi)$ as in \Cref{prop:polymer_counting},
on at most four trace colors, with $|\pi|=h$ and $|M|=j$. Adjoin $a\geq0$ observable
slots indexed by $[a]$, each assigned to a fixed trace color. A labeling assigns a quartic
label to each object in $M\sqcup\pi\sqcup[a]$. The colored support of a perturbation or
observable slot in color $c$ is $\{c\}$ times its label; the support of a pair is the union
of these supports over its two endpoint colors. Connectedness means connectedness of the
graph joining objects whose colored supports intersect.
Consider any partial labeling $\widetilde{\mc A}$ specifying the labels of every perturbation
slot and observable slot, but only $k \leq h$ of the pairs. Assume $k+j+a>0$ whenever $h>0$,
so the family of already labeled roots is nonempty. Then, there are at most
\begin{equation}
	\label{eq:rooted_exposure_shared}
	[64(h+j+a)]^{h-k}n^{3(h-k)}
\end{equation}
completed, connected labelings $\mc A$ which are consistent with $\widetilde{\mc A}$. When $h=k$,
the bound is $1$.
\end{lemma}

\begin{proof}

	First we show \Cref{eq:parity_count_shared}. We encode each possible ordered set as a
	``word'' and count the number of possible words. Start with the word
	\[
		w_0 = A_1 \circ A_2 \circ \dots \circ A_b\,,
	\]
	encoding each $A_i$ as an increasing sequence of its elements. We have $|w_0| = D$.

	Separately, create a list $\mc L$ of the distinct elements $\overline Q := \Paren{\bigcup_{i
	\leq b} A_i} \setminus Q$ and of $Q$ in increasing order; this includes all elements which
	occur in some $A_i$. Elements of $\overline Q$ occur with positive even multiplicity and thus
	$|\overline Q| \leq \tfrac 1 2 (D - s)$. This means that $|\mc L| \leq \tfrac{1}{2}(D-s) + s
	= \tfrac 1 2 (D+ s)$.

	Replace each entry of $w_0$ by its index in $\mc L$. Since $|\mc L| \leq
	\tfrac 1 2 (D+s)$, there are,
	\[
		\Paren{\frac{D+s}{2}}^D
	\]
	possible encoded words. Finally, we sum over possible lists $\mc L$. Since $Q$ is fixed, we
	don't have to sum over entries corresponding to $Q$. There are $n^{(D-s)/2}$ possible
	choices for the positions corresponding to $\overline Q$. Taking the product of these two
	counts yields \Cref{eq:parity_count_shared}.

	For the second assertion, use the originally-labeled objects as roots, which are nonempty whenever $h>0$. Each unknown pair
	chooses a neighboring parent at distance one less from the root set, together with a common
	colored site. Record the parent object, the common color, the positions of that site in the
	parent and child labels, and the other three entries of the child label. There are at most
	$h+j+a$ parent choices, four colors, $4^2$ position choices, and $n^3$ choices for the
	remaining three entries of each unknown pair label. Taking the product over the $h-k$ unknown pairs proves \Cref{eq:rooted_exposure_shared}.

\end{proof}
 \subsection{Partition function moments}
\label{sec:thermal_kp_proof}

For $\vartheta\in\{0,1\}$, let $\vec\bsg_1, \vec\bsg_2$ be jointly Gaussian with standard Gaussian marginals and
$\E_{\vec\bsg_1,\vec\bsg_2}[(\vec\bsg_1)_I (\vec\bsg_2)_J] = \vartheta \mathbf 1_{I=J}$, and put
\[
	M_\vartheta(\beta_1,\beta_2) = \E_{\vec\bsg_1,\vec\bsg_2}[Z(\beta_1,\vec\bsg_1) Z(\beta_2,\vec\bsg_2)]\,.
\]
Thus
\begin{equation}
	\label{eq:part_ratio}
	\frac{\E_{\vec\bsg}|Z(\beta, \vec\bsg)|^2}{|Z_\mathrm{ann}(\beta)|^2} = \frac{M_1(\beta,
	\overline\beta)}{M_0(\beta, \overline\beta)}.
\end{equation}
We used $Z(\overline\beta,\vec g)=\overline{Z(\beta,\vec g)}$, which follows from Hermiticity of
$H(\vec g)$.

\begin{lemma}[Self-averaging of the partition function]
\label{lem:ordinary_replica}
There are absolute constants $\beta_\mathrm{thermal}>0$ and $C<\infty$ such that
for all even $n\geq8$, $Z_\mathrm{ann}$ is zero-free on a neighborhood of
$|\beta| \leq 16\beta_\mathrm{thermal}$ and
\begin{equation}
	\label{eq:replica_bound}
	\sup_{|\beta|\leq 16\beta_\mathrm{thermal}} \Norm{\frac{Z(\beta,
	\vec\bsg)}{Z_\mathrm{ann}(\beta)} - 1}_{L^2} \leq Cn^{-1}\,.
\end{equation}
Moreover, the quotient above is $L^2$-holomorphic on $|\beta| \leq
16\beta_\mathrm{thermal}$.
\end{lemma}
The factor $16$ in the above lemma is to absorb further losses in obtaining the final thermal
estimate \Cref{lem:thermal_estimates} for the threshold $\beta_\mathrm{thermal}$.
\begin{proof}
	We bound the left-hand side by a two-replica quantity. Jensen's inequality implies that the ratio below is at least one.
\begin{align*}
	\Norm{\frac{Z(\beta, \vec\bsg)}{Z_\mathrm{ann}(\beta)} - 1}^2_{L^2} &= \frac{\E_{\vec\bsg}
	|Z(\beta, \vec\bsg)|^2}{|Z_\mathrm{ann}(\beta)|^2} - \frac{\E_{\vec\bsg} Z(\beta, \vec\bsg)}{Z_\mathrm{ann}(\beta)}\\
	&= \frac{\E_{\vec\bsg} |Z(\beta, \vec\bsg)|^2}{|Z_\mathrm{ann}(\beta)|^2} - 1\\
	&\leq \Abs{\frac{\E_{\vec\bsg} |Z(\beta, \vec\bsg)|^2} {|Z_\mathrm{ann}(\beta)|^2} - 1}\\
	\label{eq:desired_final_bound}
	&\leq \Abs{\frac{\E_{\vec\bsg} |Z(\beta, \vec\bsg)|^2}
	{|Z_\mathrm{ann}(\beta)|^2}}\Abs{\log \frac{\E_{\vec\bsg} |Z(\beta, \vec\bsg)|^2}
	{|Z_\mathrm{ann}(\beta)|^2}}\numberthis
\end{align*}
In the last line, we used $e^u-1\leq e^u u$ for $u\geq0$. By
\Cref{eq:part_ratio}, it therefore suffices to bound the logarithm of
$M_1(\beta,\overline\beta)/M_0(\beta,\overline\beta)$.
We follow the exposition in \Cref{sec:wick_placement}.
Specializing \Cref{eq:diagram_activity} to $C=2$, with same-color covariance $1$
and cross-color covariance $\vartheta$ yields the polymer activity
\[
	\mathsf{Act}_\vartheta(\mc P; \beta_1, \beta_2) = \Paren{\prod_{c=1}^2 \frac{(-\beta_c)^{r_c}}{r_c!}}
	\sigma_n^{2|\pi|} \Paren{\prod_{e=\{(c,i), (c',j)\} \in \pi} \vartheta^{\mathbf 1_{c \neq
	c'}}} \prod_c \Tr[\mathsf{Mon}_c(\mc P)]\,.
\]
Call a polymer \emph{mixed} when it contains a cross-color pair. Notice that its activity for $\vartheta=0$ is
zero since the corresponding covariance term is $0$. The activities of polymers that are not mixed agree at $\vartheta=0$ and $\vartheta=1$.

\Cref{lem:component_placement} yields coefficient-wise equality with the hard-core gas
\[
	M_\vartheta(\beta_1, \beta_2) = \Xi_\vartheta(\beta_1, \beta_2) = \sum_{q \geq 0}\frac1{q!}
	\sum_{\substack{\mc P_1, \dots, \mc P_q \in \mathsf P\\ \mc P_i\sim\mc P_j\ (i < j)}}
	\prod_{j=1}^q \mathsf{Act}_\vartheta(\mc P_j; \beta_1, \beta_2)\,.
\]
We establish the KP conditions in \Cref{thm:kp}. Write $b=\max\{|\beta_1|,|\beta_2|\}$
and use the common majorant
\[
  w_b(\mc P)
    :=\sup_{|\zeta_1|,|\zeta_2|\leq b}
         |\mathsf{Act}_1(\mc P;\zeta_1,\zeta_2)|
    =\left(\prod_{c=1}^2\frac{b^{r_c}}{r_c!}\right)
       \sigma_n^{2h}\prod_{c=1}^2|\Tr[\mathsf{Mon}_c(\mc P)]|,
       \qquad h=|\pi(\mc P)|.
\]
It majorizes both $|\mathsf{Act}_0|$ and $|\mathsf{Act}_1|$ throughout this polydisk.
We apply \Cref{prop:polymer_counting} separately to monochromatic and mixed polymers.

Start with monochromatic polymers of color $c\in\{1,2\}$, and fix a physical site $x\in[n]$ and $h
\geq 1$. For the monochromatic count, take $\mc B=\{c\}$, $\mc R=\mathbb Z_{\geq0}^{\{c\}}$,
$b_c=b$, $j=0$, and $t=1$. The single Taylor block belongs to trace $c$, and there are no observable slots.
For $\mc D=(\vec r,\varnothing,\pi)$, let $\lambda=\mc A:\pi\to\mc I_4$, and let
$\mathsf{Lab}_{c,x}(\mc D)$ consist of the labelings whose colored overlap graph is connected, whose
support contains $(c,x)$, and whose trace in color $c$ is nonzero. Take $\mathsf{Lab}(\mc
D)=\mathsf{Lab}_{c,x}(\mc D)$ and $v(\mc D,\mc A)=1$. The proposition's $\mathcal M_{h,0}$ is exactly
\[
  \mathcal M^{\mathrm{mono},c,x}_{h,0}
    :=\sum_{\substack{\mc P\text{ monochromatic in }c\\
                       (c,x)\in\supp(\mc P),\ |\pi(\mc P)|=h}}
          w_b(\mc P)e^{|\supp(\mc P)|}.
\]
The support is colored but uses only color $c$, so take $s_{h,0}=4h$.
The pinned estimate in \Cref{lem:shared_counts} gives
$\mathsf{LabSum}_v(\mc D)\leq L_{h,0}=C^hh!n^{3h}$, and is independent of the skeleton,
$c$, and $x$. The four-factor bound \Cref{eq:shared_counting_bound} is therefore
\[
\begin{aligned}
  \mathcal M^{\mathrm{mono},c,x}_{h,0}
    &\leq
      \underbrace{\frac{b^{2h}}{0!\,2^hh!}}_{\text{Taylor weights and skeletons}}
      \underbrace{\sigma_n^{2h}}_{\text{Wick covariances}}
      \underbrace{C^hh!n^{3h}}_{\text{weighted label sum}}
      \underbrace{e^{4h}}_{\text{support weight}}\\
    &\leq (Cb^2)^h,
\end{aligned}
\]
where the last step uses $\sigma_n^{2h}n^{3h}=6^h$. Hence
\begin{equation}
	\label{eq:monochromatic_bound}
  \sup_{c\in\{1,2\},\,x\in[n]}
    \sum_{\substack{\mc P\text{ not mixed}\\(c,x)\in\supp(\mc P)}}
       w_b(\mc P)e^{|\supp(\mc P)|}
    =\sup_{c,x}\sum_{h\geq1}\mathcal M^{\mathrm{mono},c,x}_{h,0}
    \leq\sum_{h\geq1}(Cb^2)^h.
\end{equation}

For the mixed count, fix $h\geq0$ and take $\mc B=\{1,2\}$, $\mc R=\mathbb Z_{\geq0}^{\{1,2\}}$,
$b_1=b_2=b$, $j=0$, and $t=1$; there are no observable slots. For $\mc D=(\vec r,\varnothing,\pi)$, the
labeling is again
$\lambda=\mc A:\pi\to\mc I_4$. Let $\mathsf{Lab}_{\mathrm{mix}}(\mc D)$ consist of
the labelings for which the colored overlap graph is connected, the matching has at least
one cross-color pair, and both traces are nonzero. Take
$\mathsf{Lab}(\mc D)=\mathsf{Lab}_{\mathrm{mix}}(\mc D)$ and $v(\mc D,\mc A)=1$.
For this application, the proposition's $\mathcal M_{h,0}$ is the global mixed mass
\[
  \mathcal M^{\mathrm{mix}}_{h,0}
    :=\sum_{\substack{\mc P\text{ mixed}\\|\pi(\mc P)|=h}}
          w_b(\mc P)e^{|\supp(\mc P)|}.
\]
Since there are two colors, each Wick pair has colored support of size at most $8$, so $s_{h,0}=8h$.

To refine the label count, we use that nonzero trace requires the symmetric difference of all
physical sites assigned to each trace to be $\emptyset$. There are two sets of edges, \emph{mixed}
edges which join the two traces, and \emph{non-mixed} edges which are internal to a trace. Non-mixed
edges contribute each physical site twice and thus contribute nothing to the symmetric difference. On
the other hand, each mixed edge only contributes each physical site once and thus we must have,
\[
	\mathop{\bigsymdiff}_{e\text{ mixed}}\mc A(e)=\emptyset\,.
\]
This supplies the pinned edges for \Cref{eq:shared_xor_count} in \Cref{lem:shared_counts} and
$\mathsf{LabSum}_v(\mc D)\leq L_{h,0}=C^hh^{h+4}n^{3h-2}$ for $h\geq2$. There are no labelings for
$h=0$ or $h=1$, so take $L_{0,0} = L_{1,0} = 0$; only $h \geq 2$ contributes. Substituting these
parameters into \Cref{eq:shared_counting_bound} gives
\[
\begin{aligned}
  \mathcal M^{\mathrm{mix}}_{h,0}
    &\leq
      \underbrace{\frac{(b+b)^{2h}}{0!\,2^hh!}}_{\text{Taylor weights and skeletons}}
      \underbrace{\sigma_n^{2h}}_{\text{Wick covariances}}
      \underbrace{C^hh^{h+4}n^{3h-2}}_{\text{weighted label sum}}
      \underbrace{e^{8h}}_{\text{support weight}}\\
    &\leq n^{-2}(Cb^2)^h h^4,
\end{aligned}
\]
using $\sigma_n^{2h}n^{3h-2}=6^hn^{-2}$ and $h!\geq(h/e)^h$. This trivially also implies a bound on
the pinned sum,
\begin{align}
  \sum_{\substack{\mc P\text{ mixed}\\y\in\supp(\mc P)}}
      |\mathsf{Act}_1(\mc P)|e^{|\supp(\mc P)|}
    &\leq\sum_{\mc P\text{ mixed}}w_b(\mc P)e^{|\supp(\mc P)|}\notag\\
  \label{eq:global_mixed_mass}
    &=\sum_{h\geq2}\mathcal M^{\mathrm{mix}}_{h,0}
      \leq n^{-2}\sum_{h\geq2}(Cb^2)^hh^4.
\end{align}
The same right-hand side bounds the global absolute activity mass and the pinned $w_b$-mass. Choose
a parameter $b_\mathrm{gas} > 0$ such that \Cref{eq:monochromatic_bound,eq:global_mixed_mass} are at most
$1/4$ for $b\leq2b_\mathrm{gas}$. To account for the margin in a future proof, set
$\beta_\mathrm{thermal}=b_\mathrm{gas}/16$. 
The conditions of \Cref{thm:kp} hold for both the full
gas and the gas with mixed polymers removed, with $a_\mathrm{KP}=1$ and $\delta=1/2$.

Taking $\beta_2=0$ shows that $Z_\mathrm{ann}(\beta_1)=M_1(\beta_1,0)$ is zero-free for
$|\beta_1|<2b_\mathrm{gas}$. Finally, since
\[
	\E_{\vec\bsg} |Z(\beta, \vec\bsg)|^2 = \E_{\vec\bsg} Z(\beta, \vec\bsg) Z(\overline \beta, \vec\bsg) = \Xi_1 = \exp(\mathsf{Mayer}_1)
\]
and
\[
	|Z_\mathrm{ann}(\beta)|^2 = \E_{\vec\bsg_1, \vec\bsg_2} Z(\beta, \vec\bsg_1) Z(\overline
	\beta, \vec\bsg_2) = \Xi_0 = \exp(\mathsf{Mayer}_0)\,,
\]
we have, 
\[
	\log \frac{\E_{\vec\bsg}|Z(\beta, \vec\bsg)|^2}{|Z_\mathrm{ann}(\beta)|^2} =
	\Re\Paren{\mathsf{Mayer}_1 - \mathsf{Mayer}_0} \stackrel{\text{\Cref{cor:kp_comparison}}} \leq \sum_{\mc P\text{ mixed}} |\mathsf{Act}_1(\mc P)|
	e^{|\supp(\mc P)|} \stackrel{\mathrm{Eq.\ \eqref{eq:global_mixed_mass}}} \leq C n^{-2}\,.
\]
Plugging into \Cref{eq:desired_final_bound} proves the norm bound. For fixed $n$, the bound
$|Z(\beta,\vec g)|\leq e^{|\beta|\sigma_n\sqrt m\|\vec g\|_2}$ gives an $L^2$ majorant on every compact temperature disk. Thus $Z$ is $L^2$-holomorphic, and division by the zero-free $Z_\mathrm{ann}$ proves the remaining assertion.
\end{proof}

\subsection{Weighted thermal estimates}
\label{sec:weighted_thermal_proof}

We prove \Cref{lem:thermal_estimates}.

\thermalestimates*

\begin{proof}
Fix $s \geq 1$, and let $b=16\beta_\mathrm{thermal}>0$ and $C$ be the radius and constant
in the $L^2$ bound of \Cref{lem:ordinary_replica}.
Write
\[
	\Delta_\beta(\vec g) = \frac{\Tr[e^{-\beta H(\vec g)}]}{Z_\mathrm{ann}(\beta)} - 1 = \sum_{k\geq0}a_k(\vec g)
	\beta^k.
\]
Since $\beta H(\vec g)$ is linear in both $\vec g$ and $\beta$, the coefficient $a_k$ has Gaussian degree
at most $k$. Moreover, for $|\beta| \leq b$, $Z_\mathrm{ann}$ is zero-free by
\Cref{lem:ordinary_replica} and thus $\Delta_\beta$ is analytic. Thus the $L^q$ Cauchy estimate
in~\Cref{cor:cauchy_lq_bound} with $q=2$ yields $\|a_k\|_{L^2} \leq Cn^{-1} b^{-k}$. For every $q\geq1$, Hermite orthogonality yields
\[
	\|a_k\|_{2,q}^2 = \sum_{d=0}^k q^{2d} \|\Pi_d a_k\|_{L^2}^2 \leq q^{2k} \|a_k\|_{L^2}^2\,,
\]
and in particular, $\|a_k\|_{2,q} \leq q^k \|a_k \|_{L^2}$. Summing the geometric series yields
\begin{equation}
	\sup_{|\beta|\leq b/(2q)}\|\Delta_\beta\|_{2,q} \leq \sum_{k \geq 0} \|a_k\|_{2,q}
	\frac{b^k}{(2q)^k}
    \leq Cn^{-1}\sum_{k\geq0}2^{-k}
    =2Cn^{-1}.
  \label{eq:weighted_upgrade}
\end{equation}
Note that $\Delta_\beta$ is analytic on a larger disk $|\beta| < b/q$ since the series still converges absolutely.

Apply this with $q=2s$, and put $b_\mathrm{gradient}(s)=b/(4s)$,
a radius used only in this proof.
For a finite Hermite expansion $F$, differentiation
$D_I=\partial/\partial g_I$ satisfies
\begin{equation}
  \sum_I\|D_IF\|_{2,s}^2
    =\sum_{d\geq1}d\,s^{2(d-1)}\|\Pi_dF\|_{L^2}^2
    \leq\frac1{4s^2}\|F\|_{2,2s}^2\,,
  \label{eq:weighted_gradient_shared}
\end{equation}
where we used the definition of the weighted Gaussian norm in \Cref{def:weighted_norm} and
$d\leq4^{d-1}$. Since $\Delta_\beta$ is invariant under permutations of the physical sites,
 all $\|D_I\Delta_\beta\|_{2,s}$ are equal.
Consequently, for $|\beta|\leq b_\mathrm{gradient}(s)$,
\begin{equation}
  \|D_I\Delta_\beta\|_{2,s}^2
    =\frac1m\sum_J\|D_J\Delta_\beta\|_{2,s}^2
    \leq\frac{C^2}{s^2mn^2}
    \leq\frac{C_1^2}{s^2n^6}, \label{eq:boundary_bound}
\end{equation}
where $m=\binom n4\geq cn^4$ for $n\geq8$ and $C_1$ is absolute.

Differentiating the Taylor coefficients and using cyclicity
of the trace gives
\begin{equation}
  D_I\Delta_\beta
    =-\beta\sigma_n
       \frac{\Tr[\Gamma_Ie^{-\beta H}]}{Z_\mathrm{ann}(\beta)}
    =\beta^2\sigma_n^2
       U_{\mathrm{annealed},I}(\beta,\cdot).
  \label{eq:insertion_derivative_shared}
\end{equation}
Termwise differentiation is justified by the bounded gradient
above, which also preserves weighted holomorphy.

Since $\Tr[H]=0$, we have
$\Delta_0=\partial_\beta\Delta_0=0$, and hence
$D_I\Delta_\beta$ has a zero of order at least two.
The Schwarz estimate for order-two zeros \cite[Proposition~1.3]{ReichShoikhet1996Generation}, together with the estimate in \Cref{eq:boundary_bound},
 gives
\[
  \|D_I\Delta_\beta\|_{2,s}
    \leq \frac{C_1}{sn^3}
       \left(\frac{|\beta|}{b_\mathrm{gradient}(s)}\right)^2\,.
\]
for $|\beta| < b_\mathrm{gradient}(s)$; we further restrict $|\beta| \leq b_\mathrm{gradient}(s) /
2$ to land inside this open disk.

For $\beta \neq 0$, divide \Cref{eq:insertion_derivative_shared}
by $\beta^2\sigma_n^2$ and use $\sigma_n^2=6n^{-3}$ to obtain
\[
  \max_I\|U_{\mathrm{annealed},I}(\beta,\cdot)\|_{2,s}
    \leq
    \frac{C_1n^{-3}|\beta|^2}
         {s\,b_\mathrm{gradient}(s)^2
          |\beta|^2\sigma_n^2}
    =\frac{C_1}{6s\,b_\mathrm{gradient}(s)^2}.
\]
Additionally, writing $U_{\mathrm{annealed}, I}(\beta, \vec g) = \frac{D_I \Delta_\beta(\vec
g)}{\beta^2 \sigma_n^2}$ shows this function is analytic; since $D_I \Delta_0 = \partial_\beta D_I
\Delta_0 = 0$, the Taylor expansion of $D_I\Delta_\beta$ starts at $\beta^2$, so the singularity at $\beta = 0$ is
removable (with value $g_I$).

We may thus take
\[
  \beta_\mathrm{thermal}(s)=\frac{\beta_\mathrm{thermal}}s=\frac{b}{16s},\qquad
  C_s=\max\left\{
       2C,\frac{C_1}{6s\,b_\mathrm{gradient}(s)^2}
       \right\}.
\]
The first supremum bound follows from
\Cref{eq:weighted_upgrade} and
$\|\Delta_\beta\|_{2,s}\leq\|\Delta_\beta\|_{2,2s}$;
the second is the preceding estimate, including zero by continuity.
Both functions are weighted-holomorphic on the larger disk
$|\beta|<b/(2s)$, and annealed zero-freeness follows from
\Cref{lem:ordinary_replica}.
\end{proof}
 \newcommand{\pertslots}{\mathsf{Slots}_\mathrm{pert}}
\newcommand{\gaussslots}{\mathsf{Slots}_\mathrm{gauss}}

\subsection{From shifts to degree bounds}
\label{sec:degree_proof}

In this section, we establish the necessary estimates to prove \Cref{lem:degree_bound}, restated
here.

\degreebound*

We apply the Hermite features $\vec\Phi_1, \vec\Phi_2$ to the truncated proxies and define
\[
	F^{\beta,K}_{1,I} = \Phi_{1,I}(\vec V^{[K]}_\beta) = V^{[K]}_{\beta, I}\,,\qquad
	F^{\beta,K}_{2,I} = n^{-1/2}\Phi_{2,I}(\vec V^{[K]}_\beta)\,.
\]
To bound $\|\Pi_d F^{\beta, K}_{a,
I}\|_{L^2}^2$, we show that for any scalar polynomial $F:\R^m\to\C$,
\[
	\|\Pi_d F\|_{L^2}^2 = \frac 1 {d!} \sum_{\vec B \in \mc I_4^d} \Abs{\E_{\vec\bsg}
	[(\partial_{B_1}\cdots\partial_{B_d}F)(\vec\bsg)]}^2\,.
\]
In turn, each expected derivative $\E_{\vec\bsg}[(\partial_{B_1}\cdots\partial_{B_d}F)(\vec\bsg)]$ is equal to
\[
	[v_1 \dots v_d] G_{\vec B}\,, \qquad G_{\vec B}(\vec v) = \E_{\vec\bsg} F\Paren{\vec\bsg + \sum_{j
	\in [d]} v_j \vec e_{B_j}}\,.
\]
A bound on $\max_{\vec v\st\|\vec v\|_1\leq1}|G_{\vec B}(\vec v)|$
therefore bounds the required coefficients by \Cref{cor:cauchy_estimate}.
In \Cref{sec:fixed_chaos_proof} we prove,

\begin{restatable}[Perturbed expectation bound]{lemma}{physicalshiftbound}
\label{lem:physical_shift}
For every fixed $A>0$, there is $C_A<\infty$ with the following property, using the absolute
constant $C_\star$ from \Cref{lem:degree_bound}. For every even $n\geq8$, every
integer $1\leq K\leq A\log n$, $I\in\mc I_4$, $a\in\{1,2\}$, and complex $|\beta|\leq
1/(64C_\star)$,
\begin{equation}
  \sup_{\substack{\vec u\in\C^m\\\|\vec u\|_1\leq1}}
    \left|\E_{\vec\bsg}F^{\beta,K}_{a,I}
           (\vec\bsg+\sigma_n^{-1}\vec u)\right|
        \leq C_A\sigma_n^{-1}.
  \label{eq:physical_shift_bound}
\end{equation}
The expectation is the polynomial extension to complex $\vec u$.
The radius is independent of $A$, $K$, $n$, and $I$.
\end{restatable}

Assuming this lemma, we obtain a proof of \Cref{lem:degree_bound}.

\begin{proof}[Proof of \Cref{lem:degree_bound}]
For the cutoff in \Cref{def:expec_proxy}, $K \leq 3\log n$ for $n \geq 8$. Let
$F=F^{\beta,K}_{a,I}$.

For an ordered tuple $\vec B=(B_1,\ldots,B_d)\in\mc I_4^d$, put
\[
  R_{\vec B}(F)=\E_{\vec\bsg}[(\partial_{B_1}\cdots\partial_{B_d}F)(\vec\bsg)],\qquad
  G_{\vec B}(\vec v)=\E_{\vec\bsg}F\left(
       \vec\bsg+\sigma_n^{-1}\sum_{j=1}^dv_j\vec e_{B_j}\right)\,.
\]
so that $[v_1\cdots v_d]G_{\vec B}=\sigma_n^{-d}R_{\vec B}(F)$. Although the labels $B_i$ may
repeat, $|v_j|\leq1/d$ still gives $\|\sum_j v_j \vec e_{B_j}\|_1 \leq 1$ and thus we can use
\Cref{cor:cauchy_estimate} (with $A = 3$) to get a bound on $R_{\vec B}(F)$ in terms of the bound on $G_{\vec B}$
from \Cref{lem:physical_shift}:
\begin{equation}
	\label{eq:ordered_derivative_bound_shared}
	|R_{\vec B}(F)|\leq \sigma_n^d \cdot \sup_{|v_j|\leq1/d\ \forall j}|G_{\vec B}(\vec v)| \cdot
	d^d \leq C_3d^d\sigma_n^{d-1}\,.
\end{equation}

From \Cref{lem:shift_to_hermite}, we have
\begin{equation}
  \|\Pi_dF\|_{L^2}^2
      =\frac1{d!}\sum_{\vec B\in\mc I_4^d}|R_{\vec B}(F)|^2.
  \label{eq:inner_ordered_hermite_shared}
\end{equation}
However, summing \Cref{eq:ordered_derivative_bound_shared} over all $\approx n^{4d}$ choices of
$\vec B$ will overpower the $\sigma^{d-1}_n$ term. Instead, we use the sign covariance of $F$ to
constrain the permissible choices of $\vec B$. For a physical-site sign vector $\vec x \in
\npo^n$, \Cref{cor:q_covariance} yields
\[
	F(\vec x \cdot \vec g) = x_I F(\vec g)\,.
\]
This implies,
\begin{align*}
	R_{\vec B} (F) = \E_{\vec\bsg} [(\partial_{\vec B}F)(\vec\bsg)]= \E_{\vec\bsg} [(\partial_{\vec B}F)(\vec x\cdot\vec\bsg)] = x_{I \triangle S} \E_{\vec\bsg}[(\partial_{\vec B}F)(\vec\bsg)]\,,
\end{align*}
where $S = B_1 \triangle \dots \triangle B_d$ and we have used Gaussian invariance under sign changes in the
second equality and sign covariance in the third. But then if $I \triangle S \neq \emptyset$, $R_{\vec B}(F) = - R_{\vec B}(F) = 0$.
Thus, we only need to sum over $\vec B$ whose symmetric difference is $I$. Applying
\Cref{eq:parity_count_shared} in \Cref{lem:parity_exposure} with total incidence $4d$ and $|I|=4$,
we find that there are at most $(2d+2)^{4d}n^{2d-2}$ such ordered tuples. Combining this with
\Cref{eq:ordered_derivative_bound_shared,eq:inner_ordered_hermite_shared} gives
\begin{align*}
	\|\Pi_d F\|_{L^2}^2 &\leq C_3^2 \frac{d^{2d} (2d+2)^{4d}}{d!} \sigma_n^{2d-2} n^{2d-2}\\
	&= C_3^2 \frac{d^{2d}(2d+2)^{4d}}{d!}\, 6^{d-1} n^{-(d-1)} \leq (c'd)^{5d}n^{-(d-1)}\,.
\end{align*}
The last inequality uses $d!\geq(d/e)^d$ and increases the absolute $c'$. For $a=1$ this is the
first required bound; for $a=2$, undo the factor $n^{-1/2}$ in $F^{\beta,K}_{2,I}$ to obtain the
second.
\end{proof}

\subsection{Reduction to observable diagrams}
\label{sec:observable_reduction}

In this section we reduce the perturbed bound \Cref{lem:physical_shift} to a weighted sum of
diagrams whose components contain observable slots. We bound this sum in
\Cref{sec:fixed_chaos_proof}. The argument starts from the shifted expectation
\[
  \E_{\vec\bsg}F^{\beta,K}_{a,I}
     (\vec\bsg+\sigma_n^{-1}\vec u),\qquad \|\vec u\|_1\leq1,
\]
where we recall that the two normalized features are
\[
	F^{\beta,K}_{1,I} = V^{[K]}_{\beta,I}\,,\qquad
	F^{\beta,K}_{2,I} = n^{-1/2}\Phi_{2,I}(\vec V^{[K]}_\beta)\,.
\]
For this subsection, allow the cutoff $K$ in \Cref{def:expec_proxy} to be
any integer $K\geq1$, and display it as a superscript:
\begin{align}
	V^{[K]}_{\beta, I}(\vec g) &= \Brac{ -\frac{\Tr[\Gamma_Ie^{-\beta H(\vec g)}]}{\beta
	\sigma_n Z_\mathrm{ann}(\beta)} \Paren{2-\frac{Z(\beta,\vec g)}{Z_\mathrm{ann}(\beta)}}
	}_{\leq K,\beta}\\
	\label{eq:auxiliary_proxy}
	&= -\frac 1 {\beta \sigma_n} \Brac{\frac{2 \Tr[\Gamma_I e^{-\beta H(\vec
	g)}]}{Z_\mathrm{ann}(\beta)} - \frac{\Tr[\Gamma_I e^{-\beta H(\vec
	g)}]Z(\beta, \vec g)}{Z_\mathrm{ann}(\beta)^2}}_{\leq K+1,\beta}\,.
\end{align}

Define the shifted Hamiltonian by
\[
 H_{\vec u}(\vec g):=H(\vec g+\sigma_n^{-1}\vec u)
 =H(\vec g)+\sum_{B\in\mc I_4}u_B\Gamma_B.
\]
Expanding a general feature $\Phi_{a,I}(\vec V_\beta^{[K]})$ yields products of traces of the form
\[
	\prod_{b=1}^a \Paren{\Tr[\Gamma_{J_b} e^{-\beta_{b, 0} H_{\vec u}(\vec\bsg)}]
	\prod_{j=1}^{\ell_b-1} \Tr[e^{-\beta_{b, j} H_{\vec u}(\vec\bsg)}]}
\]
At this stage, we can follow the exposition in \Cref{sec:wick_placement}, but with the complication
that we have additional slots which are not paired by Wick--Isserlis; these come from terms
$\Gamma_{J_b}$ as well as from the perturbation $\vec u$. We will call the former \emph{observable}
slots and the latter \emph{perturbation} slots. This adds a complication to the argument (for
instance requiring us to separately consider polymers with and without observable slots), but the
overall framework is similar.

\paragraph{Expanding the shifted feature.}
Recall that,
\[
	\Phi_{1,I}(\vec V_\beta) = V_{\beta, I}\,, \qquad \Phi_{2,I}(\vec V_\beta) = \frac 1 {2 \sqrt {N_2}}
	\sum_{J_1 \triangle J_2 = I} \chi(I; J_1, J_2) V_{\beta, J_1} V_{\beta, J_2}\,.
\]
Define $\omega_{1,I}(I)=1$ and $\omega_{2,I}(J_1,J_2)=\frac{\chi(I;J_1,J_2)}{\sqrt{nN_2}}$ as well
as $\mc J_{1,I} = \{(I)\}$ and $\mc J_{2,I} = \{(J_1, J_2) \st J_1 < J_2\,, J_1 \triangle J_2 = I\}$, where
we fix the lexicographic ordering of the quartic labels. Then, for the normalized features,
\begin{equation}
	\label{eq:general_phi_expr}
	F^{\beta,K}_{a,I}(\vec g) = \sum_{\vec J \in \mc J_{a,I}} \omega_{a, I}(\vec J) \prod_{b =
	1}^a V^{[K]}_{\beta,J_b}(\vec g)\,.
\end{equation}
When expanding $F^{\beta,K}_{a,I}$ as a Taylor expansion, the number of colors is
determined by the choice of $\Phi_1$ or $\Phi_2$, as well as the summand we take from
\Cref{eq:auxiliary_proxy}. To track these choices,
\begin{itemize}
	\item The integer $a \in [2]$ indicates the choice $F^{\beta,K}_{a,I}$ and the number of observable slots.
	\item The vector $\vec\ell\in[2]^a$ tracks for each copy of $V_\beta$, whether we are taking
		the first or second term of \Cref{eq:auxiliary_proxy}.
\end{itemize}
Therefore, for some $a$ and $\vec \ell$, we have the colors:
\[
	\mc C_{\vec \ell} = \{(b,j) \st 1 \leq b \leq a,\ 0 \leq j \leq \ell_b-1\}\,,
	\qquad |\mc C_{\vec\ell}|=\sum_{b=1}^a\ell_b\leq2a\leq4\,.
\]

To handle the Taylor truncation correctly, we need to associate an independent numerator cutoff
of $K+1$ to each $V_{\beta}$ term (i.e. each distinct $b\in[a]$), shared by the colors
$(b,0),\ldots,(b,\ell_b-1)$. To track this carefully, define separate $\beta$-proxies
$\vec\beta=(\beta_c)_{c\in\mc C_{\vec\ell}}$ for each color. We need to retain the cutoff of
each proxy before setting all these arguments equal to $\beta$. 

\begin{definition}[Separate Taylor cutoffs]
\label{def:shift_truncation}
For a finite color set $\mc C$, a finite set $\mc R\subseteq\mathbb Z_{\geq0}^{\mc
C}$, and a formal series $f(\vec\beta)=\sum_{\vec r\geq0}f_{\vec r}\vec\beta^{\,\vec r}$, define
\begin{equation}
	\label{eq:shift_truncation}
	\paren{\mathsf{Trunc}_{\mc R} f}(\vec \beta) := \sum_{\vec r \in \mc R} f_{\vec r}
	\vec \beta^{\,\vec r}\,.
\end{equation}
The coefficients may be scalars or Gaussian polynomials. For the colors
$\mc C_{\vec\ell}$ and an integer $K\geq1$, put
\[
	\mc R_{\vec\ell,K} =\Brac{\vec r\in\mathbb Z_{\geq0}^{\mc C_{\vec\ell}} \st
	\sum_{j = 0}^{\ell_b-1} r_{b, j} \leq K+1 \text{ for every } b \in [a]},\qquad
	\mathsf{Trunc}_{\vec \ell, K} := \mathsf{Trunc}_{\mc R_{\vec \ell, K}}.
\]
We write $(\mathsf{Trunc}_{\vec \ell, K} f)(\beta \vec 1)$ for the resulting polynomial evaluated
with every $\beta_c = \beta$. The bound $K+1$ is the numerator cutoff before division by $\beta$ in
each proxy, as explained after \Cref{eq:auxiliary_proxy}.
\end{definition}

In one variable $\mathsf{Trunc}$ agrees with the Taylor truncation $[\cdot]_{\leq K+1,\beta}$ when
$\mc R=\{0,\ldots,K+1\}$. The set $\mc R_{\vec\ell,K}$ is down-closed. Thus, for formal
series $f,g$ with nonnegative multidegrees,
\[
	\mathsf{Trunc}_{\vec \ell, K}(f g) =\mathsf{Trunc}_{\vec \ell, K}
	\paren{(\mathsf{Trunc}_{\vec \ell, K} f) (\mathsf{Trunc}_{\vec \ell, K} g)}\,.
\]
Indeed, if $\vec r+\vec s\in\mc R_{\vec\ell,K}$, then both $\vec r$ and $\vec s$ belong to the
same set.

\begin{lemma}[Exact expansion of the shift]
\label{lem:shift_expansion}
Let $n \geq 8$ be even, $I \in \mc I_4$, $a \in [2]$, $K \geq 1$ an integer, and $\vec u \in \C^m$.
Put $a_1=2$ and $a_2=-1$. For $\vec\ell \in [2]^a$ and $\vec J \in \mc J_{a,I}$, define
\begin{align}
	\label{eq:shift_trace_numerators}
	N^{\vec u}_{\vec \ell, \vec J}(\vec \beta) &= \E_{\vec\bsg} \prod_{b=1}^a
	\Paren{\Tr[\Gamma_{J_b} e^{-\beta_{b, 0} H_{\vec u}(\vec\bsg)}] \prod_{j=1}^{\ell_b-1}
	\Tr[e^{-\beta_{b, j} H_{\vec u}(\vec\bsg)}]}, \notag\\
	\mc N^{\vec u}_{\vec \ell, \vec J}(\vec \beta) &= \frac{N^{\vec u}_{\vec \ell, \vec
	J}(\vec \beta)} {\prod_{c \in \mc C_{\vec \ell}} Z_\mathrm{ann}(\beta_c)}\,.
\end{align}
For every $\beta \neq 0$, the exact expansion is
\begin{savedequation}{finalformexpectation}
	\label{eq:polarized_bridge_shared}
	\E_{\vec\bsg}F^{\beta,K}_{a,I}(\vec\bsg+\sigma_n^{-1}\vec u) =\frac 1 {(-\beta \sigma_n)^a}
	\sum_{\vec \ell \in [2]^a} \Paren{\prod_{b=1}^a a_{\ell_b}} \sum_{\vec J \in \mc
	J_{a,I}}\omega_{a,I} (\vec J) \paren{\mathsf{Trunc}_{\vec \ell, K} \mc N^{\vec u}_{\vec
\ell, \vec J}}(\beta \vec
	1)\,.
\end{savedequation}
\end{lemma}
\begin{proof}
	Using \Cref{eq:general_phi_expr} we can write
	\begin{alignat}{2}
		\label{eq:shift_feature_product}
		&\E_{\vec\bsg} F^{\beta,K}_{a,I} (\vec\bsg + \sigma_n^{-1} \vec u) &&= \sum_{\vec
		J \in \mc J_{a, I}} \omega_{a, I}(\vec J) \E_{\vec\bsg} \prod_{b = 1}^a
		V^{[K]}_{\beta, J_b}(\vec\bsg + \sigma_n^{-1} \vec u)\,,\\
		\label{eq:shift_single_proxy_expansion}
		&V^{[K]}_{\beta, J}(\vec\bsg + \sigma_n^{-1}\vec u) &&= \frac 1 {-\beta
		\sigma_n}\sum_{\ell = 1}^2 a_\ell \Brac{ \frac{\Tr[\Gamma_J e^{-\beta H_{\vec u}
		(\vec\bsg)}] \paren{\Tr[e^{-\beta H_{\vec u}(\vec\bsg)}]}^{\ell-1}}
		{Z_\mathrm{ann}(\beta)^{\ell}} }_{\leq K+1,\beta}\,.
	\end{alignat}
	If we substitute \Cref{eq:shift_single_proxy_expansion} into \Cref{eq:shift_feature_product}
	and distribute the product, we get a sum over $\vec \ell \in [2]^a$ as desired. Applying the
	rectangular truncation via $\mathsf{Trunc}$ yields exactly \Cref{eq:polarized_bridge_shared}.
\end{proof}

It remains to bound the truncated expressions $\mathsf{Trunc}_{\vec \ell, K} \mc N^{\vec
u}_{\vec \ell, \vec J}$.  We first expand $N^{\vec u}_{\vec\ell,\vec J}$ into Wick diagrams, so that
the KP estimates can be applied to its connected components. Fix $\vec\ell$ and $\vec J$, and
abbreviate $\mc C=\mc C_{\vec\ell}$ and $c_b=(b,0)$. Expanding each exponential in
\Cref{eq:shift_trace_numerators} gives
\[
	N^{\vec u}_{\vec\ell,\vec J}(\vec\beta) = \sum_{\vec r \geq 0}
	\Paren{\prod_{c \in \mc C} \frac{(-\beta_c)^{r_c}}{r_c!}} \E_{\vec\bsg} \prod_{c \in
	\mc C} \Tr\Brac{ \Paren{\prod_{\substack{ b \in [a]\\ c_b = c}} \Gamma_{J_b}}
	H_{\vec u}(\vec\bsg)^{r_c} }\,.
\]
There is still only one exponential per trace, as in
\Cref{lem:component_placement}. The changes are the fixed observable before
some exponentials and the perturbation terms in $H_{\vec u}$; these will yield the observable and
perturbation slots.

First we only consider slots originating from the Hamiltonian (so excluding observable slots). For
the occurrence vector $\vec r=(r_c)_{c\in\mc C}$, let
\[
	\mathsf{Slots}(\vec r) = \bigsqcup_{c \in \mc C}(\{c\} \times [r_c])\,.
\]
At each Hamiltonian slot choose either the Gaussian term from $\sigma_n\sum_B\bsg_B\Gamma_B$ or a
perturbation term from $\sum_B u_B\Gamma_B$. This partitions the slots into perturbation and
Gaussian slots as,
\[
	\mathsf{Slots}(\vec r) = \mathsf{Slots}_\mathrm{pert} \sqcup \mathsf{Slots}_\mathrm{gauss}\,.
\]
Wick--Isserlis is applied only to the Gaussian slots, and thus the domain of our labeling
function $\mc A$ is the union of a perfect matching over the Gaussian slots $\pi \in
\mathrm{PM}(\mathsf{Slots}_\mathrm{gauss})$ as well as the raw perturbation slots
$\mathsf{Slots}_\mathrm{pert}$. A pair edge $e \in \pi$ contributes $\sigma_n^2$, since every
trace uses the same Gaussian vector; thus $\kappa_{c,c'}=1$. A perturbation slot contributes
$u_{\mc A(v)}$ and has no covariance factor.

Finally, a full diagram is $\mc D=(\vec r,\pertslots ,\pi,\mc A)$. Abusing notation, for colored slot
$(c,i)$, we define $\mc A(c,i)$ as,
\[
	\mc A(c,i)= \begin{cases}
		\mc A((c,i)), &(c,i)\in \pertslots\,,\\
		\mc A(e_{\mc D}(c,i)), &(c,i) \in \gaussslots\,,
	\end{cases}
\]
where $e_{\mc D}(c,i)$ is the incident matching edge. Its word in color $c$ is
\[
	\mathsf{Mon}_c(\mc D) =\Paren{\prod_{\substack{b\in[a]\\c_b=c}}\Gamma_{J_b}}
	\prod_{i=1}^{r_c} \Gamma_{\mc A(c,i)}\,.
\]
Putting this together, we have the activity,
\begin{equation}
	\label{eq:shift_diagram_weight}
	\mathsf{Act}_{\vec J,\vec u}(\mc D;\vec\beta) =\Paren{\prod_{c\in\mc C}
	\frac{(-\beta_c)^{r_c}}{r_c!}} \sigma_n^{2|\pi|} \Paren{\prod_{v\in\pertslots}u_{\mc A(v)}}
	\prod_{c\in\mc C}\Tr[\mathsf{Mon}_c(\mc D)]\,.
\end{equation}
Wick terms are mapped bijectively to exactly one such diagram.

Next, as in \Cref{lem:component_placement}, we decompose each diagram into connected polymers. To
carry this out, include the observable slots as distinguished objects of the diagram graph. Thus, its vertex
set is given by the disjoint union $\pertslots\sqcup\pi\sqcup[a]$. We define supports for each of these
vertices with the same typing so that edges of the graph can still be defined via overlapping
supports:
\[
	\supp(v) = \begin{cases}
		\{c\} \times \mc A(v) & (v=(c,i)\in \pertslots)\,,\\
		\bigcup_{(c,i)\in e} \Paren{\{c\}\times\mc A(e)} & (v = e \in \pi)\,,\\
		\{c_b\} \times J_b & (v = b \in [a])\,.\\
	\end{cases}
\]

As usual, this decomposition will distribute slots among the polymers. A polymer $\mc P$ is defined by
$(M, \pi, S, \mc A)$, where $M \subseteq \pertslots$ and $S \subseteq [a]$; in particular,
perturbation slots and observable slots are spread among the polymers which make up a diagram. If we
want to refer to the observable slots of a polymer $\mc P$, we use the notation $S(\mc P)$. The
corresponding polymer monomial is,
\[
	\mathsf{Mon}_c(\mc P) =\Paren{\prod_{\substack{b\in S \\c_b=c}}\Gamma_{J_b}}
	\prod_{i=1}^{r_c(\mc P)}\Gamma_{\mc A(c,i)}\,.
\]
Compatibility means disjoint colored supports, including any observable slots.
These are the only changes to the diagrams in
\Cref{lem:component_placement}.

\begin{lemma}[Component placement with perturbation and observable slots]
\label{lem:observable_component_placement}
Use \Cref{eq:shift_diagram_weight} for the activity of an arbitrary
full diagram. The corresponding activity of a connected diagram is
\begin{equation}
  \mathsf{Act}_{\vec J,\vec u}(\mc P;\vec\beta)
    =
    \left(
      \prod_{c\in\mc C}
        \frac{(-\beta_c)^{r_c(\mc P)}}{r_c(\mc P)!}
    \right)
    \sigma_n^{2|\pi|}
    \left(\prod_{v\in M}u_{\mc A(v)}\right)
    \prod_{c\in\mc C}\Tr[\mathsf{Mon}_c(\mc P)].
  \label{eq:observable_activity_shared}
\end{equation}
Here $M,\pi,S,\mc A$ are the data of $\mc P$. A finite family of polymers $\mc F$ is admissible
when its members are pairwise compatible and their observable-slot sets partition $[a]$:
\[
	\bigsqcup_{\mc P \in \mc F} S(\mc P) =[a]\,.
\]
Then the numerator in \Cref{lem:shift_expansion} satisfies, coefficientwise,
\begin{align}
	\label{eq:shift_wick_expansion}
	N^{\vec u}_{\vec \ell, \vec J}(\vec\beta) &=\sum_{\text{diagrams }\mc D} \mathsf{Act}_{\vec
	J,\vec u}(\mc D; \vec \beta)\,,\\
	\label{eq:observable_placement_identity}
	&=\sum_{\mc F\text{ admissible}} \prod_{\mc P\in\mc F} \mathsf{Act}_{\vec J,\vec
	u}(\mc P;\vec\beta)\,.
\end{align}
\end{lemma}
\begin{proof}
The Taylor--Wick expansion above has one term for each full diagram, with activity
\Cref{eq:shift_diagram_weight}. This gives \Cref{eq:shift_wick_expansion}. Apply the compression and
placement argument of \Cref{lem:component_placement} to the object graph above. The maps
$\mathsf{ord}_{j,c}$ list the Taylor slots of component $j$ in inherited increasing order, and the
strictly increasing maps $\mathsf{Place}_{c,j}$ place them back in their global color slots.
Perturbation slots remain perturbation slots, and all matching incidences and quartic labels are retained.
Compatibility prevents new overlap between different components, so the same reconstruction gives
the inverse maps. The additional observable-slot condition is captured by the restriction to admissible
families.

For the weights, distinct component words have disjoint physical supports
at each color, including their observable-slot labels. The normalized-trace
factorization in \Cref{lem:component_placement} therefore still applies.
Every pair belongs to one component, so the covariance factors split as
before. The new perturbation factor splits as well:
\[
  \prod_{v\in M}u_{\mc A(v)}
    =\prod_j\prod_{v\in M_j}u_{\mc A(v)}.
\]
Observable slots contribute to the words but not to the Taylor counts. For
$r_{c,j}=r_c(\mc P_j)$ and $r_c=\sum_jr_{c,j}$, the same placement count
therefore converts the full-diagram factorials into the component factorials:
\[
  \left(\prod_c\frac1{r_c!}\right)
  \left(\prod_c\frac{r_c!}{\prod_jr_{c,j}!}\right)
    =\prod_j\prod_c\frac1{r_{c,j}!}.
\]
Consequently, for every admissible family $\mc F$,
\begin{equation}
  \sum_{\substack{\text{diagrams }\mc D\\
                   \text{compressed components form }\mc F}}
      \mathsf{Act}_{\vec J,\vec u}(\mc D;\vec\beta)
    =\prod_{\mc P\in\mc F}
       \mathsf{Act}_{\vec J,\vec u}(\mc P;\vec\beta).
  \label{eq:observable_placement_weights}
\end{equation}
Summing \Cref{eq:shift_wick_expansion} over these diagram classes gives
\Cref{eq:observable_placement_identity}.
\end{proof}

\paragraph{Separating the observable-carrying and observable-free components.}
The point of \Cref{eq:observable_placement_identity} is that its admissible families can now be split
according to which components contain the observables. Call a component observable-free when $S(\mc
P)=\emptyset$, and observable-carrying otherwise. Let $\Xi=\Xi(\vec\beta)$ be the hard-core gas of
observable-free components, with activity $\mathsf{Act}_{\vec u}$, with observable-slot label vector $\vec J$ removed
as a subscript, to emphasize that it is independent of $\vec J$. As before, zero-trace
components are omitted from the gas.

Let $\mc S$ be a compatible family of observable-carrying components whose observable-slot sets partition
$[a]$. We call this an \emph{observable-slot configuration}. Write $\supp(\mc S)$ for the union of its colored supports, and set
\[
	\mathsf{Act}_{\vec J, \vec u}(\mc S; \vec\beta) =\prod_{\mc P \in \mc S}
	\mathsf{Act}_{\vec J, \vec u}(\mc P; \vec\beta)\,.
\]
Every admissible family decomposes uniquely into an observable-slot configuration $\mc S$ and a
compatible observable-free family disjoint from $\supp(\mc S)$. For fixed $\mc S$, we can sum
over these compatible observable-free families via the restricted gas $\Xi[\supp(\mc S)]$; recall this
gas only sums polymers disjoint from $\supp(\mc S)$. Thus
\[
	N^{\vec u}_{\vec\ell,\vec J} = \sum_{\mc S} \mathsf{Act}_{\vec J,\vec u}(\mc S;\vec\beta)\,
	\Xi[\supp(\mc S)]\,.
\]
Since $\mc N^{\vec u}_{\vec \ell, \vec J}$ is simply the normalized form of $N^{\vec u}_{\vec \ell,
\vec J}$, we get that 
\begin{savedequation}{numeratorform}
	\mc N^{\vec u}_{\vec \ell, \vec J} = \frac{\Xi}{\prod_cZ_\mathrm{ann}(\beta_c)}
	\sum_{\mc S} \mathsf{Act}_{\vec J,\vec u}(\mc S;\vec\beta)\,\frac{\Xi[\supp(\mc S)]}{\Xi}\,.
	\label{eq:observable_extraction_shared}
\end{savedequation}
Note we have also divided by $\Xi$ in the sum, and multiplied by $\Xi$ outside.
Before verifying the Koteck\'y--Preiss conditions, this equality is coefficientwise in the formal
power series in $\vec \beta = (\beta_c)_{c \in \mc C}$.

For $a=2$, an observable-slot configuration either has one component carrying both
observable slots or two components carrying observable-slot sets $\{1\}$ and $\{2\}$.
These identities distinguish the two components, so no additional factor
$2!$ occurs. The coefficient $\omega_{a,I}(\vec J)$ stays outside this
decomposition and appears once for the complete observable-slot configuration.

We have reduced the original shifted expectation to the projected expression
\Cref{eq:polarized_bridge_shared}, with its normalized numerator written as
\Cref{eq:observable_extraction_shared}. We will bound the two observable-free ratios by KP and the remaining
observable-slot configurations by their label count. Finally, we use \Cref{lem:cauchy_cutoff} to
apply the ratio bounds to their Taylor truncations. We record the cutoff estimate next.

\paragraph{Passing a bound through the cutoff.}
Next, we show that a bound on $|f(\vec \beta)|$ for all $\vec \beta$ with $|\beta_c| < \rho$ implies
a bound on its Taylor truncation.
\begin{lemma}[Cauchy estimate for a finite Taylor cutoff]
\label{lem:cauchy_cutoff}
Let $\mc C$ be a finite color set, $\rho>0$, and let $f$ be a scalar
holomorphic function on $|\beta_c|<\rho$, $c\in\mc C$, with
$|f(\vec\beta)|\leq M$. For every finite set
$\mc R\subseteq\mathbb Z_{\geq0}^{\mc C}$, every
$\vec\nu\in\mathbb Z_{\geq0}^{\mc C}$, and $|\beta|\leq\rho/2$,
\begin{equation}
  \left|\bigl(\mathsf{Trunc}_{\mc R}
                     (\vec\beta^{\,\vec\nu}f)\bigr)(\beta\vec1)\right|
       \leq 2^{|\mc C|}M|\beta|^{|\vec\nu|_1}.
  \label{eq:cauchy_cutoff_shared}
\end{equation}
Here $\mathsf{Trunc}_{\mc R}$ is the operator in
\Cref{def:shift_truncation}, and $\beta\vec1$ sets every $\beta_c=\beta$.
The constant is independent of both the cutoff set and the monomial degree;
we use $0^0=1$ at $\beta=0$.
\end{lemma}
\begin{proof}
Write $f(\vec\beta)=\sum_{\vec r\geq0}f_{\vec r}\vec\beta^{\,\vec r}$ and
$|\vec r|_1=\sum_cr_c$.
The coefficient bound in \Cref{cor:cauchy_estimate}, iterated over the variables,
gives $|f_{\vec r}|\leq M\rho^{-|\vec r|_1}$. Therefore
\begin{align*}
  \left|\bigl(\mathsf{Trunc}_{\mc R}
                     (\vec\beta^{\,\vec\nu}f)\bigr)(\beta\vec1)\right|
    &\leq \sum_{\substack{\vec r\geq0\\\vec\nu+\vec r\in\mc R}}
                  |f_{\vec r}|\,|\beta|^{|\vec\nu|_1+|\vec r|_1}\\
    &\leq M|\beta|^{|\vec\nu|_1}
              \sum_{\vec r\geq0}(|\beta|/\rho)^{|\vec r|_1}\\
    &=\frac{M|\beta|^{|\vec\nu|_1}}
             {(1-|\beta|/\rho)^{|\mc C|}}
      \leq2^{|\mc C|}M|\beta|^{|\vec\nu|_1}.
\end{align*}
The monomial shifts the coefficient set and contributes its value at the
physical temperature; no contour-radius factor exponential in its degree
is introduced.
\end{proof}

We now state the reduction used in the next section. Its radius is independent of the Taylor
cutoff; the restriction $K\leq A\log n$ will only enter when we count the observable diagrams.

\begin{lemma}[Reduction to observable diagrams]
\label{lem:observable_diagram_reduction}
There are absolute constants $b_\mathrm{background},C_0>0$, with
$1/(64C_\star)\leq b_\mathrm{background}/2$, with the following property. Let $n\geq8$ be even,
$K\geq1$ an integer, $I\in\mc I_4$, $a\in\{1,2\}$, and $\vec u\in\C^m$ with
$\|\vec u\|_1\leq1$.

For $\vec\ell\in[2]^a$, $\vec J\in\mc J_{a,I}$, and $h,j\geq0$, let
$\mathsf{Diag}^{\mathrm{obs}}_{h,j}(\vec J)$ consist of full diagrams on
$\mc C_{\vec\ell}$ with $h$ Wick pairs and $j$ perturbation slots, with
$\vec r\in\mc R_{\vec\ell,K}$, in which every connected component contains an observable slot
and every trace is nonzero. The colors $\vec\ell$ and the cutoff $K$ are understood in this
notation. For $b\geq0$, define the weighted diagram sum
\begin{equation}
\label{eq:observable_diagram_mass}
\begin{aligned}
 W_{a,I,\vec\ell,K}(b,\vec u)
  :=\sum_{h,j\geq0}\sum_{\vec J\in\mc J_{a,I}} |\omega_{a,I}(\vec J)|
    \sum_{\mc D\in\mathsf{Diag}^{\mathrm{obs}}_{h,j}(\vec J)}
    |\mathsf{Act}_{\vec J,\vec u}(\mc D;b\vec1)|e^{|\supp(\mc D)|}.
\end{aligned}
\end{equation}
Here $\supp(\mc D)$ is the colored support and the activity is given by
\Cref{eq:shift_diagram_weight}. For complex $0<|\beta|\leq b_\mathrm{background}/2$,
\begin{equation}
\label{eq:observable_diagram_reduction}
\begin{aligned}
 &\left|\E_{\vec\bsg}F^{\beta,K}_{a,I}
       (\vec\bsg+\sigma_n^{-1}\vec u)\right|\\
 &\qquad\leq \frac{16e^{C_0}}{(|\beta|\sigma_n)^a}
   \sum_{\vec\ell\in[2]^a}\left(\prod_{b=1}^a|a_{\ell_b}|\right)
   W_{a,I,\vec\ell,K}(|\beta|,\vec u),
\end{aligned}
\end{equation}
where $a_1=2$ and $a_2=-1$.
\end{lemma}

\begin{proof}
Fix $n,K,I,a,\vec u$ as in the statement. From \Cref{lem:shift_expansion}, we have the final form,

\finalformexpectation*
and the exact form for the normalized numerator,

\numeratorform*
Here $\mc S$ corresponds to configurations which include an observable slot in each polymer and
$\Xi[\supp(\mc S)]$ is the gas composed of polymers \emph{disjoint} from the observable-slot
configuration. We first prove
\begin{equation}
\label{eq:observable_free_want_to_bound}
 \Abs{\frac{\Xi}{\prod_c Z_\mathrm{ann}(\beta_c)}\cdot
       \frac{\Xi[\supp(\mc S)]}{\Xi}}
 \leq e^{C_0+|\supp(\mc S)|/2}
\end{equation}
for all observable-slot configurations $\mc S$ and $|\beta_c|<b_\mathrm{background}$. We call this
the \emph{observable-free background}. We then pass this bound through the Taylor cutoff.

\paragraph{Observable-free background.}
We establish two bounds on absolute activities: first we establish a pinned
bound to verify the KP condition; second, we establish an \emph{unpinned} bound to bound the
partition function ratio in \Cref{eq:annealed_ratio_bound}. Fix a common absolute $C\geq1$ large
enough for all three estimates
\Cref{eq:background_site_shared,eq:global_marked_bound,eq:global_active_bound} below, and set
\begin{equation}
	\label{eq:background_radius_choice}
	b_\mathrm{background} := \min\BRAC{\frac{\log(9/8)}{C},\frac1{4\sqrt C}}>0\,.
\end{equation}

First, we verify the KP condition. Fix $\vec\ell$ and $\vec u$ with $\|\vec u\|_1\leq1$, and
abbreviate $\mc C=\mc C_{\vec\ell}$, so $|\mc C|\leq4$. We verify the KP condition,
\[
	\sum_{\substack{\mc P:\,x \in \supp(\mc P)\\\mc P\text{ observable-free}}} w_b(\mc P)
	e^{|\supp(\mc P)|} \leq \delta
\]
We first work at fixed numbers $h$ of Wick pairs and $j$ of perturbation slots, then sum over
$h,j$.

From \Cref{eq:observable_activity_shared} we can define the scalar majorant,
\begin{align}
	w_b(\mc P) &:=\sup_{\substack{\vec\beta \in \C^{\mc C}\\ |\beta_c|\leq b\ \forall c}}
	|\mathsf{Act}_{\vec u}(\mc P;\vec\beta)|\\
	\label{eq:observable_free_majorant}
	&=\left(\prod_{c\in\mc C}\frac{b^{r_c}}{r_c!}\right)
	\sigma_n^{2h} \left(\prod_{v\in M}|u_{\mc A(v)}|\right)
	\prod_{c\in\mc C}|\Tr[\mathsf{Mon}_c(\mc P)]|\,.
\end{align}
Fix $h,j\geq0$ with $h+j>0$ and a colored site $x\in\mc C\times[n]$. We apply
\Cref{prop:polymer_counting} to the mass
\[
  \mathcal M^{\mathrm{bg}}_{h,j}(x)
    :=\sum_{\substack{\mc P\text{ observable-free},\ x\in\supp(\mc P)\\
                       |\pi(\mc P)|=h,\ |M(\mc P)|=j}}
        w_b(\mc P)e^{|\supp(\mc P)|}.
\]
with parameters
\[
	\mc B =\mc C, \qquad \mc R = \mathbb Z_{\geq0}^{\mc C}, \qquad b_c = b\ (c\in\mc C),\qquad t
	= 1, \qquad s_{h,j}=8h+4j, \qquad  v(\mc D,\mc A)=\prod_{m\in M}|u_{\mc A(m)}|\,.
\]
It remains to specify a uniform $L_{h,j}$. Suppose first that $j\geq1$ and the labelings of each
perturbation slot $v_i$ are fixed, with $\mc A(v_i) = \ell_i \in \mc I_4$. Then, we can apply the
second bound in \Cref{lem:parity_exposure}, taking the perturbation slots as the roots and no
observable slots or Wick pairs fixed. This count does not use the pinned site $x$. This gives at most
\[
	[64(h+j)]^h n^{3h}
\]
possible connected extensions. We now crucially use the norm condition on $\vec u$. When summing up
\Cref{eq:observable_free_majorant}, each assignment to the perturbation slots $(\ell_1, \dots, \ell_j)$
contributes $\prod_{i =1}^j |u_{\ell_i}|$. But then,
\begin{equation}
	\label{eq:perturbation_slot_sum}
	\sum_{\ell_1, \dots, \ell_j \in \mc I_4} \prod_{i=1}^j |u_{\ell_i}| =\Paren{\sum_{B \in \mc
	I_4} |u_B|}^j =\|\vec u\|_1^j\,.
\end{equation}
This means we pay only a factor $\|\vec u \|_1^j$, rather than $n^{4j}$ when summing over the
labelings $(\ell_1, \dots, \ell_j)$. After increasing the constant $C$, this gives the bound
\begin{equation}
	\label{eq:pert_slot_observable_free}
	\mathsf{LabSum}_x(\mc D; \vec u) \leq [C(h+j+2)]^h n^{3h} \|\vec u\|_1^j\,,
\end{equation}
when $j\geq1$.

On the other hand when $j = 0$ and $h \geq 1$, the pinned count in \Cref{lem:shared_counts} bounds
the label sum by $C^h h! n^{3h} \leq [C(h+2)]^h n^{3h}$, recovering the same bound as above. There
are no polymers when $h=j=0$. In either case, we get the bound
\[
	L_{h,j}=[C(h+j+2)]^h n^{3h}\|\vec u\|_1^j\,.
\]
With these choices, \Cref{eq:shared_counting_bound} gives the bound
\[
\begin{aligned}
  \mathcal M^{\mathrm{bg}}_{h,j}(x) &\leq \frac{(|\mc C|b)^{2h+j}}{j!\,2^hh!}\,
           \sigma_n^{2h}\,
           \paren{[C(h+j+2)]^h n^{3h}\|\vec u\|_1^j}\, e^{8h+4j} \leq (Cb^2)^h\frac{(Cb)^j}{j!}\,.
\end{aligned}
\]
The last inequality uses $|\mc C|\leq4$, $\|\vec u\|_1\leq1$, $\sigma_n^{2h}n^{3h}=6^h$, and
$(h+j+2)^h/h!\leq e^{h+j+2}$, with the common $C$ chosen above. Summing $\mathcal M^{\rm bg}_{h,j}(x)$
over $h + j > 0$ and taking the supremum over $x$ gives, for $Cb^2 < 1$,
\begin{equation}
	\label{eq:background_site_shared}
	\sup_x\sum_{\mc P \st x\in\supp(\mc P)} w_b(\mc P) e^{|\supp(\mc P)|} \leq
	\sum_{\substack{h, j \geq 0\\ h + j > 0}} (Cb^2)^h \frac{(C b)^j}{j!} = \frac{e^{C b}}{1-C
	b^2} - 1\,.
\end{equation}
For $0\leq b\leq b_\mathrm{background}$, this gives $e^{Cb}\leq9/8$ and $Cb^2\leq1/16$, hence
$e^{Cb}/(1-Cb^2)-1\leq1/5<1/2$. It therefore verifies the KP condition \Cref{eq:shared_kp} with
$\delta=1/2$ and $\Xi = \exp(\mathsf{Mayer})$. The majorant is uniform on the polydisk, so
the gas and its restrictions converge normally there and define holomorphic functions.

Note that the KP condition still holds after restricting any subset of activities to be $0$. We use
this to bound the ratio $|\Xi[\supp(\mc S)]/\Xi|$. The corresponding restricted activity is
\[
	\mathsf{Act}_{[S]}(\mc P; \vec \beta) :=\mathbf 1_{\{\supp(\mc P) \cap S=\emptyset\}}
	\mathsf{Act}_{\vec u}(\mc P; \vec \beta)\,,
\]
with gas $\Xi[S]$. The same inequality from
\Cref{eq:background_site_shared} plus Koteck\'y--Preiss gives $\Xi[S]=\exp(\mathsf{Mayer}[S])$. Apply
the comparison corollary \Cref{cor:kp_comparison} to $\mathsf{Act}_{[S]}$ and $\mathsf{Act}_{\vec
u}$ with the same majorant $w_b$. These activities differ exactly on polymers meeting $S$, so
\begin{align*}
	|\mathsf{Mayer}[S](\vec \beta) - \mathsf{Mayer}(\vec \beta)| &\leq \sum_{\substack{\mc P \in
	\mathsf P\\
	\supp(\mc P)\cap S \neq \emptyset}} |\mathsf{Act}_{\vec u}(\mc P; \vec \beta)| e^{|\supp(\mc
	P)|}\\
	&\leq \sum_{x \in S} \sum_{\mc P \st x \in \supp(\mc P)} w_b(\mc P) e^{|\supp(\mc P)|} \leq
	\frac{|S|}{2}\,.
\end{align*}
The last step uses the site bound \Cref{eq:background_site_shared}. We conclude that for all
$|\beta_c| \leq b_{\rm background}$ that
\[
	\Abs{\frac{\Xi[S](\vec\beta)}{\Xi(\vec\beta)}}
	=\exp\Paren{\Re\Paren{\mathsf{Mayer}[S](\vec\beta)-\mathsf{Mayer}(\vec\beta)}}
	\leq e^{|S|/2} \leq e^{|S|}\,.
\]

It remains to bound $|\Xi/\prod_c Z_{\rm ann}(\beta_c)|$. We again restrict the set of activities,
this time keeping only observable-free polymers that are both \emph{monochromatic} and
\emph{contain no perturbation slots}. Define the corresponding restricted activity as
$\mathsf{Act}_{\mathrm{mono}}$ (which agrees with $\mathsf{Act}_{\vec u}$ on this restriction). This
is useful as the corresponding gas in each color is exactly the annealed partition function
evaluated at $\beta_c$:
\[
	\Xi_{\mathrm{mono},c}(\beta_c) =\E_{\vec\bsg}\Tr[e^{-\beta_cH(\vec\bsg)}] =
	Z_\mathrm{ann}(\beta_c)\,.
\]
This quantity is non-zero as a consequence of Koteck\'y--Preiss. Since polymers in different
colors have disjoint colored supports, the comparison gas factors as
$\Xi_{\mathrm{mono}}=\prod_c\Xi_{\mathrm{mono},c}=\prod_c Z_\mathrm{ann}(\beta_c)$.
Call the remaining observable-free polymers \emph{active}; they either contain a perturbation
slot or span multiple colors. Thus,
\begin{align*}
	\Abs{\frac{\Xi(\vec\beta)} {\prod_c Z_\mathrm{ann}(\beta_c)}}
	&= \exp\Paren{\Re\Paren{\mathsf{Mayer}(\vec\beta)-\mathsf{Mayer}_{\mathrm{mono}}(\vec\beta)}}\\
	&\leq \exp\Paren{\sum_{\mc P\in\mathsf P}
	|\mathsf{Act}_{\vec u}(\mc P;\vec\beta)-\mathsf{Act}_{\mathrm{mono}}(\mc P;\vec\beta)|
	e^{|\supp(\mc P)|}}\\
	&= \exp\Paren{\sum_{\mc P\text{ active}}|\mathsf{Act}_{\vec u}(\mc P;\vec\beta)|e^{|\supp(\mc P)|}}\\
	\label{eq:annealed_ratio_bound}
	&\leq \exp\Paren{\sum_{\mc P\text{ active}}w_b(\mc P)e^{|\supp(\mc P)|}}\numberthis\,.
\end{align*}
To bound the exponent, we need an unpinned bound on the weighted activities. We separate the cases according to whether the
polymer contains a perturbation slot.

First consider polymers with at least one perturbation slot, so $j\geq1$, and fix $h\geq0$. The corresponding
absolute mass is,
\[
	\mathcal M^{\mathrm{mark}}_{h,j} := \sum_{\substack{\mc P\text{ observable-free}\\ |\pi(\mc P)|
	= h,\ |M(\mc P)| = j}} w_b(\mc P) e^{|\supp(\mc P)|}\,,
\]
which corresponds to taking \Cref{prop:polymer_counting} with 
\[
	\mc B=\mc C,\qquad \mc R = \mathbb Z_{\geq0}^{\mc C},\qquad b_c=b\ (c\in\mc C),\qquad t =
	1,\qquad s_{h,j} = 8h+4j\,.
\]
As before $v(\mc D,\mc A)=\prod_{m\in M}|u_{\mc A(m)}|$. Since the argument culminating in the count
\Cref{eq:pert_slot_observable_free} did not depend on a pinned site, we get the same bound $L_{h,j} =
[C(h+j+2)]^h n^{3h} \|\vec u\|_1^j$. Following the same derivation as
$\mathcal M^{\mathrm{bg}}_{h,j}(x)$, \Cref{eq:shared_counting_bound} yields
\[
\begin{aligned}
  \mathcal M^{\mathrm{mark}}_{h,j}
    &\leq \frac{(|\mc C|b)^{2h+j}}{j!\,2^hh!}\,
           \sigma_n^{2h}\,
           \bigl([C(h+j+2)]^h n^{3h}\|\vec u\|_1^j\bigr)\,
           e^{8h+4j}\\
    &\leq (Cb^2)^h\frac{(Cb)^j}{j!},
\end{aligned}
\]
Summing over $h \geq 0$ and $j \geq 1$ yields
\begin{equation}
  \label{eq:global_marked_bound}
  \sum_{\mc P:\,|M(\mc P)|\geq1} w_b(\mc P)e^{|\supp(\mc P)|}
    =\sum_{h\geq0,\,j\geq1}\mathcal M^{\mathrm{mark}}_{h,j}
    \leq \frac{e^{Cb}-1}{1-Cb^2}.
\end{equation}
Finally, consider the unmarked, multicolor contribution. Fix $h \geq 1$ and $j = 0$, and
define
\[
  \mathcal M^{\mathrm{multi}}_{h,0}
    :=\sum_{\substack{\mc P\text{ observable-free and multicolor}\\
                       |\pi(\mc P)|=h,\ |M(\mc P)|=0}}
        w_b(\mc P)e^{|\supp(\mc P)|}.
\]
which corresponds to \Cref{prop:polymer_counting} with
\[
	\mc B = \mc C, \qquad \mc R = \mathbb Z_{\geq 0}^{\mc C},\qquad b_c = b\ (c\in\mc C), \qquad
	t = 1, \qquad s_{h,0} = 8h\,.
\]
Choose a color $c_0$ belonging to a cross-color edge; such an edge must exist as otherwise the
polymer would factorize across the colors. Let $B$ be the nonempty set of Wick pairs with exactly
one endpoint in color $c_0$. The nonzero trace within that color forces the symmetric difference of
these mixed-edge labels to be empty, since each same-color pair label appears twice:
\[
	\mathop{\triangle}_{e\in B}\mc A(e) = \varnothing\,.
\]
Thus, we can apply \Cref{eq:shared_xor_count} from \Cref{lem:shared_counts} and $L_{h, 0} =
4C^hh^{h+4}n^{3h-2}$ for $h \geq 2$ and $0$ otherwise (and $\mc M^{\rm multi} = 0$). The factor of
$4$ accounts for the possible values of $c_0$. \Cref{eq:shared_counting_bound} gives
\begin{equation}
\label{eq:global_active_bound}
\begin{aligned}
  \mathcal M^{\mathrm{multi}}_{h,0}
    &\leq \frac{(|\mc C|b)^{2h}}{2^hh!}\,
           \sigma_n^{2h}\,
           \bigl(4C^hh^{h+4}n^{3h-2}\bigr)\,
           e^{8h}\\
    &\leq n^{-2}(Cb^2)^h h^4.
\end{aligned}
\end{equation}
The last step uses $|\mc C|\leq 4$, $\sigma_n^{2h} n^{3h-2} = 6^h n^{-2}$, and $h^h/h! \leq e^h$;
these factors are absorbed in the common $C$ fixed above.

Combining \Cref{eq:global_marked_bound,eq:global_active_bound} yields that for a new constant $C_0$,
\begin{equation}
\label{eq:background_global_shared}
	M_{\mathrm{active}} :=\sum_{\mc P\text{ active}} w_b(\mc P)e^{|\supp(\mc P)|} \leq
	\frac{e^{Cb}-1}{1-Cb^2} +n^{-2}\sum_{h\geq2}(Cb^2)^hh^4 \leq C_0\,.
\end{equation}
where the last inequality follows from our choice \Cref{eq:background_radius_choice}. Plugging this
bound into \Cref{eq:annealed_ratio_bound} yields, \[
    \sup_{|\beta_c| < b_\mathrm{background}} \Abs{\frac{\Xi(\vec\beta)} {\prod_c
    Z_\mathrm{ann}(\beta_c)}} \leq e^{C_0}\,.
\]

\paragraph{Passing to observable diagrams.}
Put $b=|\beta|$. For a fixed observable-slot configuration $\mc Q$, let $\deg_c(\mc
Q)$ be the total degree of $\beta_c$ across all polymers:
\[
  \deg_c(\mc Q)
    :=\sum_{\mc P\in\mc Q}r_c(\mc P),\qquad
  \vec{\deg}(\mc Q)
    :=(\deg_c(\mc Q))_{c\in\mc C_{\vec\ell}}\,.
\]
From \Cref{eq:observable_activity_shared} we can explicitly write $\mathsf{Act}_{\vec J, \vec u}$ as
$\mathsf{Act}_{\vec J, \vec u}(\mc Q; \vec \beta) = a_{\vec J,\mc Q}\, \vec\beta^{\,\vec{\deg}(\mc
Q)}$ where
\begin{equation}
	\label{eq:absolute_activity_observable}
    a_{\vec J,\mc Q} :=\prod_{\mc P\in\mc Q}\Brac{\Paren{
          \prod_{c\in\mc C_{\vec\ell}}
            \frac{(-1)^{r_c(\mc P)}}{r_c(\mc P)!}
	    }\sigma_n^{2|\pi_{\mc P}|}
        \left(\prod_{v\in M_{\mc P}}
          u_{\mc A_{\mc P}(v)}\right)
	  \prod_{c\in\mc C_{\vec\ell}}\Tr[\mathsf{Mon}_c(\mc P)]
	}\,,
\end{equation}
is independent of $\vec \beta$. Keeping the $\beta$-monomial separate allows
\Cref{lem:cauchy_cutoff} to preserve this physical-temperature weight while bounding the background
factor on its larger polydisk.
For $b\leq b_\mathrm{background}/2$, \Cref{lem:cauchy_cutoff} applied to each configuration and
\Cref{eq:observable_free_want_to_bound} give
\[
 \Abs{(\mathsf{Trunc}_{\vec\ell,K}\mc N^{\vec u}_{\vec\ell,\vec J})(\beta\vec1)}
 \leq16e^{C_0}
 \sum_{\substack{\mc S\text{ observable-slot configuration}\\
                  \vec{\deg}(\mc S)\in\mc R_{\vec\ell,K}}}
 |a_{\vec J,\mc S}|b^{|\vec{\deg}(\mc S)|_1}e^{|\supp(\mc S)|/2}.
\]
Here only configurations of degree in $\mc R_{\vec\ell,K}$ contribute, since this cutoff set is
down-closed and the background has nonnegative Taylor degrees.

The component-placement argument of \Cref{lem:observable_component_placement} also applies to
absolute activities. It preserves the support and Taylor degrees, and its placement count converts
the component factorials into the full-diagram factorials. Thus, writing the observable-slot
configurations as full diagrams gives
\begin{align}
	\label{eq:observable_configuration_bound}
	W_{a,I,\vec\ell,K}(b,\vec u) =\sum_{\vec J\in\mc J_{a,I}}
    \sum_{\substack{\mc S\text{ observable-slot configuration}\\
                      \vec{\deg}(\mc S)\in\mc R_{\vec\ell,K}}}
       |\omega_{a,I}(\vec J)|\,
       |a_{\vec J,\mc S}|\,b^{|\vec{\deg}(\mc S)|_1}e^{|\supp(\mc S)|}\,.
\end{align}
The diagrams in this sum are exactly those in \Cref{eq:observable_diagram_mass}. Replacing
$e^{|\supp(\mc S)|/2}$ by $e^{|\supp(\mc S)|}$ in the preceding bound, summing over $\vec J$, and
using \Cref{eq:polarized_bridge_shared} proves \Cref{eq:observable_diagram_reduction}.

Finally, since $C_\star\geq\max\{1,C\}$, the radius in \Cref{lem:degree_bound} satisfies
\[
  \frac{1}{64C_\star}
    \leq\min\left\{\frac{\log(9/8)}{2C},\frac{1}{8\sqrt C}\right\}
    =\frac{b_\mathrm{background}}{2}.
\]
\end{proof}

\subsection{Counting observable diagrams}
\label{sec:fixed_chaos_proof}

In this section we prove the perturbed bound \Cref{lem:physical_shift}. The only result needed
from \Cref{sec:observable_reduction} is \Cref{lem:observable_diagram_reduction}.

\begin{proof}[Proof of \Cref{lem:physical_shift}]
Fix $A>0$, even $n\geq8$, $1\leq K\leq A\log n$, $I\in\mc I_4$, $a\in\{1,2\}$, and $\vec u$
with $\|\vec u\|_1\leq1$. Put $b=|\beta|$. For each $\vec\ell\in[2]^a$, abbreviate
$\mc C=\mc C_{\vec\ell}$ and
\[
 M_\mathrm{obs}:=W_{a,I,\vec\ell,K}(b,\vec u).
\]
By \Cref{lem:observable_diagram_reduction}, it remains to show that $M_\mathrm{obs}$ is upper
bounded by $C_A\sigma_n^{a-1}b^a$ for an absolute radius $b_\mathrm{observable}>0$.

\paragraph{Observable-containing components.}
For $h,j\geq0$, define the fixed-count observable mass by
\[
	\mathcal M^{\mathrm{obs}}_{h,j} := b^{2h+j}\sigma_n^{2h} \sum_{\vec J\in\mathcal J_{a,I}}
	|\omega_{a,I}(\vec J)| \times \sum_{D \in \mathsf{Diag}^{\mathrm{obs}}_{h, j}(\vec J)}
	\Paren{\prod_{c \in \mathcal C} \frac{1}{r_c(D)!}} \Paren{\prod_{v\in M(D)} |u_{\mathcal
	A_D(v)}|} e^{|\supp(D)|}\,.
\]
Then $M_\mathrm{obs}=\sum_{h,j\geq0}\mathcal M^{\mathrm{obs}}_{h,j}$. The diagrams and cutoff
are those in \Cref{lem:observable_diagram_reduction}. We apply \Cref{prop:polymer_counting} with
\[
	\mc B = \mc C,\qquad \mc R=\mc R_{\vec\ell,K},\qquad b_c = b = |\beta|\ (c\in\mc C),\qquad t
	= 1, \qquad s_{h,j} = 8h + 4j + 4a\,.
\]
Labelings now correspond to a choice of labels $\vec J \in \mc J_{a, I}$ for the observable slots
together with an assignment to the remaining slots. The observable slots contribute the trace in each color, so let
$\mathsf{Lab}_{\vec J}(\mc D)$ denote labelings which yield a non-zero trace in each color
and an observable slot in every connected component. Then,
$\mathsf{Lab}(\mc D) = \{(\vec J,\mc A) \st \vec J \in \mc J_{a, I}, \mc A \in \mathsf{Lab}_{\vec
J}(\mc D)\}$. For each labeling $\lambda \in \mathsf{Lab}(\mc D)$, the corresponding label weight is
$v(\mc D, \lambda) = |\omega_{a,I}(\vec J)|\prod_{m\in M}|u_{\mc A(m)}|$.

For these parameters, \Cref{eq:shared_skeleton_decomposition} gives
\begin{equation}
	\label{eq:observable_skeleton_sum_shared}
	\mathcal M^{\mathrm{obs}}_{h,j} \leq \sigma_n^{2h} e^{8h+4j+4a} \sum_{\mc D \in
	\mathsf{Skel}_{h, j}(\mc R_{\vec \ell, K})} \Paren{\prod_{c \in \mc C}\frac{b^{r_c}}{r_c!}}
	\mathsf{LabSum}_{v = (a,I)}(\mc D;\vec u).
\end{equation}
Here, the original outer sum over $\vec J$ is now inside $\mathsf{LabSum}_{a,I}$. We first bound
$\mathsf{LabSum}$ for one fixed skeleton with $\vec r\in\mc R_{\vec\ell,K}$. Its counts satisfy
\[
	2h+j=\sum_c r_c \leq a(K+1) \leq 2(K+1).
\]
Additionally, fix a labeling of all perturbation slots; we will sum over labelings of these slots
later. Taking the symmetric difference over every color of a connected component cancels every pair
label, leaving only the labels of its observable and perturbation slots. Note this is different from a symmetric
difference over labels in a \emph{single} trace, which yields the empty set \emph{with} the labels of observable
and perturbation slots included. Thus, for each component $\mc P$,
\[
	\mathop{\bigsymdiff}_{v \in M(\mc P)} \mc A(v) = \mathop{\bigsymdiff}_{b
	\in S(\mc P)} J_b \,.
\]
Here $M(\mc P)$ and $S(\mc P)$ are its perturbation slots and observable-slot indices, respectively.
Any admissible labeling falls into one of three cases:
\begin{enumerate}
	\item A single observable slot in a single component,
	\item Two observable slots in distinct connected components and
	\item Two observable slots in a single component.
\end{enumerate}
We count labelings belonging to each of these cases. For one observable slot the observable-slot label is fixed to
$I$, and the rooted bound in \Cref{lem:parity_exposure} gives $[C(h + j + 2)]^h n^{3h}$ possible
pair labelings.

For two observable slots in distinct components, record which observable-carrying component contains each pair and mark,
at a cost of at most $2^{h+j}$. In each component the xor argument above implies its observable-slot label is
determined by the xor of its marks which are fixed. We either discard the record if the resulting
labels are not a permitted pair for $I$ or apply \Cref{lem:parity_exposure} separately in the two
components. Their product is at most $[C(h+j+2)]^hn^{3h}$. The one kernel $\omega_{2, I}$ satisfies
$|\omega_{2,I}|\leq\sigma_n$.

Finally, for a connected component with two observable slots, the trace-color graph is connected. Choose a spanning
tree of that graph and cut one edge on the path between its two colors containing observable slots $c_1$ and $c_2$.
Let $\mc C_1,\mc C_2$ be the two sides of this cut, with $c_i\in\mc C_i$,
and let $B$ be the set of all Wick pairs crossing between these sides, with $k=|B|\geq1$.
For each color, a nonzero trace forces the symmetric difference of all labels in its word to be
empty. Taking the symmetric difference of these identities over every color in $\mc C_i$
cancels each pair internal to that side, including pairs between different colors on the same side.
Each crossing pair contributes once, as does each perturbation slot on that side, and the unique
observable slot on that side has label $J_i$. Therefore,
\[
	\mathop{\bigsymdiff}_{e\in B}\mc A(e)=J_i\triangle D_i\,,\qquad
	D_i:=\mathop{\bigsymdiff}_{\substack{v=(c,j)\in M\\c\in\mc C_i}}\mc A(v),
	\qquad i=1,2\,.
\]
In particular, $D_1,D_2$ are fixed once the perturbation-slot labels and the cut are fixed, and
the two identities also give $D_1\triangle D_2=J_1\triangle J_2=I$.
Since $J_1 \triangle J_2 = I$, we can write $J_1 = P_1 \cup X$ and $J_2 = P_2 \cup X$, where $P_1
\sqcup P_2 = I$ and $X \in \binom{[n] \setminus I}{2}$. With this, we can rewrite the above as,
\[
	X\triangle\mathop{\bigsymdiff}_{e\in B}\mc A(e)=P_1\triangle D_1\,.
\]
But notice that $D_1$ is fixed, and there are only $3$ choices for $P_1$: the ordering
$J_1<J_2$ forces $P_1$ to contain the least element of $I$. Thus, once $P_1$ is fixed,
we can apply \Cref{eq:parity_count_shared} from \Cref{lem:parity_exposure} with set sizes
$2,4,\dots,4$ and bound the choices of $X$ and labels for edges in $B$ by
$3(4k+2)^{4k+2}n^{2k+1}$.

Since $1\leq k\leq h\leq K+1\leq(A+1)\log n$ for $n \geq 8$, there is a finite constant $C_A$ such
that $(4k+2)^{4k+2} \leq C_{A} n^{k-1}$ and the observable slots and edges in $B$ have at most $3C_{A} n^{3k}$
possible labelings. Then, treat each possible such labeling together with the perturbation slots, as roots.
The remaining Wick pair labels have at most $[C(h+j+2)]^{h-k}n^{3(h-k)}$ extensions. Again, the
factor $|\omega_{a, I}(\vec J)|$ contributes at most $\sigma_n$.

To get the final count of valid labelings, we sum over the labelings of perturbation slots. Since
\[
 \sum_{\ell_1,\ldots,\ell_j\in\mc I_4}\prod_{i=1}^j|u_{\ell_i}|
 =\|\vec u\|_1^j\leq1,
\]
this adds at most a factor of $1$. In all cases,
\begin{equation}
	\label{eq:observable_label_count_shared}
	\mathsf{LabSum}_{a,I}(\mc D;\vec u)
	\leq C_A \sigma_n^{a-1} C^{h+j} (h+j+2)^h n^{3h}\,.
\end{equation}
This bound is uniform over the skeletons with $\vec r\in\mc R_{\vec\ell,K}$, so take
\[
  L_{h,j}=C_A\sigma_n^{a-1}C^{h+j}(h+j+2)^h n^{3h}.
\]
All parameters are now specified. Equation \eqref{eq:shared_counting_bound}, applied to
$\mathcal M_{h,j}=\mathcal M^{\mathrm{obs}}_{h,j}$, gives
\[
\begin{aligned}
  \mathcal M^{\mathrm{obs}}_{h,j}
    &\leq \frac{(|\mc C|b)^{2h+j}}{j!\,2^hh!}\,
           \sigma_n^{2h}\,
           \bigl(C_A\sigma_n^{a-1}C^{h+j}(h+j+2)^h n^{3h}\bigr)\,
           e^{8h+4j+4a}\\
    &\leq C_A\sigma_n^{a-1}(Cb^2)^h\frac{(Cb)^j}{j!}.
\end{aligned}
\]
The second line uses $|\mc C|\leq4$, $a\leq2$, $\sigma_n^{2h}n^{3h}=6^h$, and
$(h+j+2)^h/h!\leq e^{h+j+2}$. A nonzero trace in each of the $a$ colors containing observable slots requires at least
one Taylor slot in each such color. Thus $\mathcal M^{\mathrm{obs}}_{h,j}=0$ unless
$a\leq2h+j\leq a(K+1)$. Summing the fixed-count bounds gives
  \begin{equation}
  \label{eq:observable_numerical_majorant}
  \begin{aligned}
    M_\mathrm{obs}
      &\leq
        C_A\sigma_n^{a-1}
        \sum_{\substack{h,j\geq0\\
                         a\leq2h+j\leq a(K+1)}}
          (Cb^2)^h\frac{(Cb)^j}{j!}\\
      &\leq
        C_A\sigma_n^{a-1}
        \sum_{\substack{h,j\geq0\\2h+j\geq a}}
          (Cb^2)^h\frac{(Cb)^j}{j!}.
  \end{aligned}
\end{equation}
Fix the absolute coefficient $C$ in \Cref{eq:observable_numerical_majorant}
and set the proof-local radius $b_\mathrm{observable}:=(2C)^{-1/2}$, so that
$Cb_\mathrm{observable}^2=1/2$. For
$0<b\leq b_\mathrm{observable}$, the last numerical sum is at most
\begin{equation}
	\label{eq:observable_containing_sum_bound}
  b^ab_\mathrm{observable}^{-a}
    \sum_{h,j\geq0}
      (Cb_\mathrm{observable}^2)^h
      \frac{(Cb_\mathrm{observable})^j}{j!}
    =
  b^ab_\mathrm{observable}^{-a}
    \frac{e^{Cb_\mathrm{observable}}}
         {1-Cb_\mathrm{observable}^2}
    \leq C' b^a,
\end{equation}
where $C'$ is absolute because $a\in\{1,2\}$. Increasing $C_A$ proves the final bound $C_A
\sigma_n^{a-1} b^a$.

\paragraph{Concluding the proof}
Combining \Cref{lem:observable_diagram_reduction} with the derivation ending in
\Cref{eq:observable_containing_sum_bound} yields, for $0<|\beta|\leq1/(64C_\star)$,
\begin{equation}
\label{eq:observable_containing_want_to_bound}
\begin{aligned}
 |\E_{\vec\bsg}F^{\beta,K}_{a,I}(\vec\bsg+\sigma_n^{-1}\vec u)|
 &\leq \frac{16e^{C_0}}{(|\beta|\sigma_n)^a}
 \sum_{\vec\ell\in[2]^a}\left(\prod_{b=1}^a|a_{\ell_b}|\right)
 C_A\sigma_n^{a-1}|\beta|^a\\
 &\leq144e^{C_0}C_A\sigma_n^{-1}.
\end{aligned}
\end{equation}
The second inequality uses $\sum_{\vec\ell}\prod_b|a_{\ell_b}|=3^a\leq9$. This yields
\Cref{eq:physical_shift_bound} after increasing $C_A$. Since $C_\star\geq\max\{1,C\}$,
$C/(64C_\star)^2\leq1/4096<1/2$, hence
$1/(64C_\star)\leq b_\mathrm{observable}$. The bound at $\beta=0$ follows by polynomial
continuity.
\end{proof}

\section{Deferred proofs}
\label{sec:deferred_proofs}
We prove \Cref{lem:hessian-bkm}.

\hessianbkm*

\begin{proof}
Let $K = K(\vec g) = -\beta H(\vec g)$. We abbreviate $Z = Z(\beta, \vec g)$ in this proof. First
compute the first derivative. By the Duhamel formula (\Cref{lem:duhamel}),
\begin{align*}
\partial_I Z
&=\Tr[\partial_I \mathrm{e}^{K}]
=\int_0^1\Tr\bigl[
\mathrm{e}^{(1-s)K}(\partial_I K)\mathrm{e}^{sK}
\bigr]\dd s=\int_0^1 \Tr[e^K \partial_I K] \dd s= \Tr[e^K \partial_I K]= -\beta\sigma_n \Tr[e^K \Gamma_I]\,.
\end{align*}
Therefore,
\begin{equation}
\partial_I\Phi
=\frac{\partial_IZ}{Z}
=-\beta\sigma_n\Tr[\rho_\beta(\vec g) \Gamma_I] = -\beta\sigma_n \langle \Gamma_I \rangle_\rho\,.
\label{eq:first-derivative}
\end{equation}
Now differentiate \Cref{eq:first-derivative} with respect to $g_J$. The quotient rule gives
\begin{equation}
\partial_J\langle \Gamma_I\rangle_\rho = \frac{\partial_J\Tr[\mathrm{e}^{K}\Gamma_I]}{Z}
-\frac{\Tr[\mathrm{e}^{K}\Gamma_I]}{Z}\frac{\partial_J Z}{Z}.
\label{eq:quotient}
\end{equation}
For the first term, apply Duhamel again:
\begin{align*}
\partial_J\Tr[\mathrm{e}^{K}\Gamma_I]
&=\int_0^1
\Tr\bigl[
\mathrm{e}^{(1-s)K}(\partial_JK)\mathrm{e}^{sK}\Gamma_I
\bigr]\dd s=-\beta\sigma_n\int_0^1
\Tr\bigl[
\mathrm{e}^{(1-s)K}\Gamma_J\mathrm{e}^{sK}\Gamma_I
\bigr]\dd s.
\end{align*}
Since \(\rho_\beta(\vec g)=e^K/Z\), we have
\(e^{aK}=Z^a\rho_\beta(\vec g)^{\,a}\).  Dividing by \(Z\) therefore gives
\begin{equation}
\frac{\partial_J\Tr[\mathrm{e}^{K}\Gamma_I]}{Z}
=-\beta\sigma_n\int_0^1
\Tr\bigl[\rho_\beta(\vec g)^{1-s}\Gamma_J\rho_\beta(\vec g)^s\Gamma_I\bigr]\dd s.
\label{eq:numerator-normalized}
\end{equation}
For the second term in \Cref{eq:quotient}, use
$\partial_JZ/Z=\partial_J\Phi=-\beta\sigma_n\langle \Gamma_J\rangle_{\rho}$, so
\begin{equation}
-\frac{\Tr[\mathrm{e}^{K}\Gamma_I]}{Z}\frac{\partial_JZ}{Z}
=\beta\sigma_n\langle \Gamma_I\rangle_{\rho}\langle \Gamma_J\rangle_{\rho}.
\label{eq:quotient-second}
\end{equation}
Inserting \Cref{eq:numerator-normalized} and \Cref{eq:quotient-second} into \Cref{eq:quotient},
\begin{equation}
\partial_J\langle \Gamma_I\rangle_{\rho}
=-\beta\sigma_n\int_0^1
\Tr\bigl[\rho_\beta(\vec g)^{1-s}\Gamma_J\rho_\beta(\vec g)^s\Gamma_I\bigr]\dd s
+\beta\sigma_n\langle \Gamma_I\rangle_{\rho}\langle \Gamma_J\rangle_{\rho}.
\label{eq:expectation-derivative}
\end{equation}
Since $\partial_I\Phi=-\beta\sigma_n\langle \Gamma_I\rangle_{\rho}$, we obtain
\begin{equation}
\partial_J\partial_I\Phi
=\beta^2\sigma_n^2\left[
\int_0^1
\Tr\bigl[\rho_\beta(\vec g)^{1-s}\Gamma_J\rho_\beta(\vec g)^s\Gamma_I\bigr]\dd s
-\langle \Gamma_I\rangle_{\rho}\langle \Gamma_J\rangle_{\rho}
\right].
\label{eq:hessian-one}
\end{equation}
Changing variables $s\mapsto1-s$, and using cyclicity of trace, this is the same as
\begin{equation}
\partial_I\partial_J\Phi
=\beta^2\sigma_n^2\left[
\int_0^1
\Tr\bigl[\rho_\beta(\vec g)^s\Gamma_I\rho_\beta(\vec g)^{1-s}\Gamma_J\bigr]\dd s
-\Tr[\rho_\beta(\vec g) \Gamma_I]\Tr[\rho_\beta(\vec g) \Gamma_J]
\right].
\label{eq:hessian-two}
\end{equation}
By the definition of the BKM covariance , the
quantity in brackets is
\(\Cov^{\mathrm{BKM}}_{\rho_\beta(\vec g)}(\Gamma_I,\Gamma_J)\),
which proves \Cref{eq:hessian-bkm}.  Summing against \(u_Iu_J\) then gives
\Cref{eq:hessian-bkm-direction}.
\end{proof}

\subsection{Sample-efficiency remaining proofs}
\label{alltemp:section}
\subsubsection{Tools}

\begin{lemma}[Derivatives of an unnormalized Gibbs trace]\label{alltemp:lem:derivatives}
Fix a deterministic matrix $W$ with $\|W\|_{\rm op}\le1$, and put $\mathcal Z_W(K)=\Tr(e^KW)$. For Hermitian $K,C_1,\ldots,C_r$ and every integer $r\ge0$,
\begin{equation}\label{alltemp:eq:trace-derivatives}
 \left|D^r\mathcal Z_W(K)[C_1,\ldots,C_r]\right|
 \le\Tr(e^K)\prod_{j=1}^r\|C_j\|_{\rm op}.
\end{equation}
In particular, for $f_W(\vec g)=\Tr(e^{-\beta H(\vec g)}W)$ and every ordered list $J_1,\ldots,J_r$,
\begin{equation}\label{alltemp:eq:derivatives}
 \left\lvert \partial_{g_{J_1}}\cdots\partial_{g_{J_r}}f_W(\vec g)\right\rvert
 \le (\beta\sigma_n)^r Z(\beta,\vec g).
\end{equation}
This holds pointwise for every real $\vec g$ and $\beta\ge0$.
\end{lemma}
\begin{proof}
For $r\ge1$, iterate the Duhamel formula of Lemma~\ref{lem:duhamel} to obtain
\[
 D^r\mathcal Z_W(K)[C_1,\ldots,C_r]
 =\sum_{\pi\in S_r}\int_{\Delta_r}
 \Tr\!\left(e^{t_0K}C_{\pi(1)}e^{t_1K}\cdots
 C_{\pi(r)}e^{t_rK}W\right)dt,
\]
where $\Delta_r=\{(t_1,\ldots,t_r):t_j\ge0,\ \sum_{j=1}^rt_j\le1\}$, $t_0=1-\sum_{j=1}^rt_j$, and $dt=dt_1\cdots dt_r$. Apply Lemma~\ref{alltemp:lem:schatten-holder} with exponent $1/t_j$ for $e^{t_jK}$ and exponent $\infty$ for the insertions and $W$. Since $\|e^{t_jK}\|_{1/t_j}=\Tr(e^K)^{t_j}$, each integrand is bounded in absolute value by $\Tr(e^K)\prod_j\|C_j\|_{\rm op}$; zero $t_j$ are handled by continuity. The $r!$ orderings cancel the simplex volume $1/r!$, proving \eqref{alltemp:eq:trace-derivatives}. The case $r=0$ follows directly from Lemma~\ref{alltemp:lem:schatten-holder}. Finally, substitute $K=-\beta H(\vec g)$ and $C_j=-\beta\sigma_n\Gamma_{J_j}$ to obtain \eqref{alltemp:eq:derivatives}.
\end{proof}

\begin{lemma}[Gaussian partial pairings]\label{alltemp:lem:wick}
Let $\vec\bsg=(\bsg_i)_{i=1}^d$ have independent standard Gaussian coordinates. Let $f:\mathbb R^d\to\mathbb C$ be smooth, with $f$ and its derivatives through order $N$ bounded at each $\vec g\in\mathbb R^d$ in absolute value by a polynomial times $e^{C\|\vec g\|_2}$ for some finite $C$. Write $\partial_i=\partial/\partial g_i$. For indices $i_1,\ldots,i_N$,
\begin{equation}\label{alltemp:eq:wick}
 \E_{\vec\bsg}\!\left[\prod_{a=1}^N \bsg_{i_a}f(\vec\bsg)\right]
 =\sum_{P}\left(\prod_{\{a,b\}\in P}\delta_{i_a,i_b}\right)
 \E_{\vec\bsg}\!\left[\left(\prod_{a\notin V(P)}\partial_{i_a}\right)f(\vec\bsg)\right].
\end{equation}
Here $P$ ranges over all partial matchings of $[N]$, including the empty matching, and $V(P)$ is the set of matched positions.
\end{lemma}
\begin{proof}
Let $\varphi_d$ denote the density of $\vec\bsg$, and set
\[
 a(t)=\sum_{r=1}^N t_r e_{i_r},\qquad
 Q(t)=\frac12\|a(t)\|_2^2
     =\frac12\sum_{r,s=1}^N t_rt_s\delta_{i_r,i_s},
 \qquad t\in\mathbb R^N.
\]
Completing the square gives
\[
 e^{a(t)\cdot x}\varphi_d(x)
 =e^{Q(t)}\varphi_d(x-a(t)).
\]
Integrating against $f(x)$ and changing variables therefore yields the Gaussian shift identity
\[
 \E_{\vec\bsg}\!\left[e^{a(t)\cdot\vec\bsg}f(\vec\bsg)\right]
 =e^{Q(t)}\E_{\vec\bsg}[f(\vec\bsg+a(t))].
\]

We justify the mixed differentiations under these two expectations explicitly. Since there are finitely many coordinate derivatives through order $N$, the growth assumption gives constants $A>0$, $C\ge0$, and an integer $m\ge0$ such that
\[
 |\partial^\alpha f(x)|\le A(1+\|x\|_2)^m e^{C\|x\|_2}
 \qquad (|\alpha|\le N).
\]
For $S\subseteq[N]$, write $D_S=\prod_{r\in S}\partial_{t_r}$ and $\partial_S=\prod_{r\in S}\partial_{x_{i_r}}$, with empty products acting as the identity. Restrict to $\|t\|_\infty\le1$, so that $\|a(t)\|_2\le N$, and put $\rho=\|x\|_2$. The integrand in the left-hand expectation is $e^{a(t)\cdot x}f(x)$. Its mixed derivatives satisfy
\begin{align*}
 \left|D_S\!\left(e^{a(t)\cdot x}f(x)\right)\right|
 &=\left|\prod_{r\in S}x_{i_r}\right|e^{a(t)\cdot x}|f(x)|\\
 &\le A(1+\rho)^{m+N}e^{(C+N)\rho},
\end{align*}
using $|x_{i_r}|\le\rho$ and $a(t)\cdot x\le N\rho$. The integrand in the right-hand expectation is $f(x+a(t))$. By the chain rule and $\|x+a(t)\|_2\le\rho+N$, its mixed derivatives satisfy
\begin{align*}
 |D_S f(x+a(t))|
 &=|(\partial_S f)(x+a(t))|\\
 &\le A(1+\rho+N)^m e^{C(\rho+N)}\\
 &\le A(1+N)^m e^{CN}(1+\rho)^m e^{C\rho}.
\end{align*}
Thus both families of differentiated integrands are bounded, uniformly in $S$ and $\|t\|_\infty\le1$, by
\[
 B(x)=K(1+\rho)^{m+N}e^{(C+N)\rho},
 \qquad K=A(1+N)^m e^{CN}.
\]
This bound is integrable against the Gaussian density: the inequality $(C+N)\rho\le\rho^2/4+(C+N)^2$ gives
\[
 B(x)\varphi_d(x)
 \le K(2\pi)^{-d/2}e^{(C+N)^2}
       (1+\rho)^{m+N}e^{-\rho^2/4},
\]
whose integral over $\mathbb R^d$ is finite. Dominated convergence, applied successively to the mixed derivatives, therefore permits all the differentiations below. The same bounds with $f\equiv1$ justify differentiating the Gaussian moment-generating function.

In particular, applying the shift identity with $f\equiv1$ and then Theorem~\ref{thm:wick} gives
\[
 D_S e^{Q(t)}\big|_{t=0}
 =\E_{\vec\bsg}\!\left[\prod_{r\in S}\bsg_{i_r}\right]
 =\sum_{P\in\PM(S)}\prod_{\{a,b\}\in P}\delta_{i_a,i_b},
\]
where $\PM(S)$ is the set of perfect matchings of $S$, empty when $|S|$ is odd and containing the empty matching when $S=\emptyset$. Applying $D_{[N]}$ to the shift identity at $t=0$ and using the product rule now yields
\begin{align*}
 \E_{\vec\bsg}\!\left[\prod_{r=1}^N \bsg_{i_r}f(\vec\bsg)\right]
 &=\sum_{S\subseteq[N]}
   \left(D_S e^{Q(t)}\big|_{t=0}\right)
   \E_{\vec\bsg}[(\partial_{[N]\setminus S}f)(\vec\bsg)]\\
 &=\sum_{S\subseteq[N]}\sum_{P\in\PM(S)}
   \left(\prod_{\{a,b\}\in P}\delta_{i_a,i_b}\right)
   \E_{\vec\bsg}[(\partial_{[N]\setminus S}f)(\vec\bsg)].
\end{align*}
Each partial matching $P$ of $[N]$ occurs exactly once in this sum, with $S=V(P)$, proving \eqref{alltemp:eq:wick}. The auxiliary variables label positions separately, so repeated indices have the correct multiplicities. The empty matching, and the case $N=0$, are included by the empty-product convention. Only finitely many real derivatives were used; no analyticity of $f$ is required.

For the Gibbs-trace weights used below, Lemma~\ref{alltemp:lem:derivatives} bounds every required derivative by a constant times $Z(\beta,\vec g)$. Since $H$ is linear in finitely many Gaussian coordinates, $Z(\beta,\vec g)\le e^{\beta\|H(\vec g)\|_{\rm op}}\le e^{C\|\vec g\|_2}$ for some finite $C$. Thus these applications satisfy the stated growth assumption.
\end{proof}

We also give the deferred proof of the Lipschitz property of the Hessian.

\hessianlipschitz*

\begin{proof}
For a real direction $\vec u$, Lemma~\ref{alltemp:lem:elementary} gives
\[
 \left\|\sigma_n\sum_Iu_I\Gamma_I\right\|_{\rm op}
 \le\frac{\sigma_n\sqrt{B_n}}2\|\vec u\|_2
 \le\frac12\|\vec u\|_2.
\]
Write $\mathcal Z(\vec h)=Z(\beta,\vec h)$. Apply Lemma~\ref{alltemp:lem:derivatives} with $W=\Id$, $K=-\beta H(\vec h)$, and insertions $C_j=-\beta\sigma_n\sum_I(\vec u_j)_I\Gamma_I$. For $r=1,2,3$ this gives
\[
 \frac{|D^r\mathcal Z(\vec h)[\vec u_1,\ldots,\vec u_r]|}{\mathcal Z(\vec h)}
 \le(\beta/2)^r\prod_{j=1}^r\|\vec u_j\|_2.
\]
Let $\mathcal Z_u=D\mathcal Z(\vec h)[\vec u]$, and use analogous subscripts for higher derivatives. The chain rule gives
\[
 D^3\Phi(\vec h)[\vec u,\vec v,\vec w]
 =\frac{\mathcal Z_{uvw}}{\mathcal Z}
 -\frac{\mathcal Z_{uv}\mathcal Z_w+\mathcal Z_{uw}\mathcal Z_v+\mathcal Z_{vw}\mathcal Z_u}{\mathcal Z^2}
 +2\frac{\mathcal Z_u\mathcal Z_v\mathcal Z_w}{\mathcal Z^3}.
\]
The preceding bounds and the triangle inequality therefore imply
\begin{equation}\label{alltemp:eq:third-derivative}
 |D^3\Phi(\vec h)[\vec u,\vec v,\vec w]|
 \le6(\beta/2)^3\|\vec u\|_2\|\vec v\|_2\|\vec w\|_2
 =L_\beta\|\vec u\|_2\|\vec v\|_2\|\vec w\|_2.
\end{equation}
For $\vec a=\vec h'-\vec h$ and unit vectors $\vec u,\vec v$, the fundamental theorem of calculus now gives
\[
 \left|\vec u^\top\bigl(\nabla^2\Phi(\vec h')-\nabla^2\Phi(\vec h)\bigr)\vec v\right|
 =\left|\int_0^1D^3\Phi(\vec h+t\vec a)[\vec u,\vec v,\vec a]dt\right|
 \le L_\beta\|\vec a\|_2.
\]
Taking the supremum over $\vec u,\vec v$ proves \eqref{alltemp:eq:lipschitz}. Finally,
\[
 \nabla^2\Phi(\vec h)\succeq\nabla^2\Phi(\vec g)-L_\beta\|\vec h-\vec g\|_2\Id
 \succeq\frac\lambda2\Id
\]
on the stated ball, proving \eqref{alltemp:eq:local-curvature}.
\end{proof}

\subsubsection{Proof of the local-quench estimate}

For $U=U_A$ define $C_k(\vec g)=\operatorname{ad}_{H(\vec g)}^k(U)$, where $\operatorname{ad}_H(B)=[H,B]$.
\begin{lemma}[Moment bound]\label{alltemp:lem:moments}
For $|A|\le4$, $\beta\ge0$, and every integer $k\ge1$, with $H=H(\vec\bsg)$ and $C_k=C_k(\vec\bsg)$,
\begin{equation}\label{alltemp:eq:moments}
 M_{2k}:=\frac{\E_{\vec\bsg}\Tr(e^{-\beta H}C_kC_k^\dagger)}{Z_\mathrm{ann}(\beta)}
 \le [8k(1+\beta)]^{2k}.
\end{equation}
\end{lemma}
\begin{proof}
Expand the $k$ nested commutators in their Gaussian coefficients. If $B$ is a Majorana support, then $[\Gamma_I,U_B]$ vanishes unless $|I\cap B|$ is odd. In the nonzero case it is twice a unitary Majorana monomial, up to phase, on support $I\triangle B$. Thus along a nonzero path the support increases by at most two per step. Before step $j\le k$ its size is at most
\[
 |A|+2(j-1)\le 2k+2\le4k.
\]
For any such support the number of possible next quartic labels is at most
\begin{equation}\label{alltemp:eq:Lk}
 L_k:=4k\binom{n-1}{3},\qquad \sigma_n^2 L_k\le4k.
\end{equation}
This counts all quartets meeting the current support and so is an upper bound on the nonzero choices.

The expansion of $C_kC_k^\dagger$ has two length-$k$ commutator paths. Each term is a phase times
\[
 4^k\sigma_n^{2k}\left(\prod_{a=1}^{2k}\bsg_{I_a}\right)W,
 \qquad \left\lVert W\right\rVert_{\rm op}=1,
\]
with both paths nonzero. The matrix $W$ depends on the labels but is independent of $\vec\bsg$.

Apply Lemma~\ref{alltemp:lem:wick} to each Gibbs trace. Fix a partial matching with $b$ pairs. There are $\ell=2k-2b$ unmatched positions and $r=2k-b$ label blocks: one per pair and one per unmatched position. Lemma~\ref{alltemp:lem:derivatives} bounds the remaining derivative by $(\beta\sigma_n)^\ell Z$.

There are at most $L_k^r$ label assignments compatible with this matching and two nonzero paths. To see this, write each path as $D_0=U$ and $D_j=[\Gamma_{I_j},D_{j-1}]$. Scan the labels $I_1,\ldots,I_k$ of the first path from the innermost commutator to the outermost, then do the same for the second path before taking its adjoint. At the first occurrence of a label block, the current support is already determined, and there are at most $L_k$ possible labels by \eqref{alltemp:eq:Lk}. A later occurrence of that block has no free choice. Extra coincidences between different blocks are permitted in this upper bound and cause no difficulty.

Consequently, after division by $Z_\mathrm{ann}(\beta)$, this matching contributes in absolute value at most
\begin{align*}
 4^k\sigma_n^{2k}(\beta\sigma_n)^\ell L_k^r
 &=4^k\beta^{2k-2b}(\sigma_n^2L_k)^{2k-b}\\
 &\le4^k\beta^{2k-2b}(4k)^{2k-b}.
\end{align*}
The equality uses $2k+\ell=2r$. This is the cancellation of every power of $n$: both a paired block and an unmatched block carry exactly one factor $\sigma_n^2$.

There are
\[
 N_{k,b}=\frac{(2k)!}{(2k-2b)!2^bb!}
 =\binom{2k}{2b}(2b-1)!!
 \le\binom{2k}{2b}(2k)^b
\]
such matchings, with the usual convention at $b=0$. Summing the preceding bounds gives
\begin{align*}
 M_{2k}
 &\le4^k\sum_{b=0}^k N_{k,b}\beta^{2k-2b}(4k)^{2k-b}\\
 &\le(8k)^{2k}\sum_{b=0}^k\binom{2k}{2b}\beta^{2k-2b}2^{-b}\\
 &\le[8k(\beta+2^{-1/2})]^{2k}
 \le[8k(1+\beta)]^{2k}.
\end{align*}
The penultimate inequality retains the even-indexed terms of the binomial expansion of $(\beta+2^{-1/2})^{2k}$. All inequalities hold also at $\beta=0$ by continuity.
\end{proof}

\begin{proof}[Proof of Theorem~\ref{alltemp:thm:quench}]
Fix $A$ and $\beta\ge0$, and write $U=U_A$. For each realization $\vec g$, choose an orthonormal eigenbasis $H(\vec g)|i\rangle=E_i|i\rangle$ and put $U_{ij}=\langle i|U|j\rangle$. Since $e^{sUHU^\dagger}=Ue^{sH}U^\dagger$, the normalized-trace convention gives
\begin{align*}
 Z_A(\beta,s,\vec g)
 &=\Tr\!\left(e^{-(\beta+s)H}Ue^{sH}U^\dagger\right)\\
 &=\frac1{d_n}\sum_i e^{-(\beta+s)E_i}
       \sum_j e^{sE_j}|U_{ij}|^2\\
 &=\frac1{d_n}\sum_{i,j}e^{-\beta E_i}|U_{ij}|^2
       e^{s(E_j-E_i)}.
\end{align*}
Set $C_0=U$ and recall $C_k=\operatorname{ad}_H^k(U)$. The identity
\[
 [H,X]_{ij}=(E_i-E_j)X_{ij}
\]
implies $(C_k)_{ij}=(E_i-E_j)^kU_{ij}$ for every integer $k\ge0$. Hence
\begin{align*}
 \Tr(e^{-\beta H}C_kC_k^\dagger)
 &=\frac1{d_n}\sum_i e^{-\beta E_i}\sum_j|(C_k)_{ij}|^2\\
 &=\frac1{d_n}\sum_{i,j}e^{-\beta E_i}|U_{ij}|^2(E_j-E_i)^{2k},
\end{align*}
where the even power removes the sign difference between $E_i-E_j$ and $E_j-E_i$. In particular, $C_0C_0^\dagger=\Id$, so the $k=0$ trace is $Z(\beta,\vec g)$.

For every real $s$, apply $e^x\le2\cosh x$ to the nonnegative spectral weights and expand the even series:
\begin{align*}
 Z_A(\beta,s,\vec g)
 &\le\frac2{d_n}\sum_{i,j}e^{-\beta E_i}|U_{ij}|^2
       \cosh\!\bigl(s(E_j-E_i)\bigr)\\
 &=\frac2{d_n}\sum_{i,j}e^{-\beta E_i}|U_{ij}|^2
       \sum_{k=0}^{\infty}\frac{s^{2k}(E_j-E_i)^{2k}}{(2k)!}\\
 &=2\sum_{k=0}^{\infty}\frac{s^{2k}}{(2k)!}
       \Tr(e^{-\beta H}C_kC_k^\dagger).
\end{align*}
Every term is nonnegative, so Tonelli's theorem permits exchanging the series with the expectation over $\vec\bsg$. Moreover, $0<Z_\mathrm{ann}(\beta)<\infty$, since
\[
 0<Z(\beta,\vec g)\le e^{\beta\|H(\vec g)\|_{\rm op}}
 \le e^{\beta\sigma_n\sqrt m\|\vec g\|_2},
\]
and the last expression is Gaussian-integrable. Dividing by $Z_\mathrm{ann}(\beta)$ and setting $M_0=1$, Lemma~\ref{alltemp:lem:moments} now gives
\begin{equation}\label{alltemp:eq:mgf}
 \begin{aligned}
 \frac{\E_{\vec\bsg}Z_A(\beta,s,\vec\bsg)}{Z_\mathrm{ann}(\beta)}
 &\le2\sum_{k=0}^{\infty}\frac{s^{2k}M_{2k}}{(2k)!}\\
 &\le2\sum_{k=0}^{\infty}[4e(1+\beta)|s|]^{2k}
 \le\frac83.
 \end{aligned}
\end{equation}
The second inequality uses $(2k)!\ge(2k/e)^{2k}$ for $k\ge1$. The last follows from \eqref{alltemp:eq:srange}, which makes the geometric-series ratio at most $1/4$.
\end{proof}

\section{Sample-efficient learning for bounded-degree Hamiltonians}
\label{app:bounded_degree}
We give a bounded-degree counterpart of the sample-efficient algorithm in \Cref{sampleall:thm:learning}. We follow
the same proof strategy: reduce learning to strong convexity, identify the Hessian with a BKM
variance, and lower bound the ordinary variance by detecting a large coefficient of the underlying
Hamiltonian. The first two steps use the same arguments as in the main text,
\begin{itemize}
	\item \Cref{lem:hessian-bkm,lem:bkm-variance} relate the Hessian and ordinary variance, and
	\item the proof of \Cref{sampleall:thm:error-from-convexity} obtains an error bound on the estimates given strong
		convexity.
\end{itemize}
Note that the second item above is slightly easier in this setting as strong convexity holds
throughout the parameter regime so an analogue of \Cref{alltemp:lem:hessian-lipschitz} is not necessary.

The largest difference is in the variance argument: bounded degree permits a one-site twirl in place
of the Majorana filter \Cref{sampleall:eq:filter}, and the local energy-transition bound of
\cite{arad2016connecting} in place of the SYK local-quench estimate in \Cref{alltemp:thm:quench}. We remark that this is highly
reminiscent of \cite{anshu2020sample}, with the exception that the Hessian can be directly related
to a local operator, rather than the quasi-local observable in their setting. This drastically
simplifies the proof.

The two variance proofs use the following ingredients:
\begin{itemize}
	\item In \Cref{sampleall:lem:variance}, we compare rotated second moments using
			Schatten H\"older and Araki--Lieb--Thirring. The required fractional moments are
			controlled using the local-quench estimate in \Cref{alltemp:thm:quench}.
	\item For this section, \cite{arad2016connecting} yields an energy-transition estimate in
			\emph{operator norm}; thus, we can directly bound the leakage under the true
			Gibbs state, as used in the variance argument below.
\end{itemize}

\paragraph{Setting.}
In this appendix, $n$ denotes the number of qubits, and $\mc P=\{P_1,\ldots,P_m\}$ is a known
collection of $m\geq1$ distinct, nonidentity Pauli strings. Each $P_a$ acts on at most $k$ qubits,
and each qubit participates non-trivially in at most $\Delta$ terms. In particular, $m\leq n\Delta$.
Write $d=2^n$ and retain the normalized-trace convention $\Tr=d^{-1}\operatorname{tr}_{\rm
ordinary}$. For $K=[-J,J]^m$, define
\[
  H(\vec g)=\sum_{a=1}^m g_aP_a,\qquad
  Z(\beta,\vec g)=\Tr[e^{-\beta H(\vec g)}],\qquad
  \rho_\beta(\vec g)=\frac{e^{-\beta H(\vec g)}}{Z(\beta,\vec g)},\qquad
  \Phi(\vec g)=\log Z(\beta,\vec g).
\]
The physical Gibbs state is $\rho_\beta(\vec g)/d$. Fix an inverse temperature $\beta > 0$ and a
per-term coefficient norm bound $J$; we assume each $|g_a| \leq J$.

\begin{theorem}[Bounded-degree strong convexity]
\label{thm:bd_strong_convexity}
There is a constant $c_{\beta,J,k,\Delta}>0$ such that uniformly for
$\vec g\in K = [-J, J]^m$,
\begin{equation}
  \label{eq:bd_curvature}
  \nabla^2 \Phi(\vec g) \succeq \lambda_m \Id, \qquad \lambda_m := \frac{c_{\beta, J, k,
  \Delta}}{m(1 + \log(2m))}\,.
\end{equation}
In particular, $\lambda_m \in \Omega_{\beta,J,k,\Delta}(1/(n\log(2n)))$. The bound holds at every
fixed inverse temperature $\beta>0$.
\end{theorem}

\subsection{Where bounded-degree enters}

The derivative identities and Hessian calculation in
\Cref{lem:first_derivative_identity,lem:hessian-bkm} only use the fact that $H$ is linear in the
individual Hermitian terms. Replacing the Majorana terms $\sigma_n\Gamma_I$ with the Pauli strings
$P_a$ gives,
\begin{equation}
  \nabla\Phi(\vec g)=-\beta\vec\mu_\beta,\qquad
  \vec u^\top\nabla^2\Phi(\vec g)\vec u
    =\beta^2\Var_{\rho_\beta(\vec g)}^{\BKM}(A_{\vec u}),
  \quad
  A_{\vec u}=\sum_a u_aP_a,
  \label{eq:bd_hessian}
\end{equation}
where $\mu_{\beta,a}=\Tr[\rho_\beta(\vec g)P_a]$.

The only additional many-body estimate we need is the following leakage bound.

\begin{theorem}[Operator norm leakage bound {\cite[Theorem~2.1]{arad2016connecting}}]
\label{thm:bd_akl}
Let $H=\sum_{\substack{X\subset\Lambda\\ |X| \leq k}}h_X$ be a $k$-local Hamiltonian, where $h_X$
are PSD terms acting non-trivially only on sites in $X$. Assume a bound $g$ on the interaction strength
of any physical site $i \in \Lambda$:
\[
	\sum_{X \ni i} \|h_X\|_\op \leq g\,.
\]
Given an operator $A$, let $\mc E_A$ be the set of supports such
that $[H,A]=\sum_{X\in\mc E_A}[h_X,A]$ and let $R:=\sum_{X\in\mc E_A}\|h_X\|_\op$. Then, for $E\geq0$ and $L\geq0$,
\[
  \big\|\Pi_{[E+L,\infty)}A\,\Pi_{[0,E]}\big\|_\op
  \leq\|A\|_\op e^{-\lambda(L-2R)}\,, \qquad \lambda = \frac 1 {2gk}\,.
\]
Here $\Pi_I$ denotes the projector onto eigenstates of $H$ with eigenvalues in $I$.
\end{theorem}

For an operator $A_X$ supported on $X$, the bound gives $R\leq g|X|$. The result also
extends to Hermitian $h_X$ by replacing each term $h_X$ with $h_X +
\|h_X\|_\op \Id$, at the cost of increasing the local strength to $2g$. In particular, there are
constants $C,c,c_0$, depending only on $k,g$, such that
\begin{equation}
	\label{eq:bd_akl}
	\Norm{\Pi_{[E+L,\infty)} A_X \Pi_{(-\infty,E]}}_\op \leq C\|A_X\|_\op e^{-c(L - c_0|X|)}
	\qquad( E \in \R,\ L \geq 0)\,.
\end{equation}
This uniform control of the sum of local interaction norms is precisely the resource that is
unavailable in dense SYK.

\subsection{A constant ordinary variance}

\begin{lemma}[Ordinary variance in a normalized direction]
\label{lem:bd_variance}
Let $\rho_\beta(\vec g)$ be the Gibbs state for $H(\vec g)$. For every $A=\sum_\ell a_\ell P_\ell$ with
$\|\vec a\|_\infty=1$ and every $\vec g\in K$,
\begin{equation}
  \Var_{\rho_\beta(\vec g)}(A)\geq c_{\rm var}>0.
  \label{eq:bd_variance}
\end{equation}
\end{lemma}

\begin{proof}
Write $\rho=\rho_\beta(\vec g)$, and center $A$ as $A_0=A-\Tr[\rho A]\Id$ so that $\sigma_A^2 =\Tr[\rho
A_0^2]$. It suffices to consider $ \sigma_A^2 \leq1$. We first control the second moment after a unitary
supported on at most $s_0=k\Delta$ qubits. All Hilbert--Schmidt norms below use the normalized
trace: $\|M\|_{2,\Tr}^2=\Tr[M^\dagger M]$.

\paragraph{Transporting a small spectral tail.}
Let $U_X$ be such a unitary, let $Q$ be any projector, and put $\delta=\Tr[\rho Q]$. Throughout,
$\Pi_j = \Pi_{(j, j+1]}$ is the unit-width spectral projector for $H(\vec g)$. Let $\Pi_{B_L(j)} =
\sum_{|\ell - j| \leq L} \Pi_\ell$. We compare Gibbs
weights on energy windows. We use a
``near--far'' decomposition then apply the leakage estimate given by \Cref{eq:bd_akl}. Let
$r_j = Z^{-1}e^{-\beta j}$.
\begin{align}
	\Tr[\rho U_X^\dagger Q U_X] &\leq 2\sum_j\Paren{\|Q \Pi_{B_L(j)} U_X \Pi_j \sqrt\rho\|_{2,\Tr}^2
	+\|Q(\Id-\Pi_{B_L(j)})U_X\Pi_j\sqrt\rho\|_{2,\Tr}^2}\notag\\
	&\leq 2\sum_j\|Q \Pi_{B_L(j)} U_X \Pi_j \sqrt\rho\|_{2,\Tr}^2 + C e^{-cL}\tag{$\|Q\|_\op
	\leq 1$ and \Cref{eq:bd_akl}}\\
	&\leq 2\sum_j \Tr[Q \Pi_{B_L(j)} U_X (\Pi_j \rho \Pi_j) U_X^\dagger \Pi_{B_L(j)} Q] + C
	e^{-cL}\notag\\
	&\leq 2\sum_j r_j \Tr[Q \Pi_{B_L(j)} U_X \Pi_j U_X^\dagger \Pi_{B_L(j)} Q] + C
	e^{-cL}\tag{$\Pi_j \rho \Pi_j \preceq r_j\Pi_j$}\\
	&\leq 2\sum_j r_j \Tr[Q \Pi_{B_L(j)} Q] + C e^{-cL}\tag{$U_X \Pi_j U_X^\dagger
	\preceq \Id$}\\
	&\leq 2\sum_j e^{\beta (L + 1)} \Tr[Q \rho \Pi_{B_L(j)} Q] + C e^{-cL}\tag{$\tfrac 1 Z
	e^{-\beta(j + L + 1)} \Pi_{B_L(j)} \preceq \rho \Pi_{B_L(j)}$}\\
	&\leq 2(2L+1)e^{\beta(L+1)} \delta + Ce^{-cL}.
  \label{eq:bd_mass_windows}
\end{align}
Choosing $L=\lceil\log(1/\delta)/(\beta+c)\rceil+1$ and absorbing the logarithmic factor into a
smaller exponent yields
\begin{equation}
  \Tr[\rho U_X^\dagger Q U_X]\leq C\delta^\alpha
  \quad\text{for some }0<\alpha\leq1.
  \label{eq:bd_mass_transport}
\end{equation}
This works for $\delta > 0$. For $\delta=0$, this same conclusion follows as both sides are $0$
since $\Tr[\rho Q] = 0 \iff Q= 0$. 

\paragraph{The rotated second moment.}
Fix $\gamma=1/8$, and take
$Q=\1_{\{|A_0|>\gamma\}}$, so $\delta\leq\sigma_A^2/\gamma^2$. If $\delta=0$, then $Q=0$ and the rotated second moment is at most $\gamma^2$. Hence assume $\delta>0$. Decompose the
spectrum of $A_0$ into
\[
  R_0=\Id-Q,\qquad
  R_s=\1_{\{s\gamma<|A_0|\leq(s+1)\gamma\}},\quad s\geq1.
\]
so that
\[
	A_0^2 \preceq \gamma^2 R_0 + \gamma^2 \sum_{s \geq 1} (s+1)^2 R_s\,,
\]
and
\begin{align*}
	\Tr[\rho U_X^\dagger A_0^2 U_X] &\leq \gamma^2 \Tr[\rho U_X^\dagger R_0 U_X] + \gamma^2 \sum_{s \geq
	1} (s + 1)^2 \Tr[\rho U_X^\dagger R_s U_X]\\
	&\leq \gamma^2 + \gamma^2 \sum_{s \geq 1} (s+1)^2 \|R_s U_X \sqrt \rho \|_{2, \Tr}^2\\
	\label{eq:bd_desired_bound}
	&\leq \gamma^2 + 2\gamma^2 \sum_{s \geq 1} \Brac{(s+1)^2 \|R_s U_X R_0 \sqrt \rho \|_{2, \Tr}^2 +
	(s+1)^2\|R_s U_X Q \sqrt \rho \|_{2, \Tr}^2}\,.\numberthis
\end{align*}
To bound this expression, apply \Cref{eq:bd_akl} with $A$ as the Hamiltonian. Its coefficients have magnitude at most one, so its local interaction strength is at most $\Delta$. Centering by a scalar does not change energy differences. Splitting each $R_s$ into its positive and negative spectral intervals and using the fixed value of $\gamma$ gives $\|R_s U_X R_t\|_\op\leq Ce^{-c|s-t|}$. Define
\[
	x_s=\|R_sU_X R_0 \sqrt\rho\|_{2,\Tr},\qquad y_s = \|R_s U_X Q \sqrt\rho\|_{2,\Tr},\qquad
	p_t = \Tr[\rho R_t]\,.
\]
We first consider the first summand, $\sum_{s \geq 1} (s+1)^2 x_s^2$. \Cref{eq:bd_akl} bounds each
individual $x_s \leq C e^{-c s}$ and the sum of $x_s^2$ by \Cref{eq:bd_mass_transport} as
\begin{align*}
	\sum_{s\geq1}x_s^2 &=\|QU_XR_0\sqrt\rho\|_{2,\Tr}^2\\
	&\leq2\Tr[\rho U_X^\dagger Q U_X]+2\|Q U_X Q \sqrt \rho \|_{2, \Tr}^2\\
  	&\leq2\Tr[\rho U_X^\dagger Q U_X]+2\|Q \sqrt \rho \|_{2, \Tr}^2\\
  	&\leq2\Tr[\rho U_X^\dagger Q U_X]+2 \delta\\
  	&\leq C\delta^\alpha\,.
\end{align*}
Pick a cutoff $s^\star$ and write,
\[
	\sum_{1 \leq s \leq s^\star} (s+1)^2 x_s^2 \leq (s^\star + 1)^2 \sum_{s \geq 1} x_s^2 \leq
	(s^\star + 1)^2 C \delta^\alpha\,.
\]
To bound the tail sum, we bound
\[
	\sum_{s > s^\star} (s+1)^2 x_s^2 \leq C \sum_{s > s^\star} (s+1)^2 e^{-2 c s} \leq C
	(s^\star+1)^2 e^{-2 c s^\star}\,.
\]
The last inequality follows from differentiating the standard geometric series $\sum_s r^s$ twice
with respect to $r$. Putting the above two bounds together and setting $s^\star = 1 + \ceil{(\alpha / 2c)
\log (1/\delta)}$ yields,
\begin{equation}
	\label{eq:xs_bound}
	\sum_{s \geq 1} (s+1)^2 x_s^2 \leq C( 1 + \log (1/\delta))^2 \delta^\alpha \leq C'
	\delta^{\alpha /2}\,.
\end{equation}

Next, we bound $\sum_{s \geq 1}(s+1)^2 y_s^2$. Write,
\[
	y_s = \|R_s U_X Q \sqrt \rho\|_{2, \Tr} \leq \sum_{t \geq 1} \|R_s U_X R_t \sqrt \rho\|_{2,
	\Tr} \leq \sum_{t \geq 1} \|R_s U_X R_t\|_\op \|R_t \sqrt \rho \|_{2, \Tr}\,.
\]
By definition, $\|R_t \sqrt \rho \|_{2, \Tr} \leq \sqrt{p_t}$. \Cref{eq:bd_akl} yields $\|R_s U_X
R_t\|_\op \leq C e^{-c|s -t|}$. Together, we get the bound $y_s \leq \sum_{t \geq 1} C e^{-c|s- t|}
 \sqrt{p_t}$.

 To bound the weighted square sum $\sum_s (s+1)^2y_s^2$, we write $(s+1)y_s$ as the $s$-th index of a
 convolution, then apply Young's convolution inequality. First, we have that $s+1 \leq (|s - t| +
 1)(t+1)$ and thus,
 \[
	 (s+1) y_s \leq \sum_t C (s+1) e^{-c |s - t|} \sqrt {p_t} \leq \sum_t \Paren{C (|s-t|+1)
	 e^{-c|s - t|}} (t+1) \sqrt{p_t} = (k \star b)_s\,.
 \]
 where
 \[
	 k_r = C(|r| + 1) e^{-c|r|}\,\qquad b_t = (t+1) \sqrt{p_t}\,.
 \]
 Then, $\sum_{s \geq 1} (s+1)^2 y_s^2 \leq \|k \star b\|^2_{\ell_2}$ and by Young's inequality, this
 is at most $\|k\|^2_{\ell_1} \|b\|_{\ell_2}^2$. $\|k\|_{\ell_1} = C \sum_r (|r|+1) e^{-c |r|} \leq
 C_k$ for a finite constant $C_k$ and thus,
 \[
	 \sum_{s\geq 1} (s+1)^2 y_s^2 \leq  \|k \star b\|^2_{\ell_2} \leq \|k\|^2_{\ell_1}
	 \|b\|_{\ell_2}^2 \leq C' \sum_{t \geq 1} (t+1)^2 p_t \leq 4 C' \sum_{t \geq 1} t^2 p_t\,.
 \]
 Now, notice that $A_0^2 \succeq \gamma^2 \sum_{t \geq 1} t^2 R_t$. Taking the trace with respect to
 $\rho$ on both sides upper bounds the RHS above by $4 C'\gamma^{-2}\Tr[\rho A_0^2] = C\sigma_A^2$, where the fixed factor $\gamma^{-2}$ is absorbed into $C$.
 We conclude that
 \begin{equation}
	 \label{eq:ys_bound}
	 \sum_{s \geq 1}(s+1)^2 y_s^2 \leq C \sigma_A^2\,.
 \end{equation}

 Finally, we plug in the bounds \Cref{eq:xs_bound} and \Cref{eq:ys_bound} into
 \Cref{eq:bd_desired_bound} to get $\Tr[\rho U_X^\dagger A_0^2 U_X] \leq \gamma^2 + C \gamma^2
 (\delta^{\alpha/2} + \sigma_A^2)$, where we have again increased $C$. Since $\delta \leq \sigma_A^2 /
 \gamma^2$, $\gamma^2 \delta^{\alpha/2} \leq \gamma^{2 - \alpha} \sigma_A^{\alpha}$. $\gamma = 1/8$ is
 a constant and can be absorbed into $C$. Lastly, $\sigma_A^2 \leq \sigma_A^\alpha$ since $\alpha \leq
 1$ and $\sigma_A^2 \leq 1$. We conclude that,
 \begin{equation}
	  \label{eq:bd_rotated_moment}
	 \Tr[\rho U_X^\dagger A_0^2 U_X] \leq \gamma^2 + C \sigma_A^\alpha\,.
 \end{equation}

\paragraph{Detecting a coefficient by a one-site twirl.}
Since $\|\vec a\|_\infty = 1$, there exists an index $\ell^*$ such that $|a_{\ell^*}|=1$. Pick a physical
site $i\in\supp(P_{\ell^*})$. Let 
$\mc T_i(M)= \E_{U \sim \mathrm{Haar}(\mathrm U(2))} [U_i^\dagger M U_i]$ be a single-qubit
twirl on site $i$. Then,
\[
	C_i = A - \mc T_i(A)  = A_0 - \mc T_i(A_0) = \sum_{\ell \st i \in \supp(P_\ell)} a_\ell
	P_\ell\,.
\]
Let $X_i = \bigcup_{\ell \st i \in \supp(P_\ell)} \supp(P_\ell)$ be the total support of operators
touching $i$. Bounded degree gives $|X_i|\leq k\Delta$, while Pauli orthogonality gives
$\Tr[C_i^2]=\sum_{\ell:\,i\in\supp(P_\ell)}a_\ell^2\geq1$. Averaging $U^\dagger C_i^2U$ over all
unitaries on $X_i$ gives $\Tr[C_i^2]\Id$, so some such $U$ satisfies $\Tr[\rho U^\dagger
C_i^2U]\geq1$. The squared triangle inequality and Jensen's inequality now give
\begin{align*}
  1 &\leq\Tr[\rho U^\dagger C_i^2U]\\
  &\leq2\Tr[\rho U^\dagger A_0^2U] + 2 \E_{V \sim \mathrm{Haar}(\mathrm U(2))}\Tr[\rho (V_i
  U)^\dagger A_0^2 (V_i U)]\\
  &\leq 4 \gamma^2 + 4C\sigma_A^\alpha \,.
\end{align*}
To get the final line, we use that $U$ and $V_i U$ are both supported on $X_i$ so
\Cref{eq:bd_rotated_moment} applies. Since $4\gamma^2=1/16$, this forces
\[
	\sigma_A^2 \geq \Paren{\frac {15}{64 C}}^{2/\alpha} =: c_{\rm var} > 0\,.
\]
\end{proof}

\subsection{From variance to learning}

\begin{proof}[Proof of \Cref{thm:bd_strong_convexity}]
We sharpen the use of \Cref{lem:bkm-variance} by discarding large energy
transitions before comparing the two variances. Fix
$A=\sum_\ell a_\ell P_\ell$ with $\|\vec a\|_\infty=1$, and center it as above.
Let $H|i\rangle=\lambda_i|i\rangle$ and
$q_i=e^{-\beta \lambda_i}/Z(\beta,\vec g)$, so $d^{-1}\sum_iq_i=1$.
Applying \Cref{eq:bd_akl} to each $P_\ell$, then using the triangle
inequality for the sum of $m$ terms, gives, for all $L\geq L_0$,
\[
  \sum_{j:\,|\lambda_j-\lambda_i|>L}|(A_0)_{ji}|^2\leq Cm^2e^{-cL}.
\]
The same bound holds when averaging over $i$,
\begin{equation}
  T_L:=\frac1d\sum_iq_i
       \sum_{j:\,|\lambda_j-\lambda_i|>L}|(A_0)_{ji}|^2
       \leq Cm^2e^{-cL}.
  \label{eq:bd_transition_tail}
\end{equation}
The termwise logarithmic-mean comparison in the proof of
\Cref{lem:bkm-variance}, restricted to $|\lambda_i-\lambda_j|\leq L$, gives directly
\begin{equation}
  \Var_\rho^{\BKM}(A)
  \geq\frac{\Var_\rho(A)-T_L}{1+\beta L}.
  \label{eq:bd_truncated_bkm}
\end{equation}
Choose
\[
	L_* = L_0 + c^{-1} \log\!\Paren{2+\frac{2Cm^2}{c_{\rm var}}} = \O(1 + \log(2m))\,.
\]
Then $T_{L_*}\leq c_{\rm var}/2$, and \Cref{lem:bd_variance} yields $\Var_\rho^{\BKM}(A)\geq
c'/(1+\log(2m))$. For an arbitrary $\vec u\neq\vec0$, apply this to $A=A_{\vec u}/\|\vec u\|_\infty$
and use \Cref{eq:bd_hessian}:
\[
  \vec u^\top\nabla^2\Phi(\vec g)\vec u
  \geq\frac{\beta^2c'}{1+\log(2m)}\|\vec u\|_\infty^2
  \geq\frac{\beta^2c'}{m(1+\log(2m))}\|\vec u\|_2^2.
\]
\end{proof}

\begin{cor}[Sample-efficient learning at fixed temperature]
\label{cor:bd_learning}
For every $\vec g\in K$ and $\eps,\zeta\in(0,1)$, the empirical
maximum-entropy dual estimator recovers $\vec g_{\rm meas}$ with
$\|\vec g_{\rm meas}-\vec g\|_2\leq\eps$, with probability at least
$1-\zeta$, using
\begin{equation}
  N=O_{\beta,J,k,\Delta}\!\left(
      \frac{m^4(1+\log(2m))^2}{\eps^2}
      \log\frac{2m}{\zeta}\right)
  \label{eq:bd_samples}
\end{equation}
copies of the physical Gibbs state $\rho_\beta(\vec g)/d$.
\end{cor}

\begin{proof}
Use the same estimator as in \Cref{eq:estimator}, with the present
parameterization and feasible set:
\[
  \vec g_{\rm meas}\in\mathop{\rm argmin}_{\vec h\in K}
  \left\{\Phi(\vec h)+\beta\vec h\cdot\vec\mu_{\beta,\rm meas}\right\}.
\]
We adapt the perturbation argument in the proof of \Cref{sampleall:thm:error-from-convexity}
to the present parameterization, with $\beta$ in place of $\beta\sigma_n$.
Here \Cref{eq:bd_curvature} holds throughout $K$, so there is no need for
\Cref{alltemp:lem:hessian-lipschitz}, a localization radius, or the second
accuracy requirement in \Cref{sampleall:eq:expectation-accuracy}. The true objective
has zero gradient at $\vec g$, and the measured objective differs by a
linear term. Integrating its Hessian along the feasible segment,
as in \Cref{sampleall:eq:localobjective}, and using minimality gives
\[
  0\geq\frac{\lambda_m}{2}\|\vec g_{\rm meas}-\vec g\|_2^2
  -\beta\|\vec\mu_{\beta,\rm meas}-\vec\mu_\beta\|_2
        \|\vec g_{\rm meas}-\vec g\|_2.
\]
It therefore suffices to estimate the moment vector to Euclidean accuracy
$\delta_\eps=\lambda_m\eps/(2\beta)$. The independent-batch Pauli
measurement argument in the proof of \Cref{lem:measurement_complexity} applies unchanged:
$\lceil(2m/\delta_\eps^2)\log(2m/\zeta)\rceil$ copies per observable
suffice. Hence
\[
  N\leq m+\frac{8\beta^2m^2}{\lambda_m^2\eps^2}
                 \log\frac{2m}{\zeta},
\]
which proves \Cref{eq:bd_samples}.
\end{proof}

\end{document}